\documentclass[10pt,reqno]{amsart}
\usepackage{latexsym,amsmath}
\usepackage{times}
\usepackage{amsfonts}
\usepackage{amssymb}
\usepackage{latexsym}

\newcommand\good{tame}
\newcommand\bad{wild}

\newcommand\gnp{G(n,p)}
\newcommand\gnm{G(n,m)}
\newcommand\gnd{G(n,d/n)}
\numberwithin{equation}{section}

\newcommand\bemph[1]{{\bf\em #1}}
\def\vec#1{\mathchoice{\mbox{\boldmath$\displaystyle#1$}}
{\mbox{\boldmath$\textstyle#1$}}
{\mbox{\boldmath$\scriptstyle#1$}}
{\mbox{\boldmath$\scriptscriptstyle#1$}}}

\newcommand{\Birk}{\mathcal D} 
\newcommand{\Zkc}{Z_{k}}
\newcommand{\Zkb}{Z_{k,\mathrm{bal}}}
\newcommand{\Zrb}{Z_{\rho,\mathrm{bal}}}
\newcommand{\Zrg}{Z_{\rho,\mathrm{\good}}}
\newcommand{\Zkg}{Z_{k,\mathrm{\good}}}
\newcommand{\Dg}{\mathcal D_{\mathrm{\good}}}
\newcommand{\Dgs}{\mathcal D_{s,\mathrm{\good}}}

\newcommand{\rhos}{\rho_{\mathrm{stable}}}
\newcommand{\rhoss}{\bar\rho_{s\mathrm{-stable}}}
\newcommand{\dc}{d_{k,\mathrm{cond}}}
\newcommand{\dRS}{\dc}
\newcommand{\dcond}{d_{k,\mathrm{cond}}}
\newcommand{\dAN}{d_{k,\mathrm{AN}}}
\newcommand{\dk}{d_{k-\mathrm{col}}}
\newcommand{\dfirst}{d_{k,\mathrm{first}}}

\newcommand\MPCPS{Mathematical Proceedings of the Cambridge Philosophical Society}
\newcommand\IPL{Information Processing Letters}
\newcommand\COMB{Combinatorica}

\DeclareMathOperator{\EX}{\mathbb E}
\DeclareMathOperator{\pr}{P}

\newtheorem{definition}{Definition}[section]
\newtheorem{claim}[definition]{Claim}

\newtheorem{theorem}[definition]{Theorem}
\newtheorem{lemma}[definition]{Lemma}
\newtheorem{proposition}[definition]{Proposition}
\newtheorem{corollary}[definition]{Corollary}

\newtheorem{fact}[definition]{Fact}

\usepackage[dvips]{epsfig}
\graphicspath{{./Fig_eps/}{./Eps/}}
\DeclareGraphicsExtensions{.ps,.eps}
\usepackage{color}
\usepackage{fullpage}

\newcommand\sign{\mathrm{sign}}

\newcommand\id{\mathrm{id}}

\newcommand\cA{\mathcal{A}}
\newcommand\cB{\mathcal{B}}
\newcommand\cC{\mathcal{C}}
\newcommand\cD{\mathcal{D}}
\newcommand\cF{\mathcal{F}}
\newcommand\cG{\mathcal{G}}

\newcommand\cH{\mathcal{H}}
\newcommand\cS{\mathcal{S}}

\newcommand\cI{\mathcal{I}}

\newcommand\cL{\mathcal{L}}

\newcommand\cX{\mathcal{X}}

\def\cR{{\mathcal R}}
\def\cC{{\mathcal C}}

\newcommand\eul{\mathrm{e}}
\newcommand\eps{\varepsilon}

\newcommand\Erw{\mathrm{E}}

\newcommand{\vecone}{\vec{1}}

\newcommand{\Bin}{{\rm Bin}}

\newcommand{\bink}[2] {{{#1}\choose {#2}}}

\newcommand\ra{\rightarrow}

\newcommand\bc[1]{\left({#1}\right)}
\newcommand\cbc[1]{\left\{{#1}\right\}}
\newcommand\bcfr[2]{\bc{\frac{#1}{#2}}}

\newcommand\brk[1]{\left\lbrack{#1}\right\rbrack}

\newcommand\norm[1]{\left\|{#1}\right\|}
\newcommand\abs[1]{\left|{#1}\right|}

\newcommand\RR{\mathbb{R}}
\newcommand\RRpos{\RR_{\geq0}}
\newcommand{\Whp}{W.h.p.}
\newcommand{\whp}{w.h.p.}

\newcommand{\Erdos}{Erd\H{o}s}
\newcommand{\Renyi}{R\'enyi}
\newcommand{\Lovasz}{Lov\'asz}

\newcommand{\Bollobas}{Bollob\'as}

\newcommand{\Luczak}{\L uczak}

\newcommand\Lem{Lemma}
\newcommand\Prop{Proposition}
\newcommand\Thm{Theorem}
\newcommand\Cor{Corollary}
\newcommand\Sec{Section}

\begin{document}

\title{Chasing the $k$-colorability threshold$^*$}

\author[Coja-Oghlan]{Amin Coja-Oghlan$^\dagger$ and Dan Vilenchik}
\thanks{$^*$An extended abstract version of this work appeared in the Proceedings of the 54th
	IEEE Symposium on Foundations of Computer Science (`FOCS'), 2013.}
\thanks{$^\dagger$The research leading to these results has received funding from the European Research Council under the European Union's Seventh Framework
			Programme (FP/2007-2013) / ERC Grant Agreement n.\ 278857--PTCC}

\address{Amin Coja-Oghlan, {\tt acoghlan@math.uni-frankfurt.de}, Goethe University, Mathematics Institute, 10 Robert Mayer St, Frankfurt 60325, Germany.}

\address{Dan Vilenchik, {\tt dan.vilenchik@weizmann.ac.il}, Facutly of Mathematics \&\ Computer Science, The Weizamnn Institute,  Rehovot, Israel.}

\maketitle

\begin{abstract}
\noindent
For a fixed number $d>0$ and $n$ large let $G(n,d/n)$ be the random graph on $n$ vertices in which any two vertices are connected with probability $d/n$ independently.
The problem of determining the chromatic number of $\gnd$ 
goes back to the famous 1960 article 
of \Erdos\ and \Renyi\ that started the theory of random graphs [Magayar Tud.\ Akad.\ Mat.\ Kutato Int.\ Kozl.\ {\bf 5} (1960) 17--61].
Progress culminated in the landmark paper of Achlioptas and Naor [Ann.\ Math.\ {\bf 162} (2005) 1333--1349],
in which they calculate the chromatic number precisely
for all $d$ in a set $S\subset(0,\infty)$ of asymptotic density $\lim_{z\ra\infty}\frac1z\int_0^z\vecone_S=\frac12$,
and up to an additive error of one for the remaining $d$.
Here we obtain a near-complete answer by
determining the chromatic number of $G(n,d/n)$ for
all $d$ in a set of asymptotic density $1$.

\end{abstract}

\section{Introduction}

\noindent
Let $\gnp$ denote the random graph on the vertex set $V=\cbc{1,\ldots,n}$ in which any two vertices are connected with probability $p\in[0,1]$ independently,
	known as the {\em \Erdos-\Renyi\ model}.%
		\footnote{Actually this model was  introduced by Gilbert~\cite{Gilbert}.
		In their seminal paper \Erdos\ and \Renyi\ consider a random graph $G(n,m)$ in which the number of edges is a fixed integer $m$~\cite{ER}.
		However, with $p=m/\bink n2$ both models are essentially equivalent~\cite{JLR}.}
We write $p=d/n$ 
and refer to $d$ as the {\em average degree}. 
As per common practice, we say that $\gnd$ has a property {\em with high probability} (`\whp') if the probability that the property holds converges to $1$ as $n\ra\infty$.
We recall that a graph $G$ is \emph{$k$-colorable} if it is possible to assign each vertex one of the colors $\cbc{1,\ldots,k}$
such that no edge connects two vertices of the same color.
Moreover, the \emph{chromatic number} $\chi(G)$ of a graph $G$ is the least integer $k$ such that $G$ is $k$-colorable.
Unless specified otherwise, we always consider $d,k$ fixed as $n\ra\infty$.

\subsection{Background and main results.}
The theory of random graphs was born with the famous 1960 article by \Erdos\ and \Renyi~\cite{ER}, and has grown since into a substantial area of research with hundreds, perhaps thousands of contributions dealing with the $\gnp$ model alone.
In their paper,  \Erdos\ and \Renyi~showed that the random graph $\gnp$ undergoes a percolation {\em phase transition} at  $p=1/n$,
	and phase transitions have been the guiding theme of the theory ever since.
In addition, \Erdos\ and \Renyi\ set the agenda for future research by posing a number of intriguing questions, all of which have been answered over the years
except for one: for a given $d>0$, what is the typical chromatic number of $G(n,d/n)$?

It is widely conjectured that for any number $k\geq3$ of colors there occurs a phase transition for $k$-colorability.
That is, there exists a number $\dk$ such that $G(n,d/n)$ is $k$-colorable \whp\ if $d<\dk$, whereas the random graph fails to be $k$-colorable \whp\ if $d>\dk$.
If true, this would imply that the likely value of the chromatic number, viewed as a function of $d$, is a step function
	that takes the value $k$ on the interval $d_{(k-1)-\mathrm{col}}<d<\dk$.

Towards this conjecture,  Achlioptas and Friedgut~\cite{AchFried} proved that for any fixed $k\geq3$
there exists a \emph{sharp threshold sequence} $\dk(n)$.
This sequence is such that for any $\eps>0$,
\begin{itemize}
\item if $p<(1-\eps)\dk(n)/n$, then $\gnp$  is $k$-colorable with probability tending to $1$ as $n\ra\infty$.
\item if $p>(1+\eps)\dk(n)/n$, then $\gnp$ fails to be $k$-colorable with probability tending to $1$ as $n\ra\infty$.
\end{itemize}
Whether the sequence $\dk(n)$ converges to an actual ``uniform'' threshold $\dk$ is a well-known open problem.

Yet~\cite{AchFried} is a pure existence result that does not provide any clue as to the location of $\dk$.
In a landmark paper Achlioptas and Naor~\cite{AchNaor} proved via the ``second moment method'' that 
	\begin{equation}\label{XeqAN}
	\liminf_{n\ra\infty}\dk(n)\geq\dAN=2(k-1)\ln(k-1)=2k\ln k-2\ln k-2+o_k(1).
	\end{equation}
Here and throughout,  $o_k(1)$ denotes a term that tends to zero in the limit of large $k$.
By comparison, a naive application of the union bound shows that
	\begin{equation}\label{XeqFirst}
	\limsup_{n\ra\infty}\dk(n)\leq\dfirst=2k\ln k-\ln k.
	\end{equation}
Recently~\cite{Covers}, a more sophisticated union bound argument was used to prove
	\begin{equation}\label{XeqFirst2}
	\limsup_{n\ra\infty}\dk(n)\leq\dfirst'=2k\ln k-\ln k-1+o_k(1).
	\end{equation}
Thus, the gap between the lower bound (\ref{XeqAN}) and the upper bound~(\ref{XeqFirst2}) on
$\dk(n)$ is about $\ln k+1$, an expression that {\em diverges} as $k$ gets large.
By improving the lower bound,
the following theorem 
reduces this gap to a small absolute constant 
of  $2\ln 2-1+o_k(1)\approx 0.39$. 

\begin{theorem}\label{XXThm_main}\label{Thm_main}
The $k$-colorability threshold satisfies
	\begin{equation}\label{Xeqmain}
	\liminf_{n\ra\infty}\dk(n)\geq\dRS-o_k(1),\quad\mbox{ with }\quad\dRS=2k\ln k-\ln k-2\ln2.
	\end{equation}
\end{theorem}

The bounds~(\ref{XeqAN}), (\ref{XeqFirst2}) yield an estimate of the chromatic number of $G(n,d/n)$.
Namely, (\ref{XeqAN}) implies that for $d<\dAN$, the random graph $G(n,d/n)$ is $k$-colorable \whp\
Moreover, (\ref{XeqFirst2}) shows that for $d>d_{k-1,\mathrm{first}}$, $G(n,d/n)$ fails to be $k-1$-colorable \whp\
Consequently, for all $d$ in the interval $(d_{k-1,\mathrm{first}}',\dAN)$ of length about $\ln k$, 
the chromatic number of $G(n,d/n)$ is precisely $k$ \whp\
However, for all $d$ in the subsequent interval $(\dAN,\dfirst')$ of length about $\ln k$, (\ref{XeqAN}), (\ref{XeqFirst2}) only imply that
the chromatic number is either $k$ or $k+1$ \whp\
Thus, (\ref{XeqAN}) and (\ref{XeqFirst2}) yield the typical value of $\chi(G(n,d/n))$ precisely for ``about half'' of all $d$.
Formally,
let us say that a (measurable) set $A\subset\RRpos$ has \bemph{asymptotic density $\alpha$} if
	$\lim_{z\ra\infty}\frac1z\int_{0}^z\vecone_A=\alpha$,
where $\vecone_A$ is the indicator of $A$.
Then the set on which (\ref{XeqAN}), (\ref{XeqFirst2}) determine $\chi(G(n,d/n))$ has asymptotic density $1/2$ \cite[\Thm~2]{AchNaor}.

\Thm~\ref{XXThm_main} enables us to pin the chromatic number down precisely on a set of asymptotic density $1$,
thereby obtaining a near-complete answer to the question of \Erdos\ and \Renyi.
More precisely, (\ref{XeqFirst}) and~(\ref{Xeqmain}) imply

\begin{theorem}\label{XCor_chromatic}
There exists a constant $k_0$ such that the following is true.
Let
	$$\textstyle
	S_k=(2(k-1)\ln(k-1)-\ln(k-1)-0.99,2k\ln k-\ln k-1.38)\quad \mbox{and}\quad S=\bigcup_{k\geq k_0}S_k.
	$$
Set $F(d)=k$ for all $d\in S_k$.
Then $S$ has asymptotic density $1$ and
	$$\lim_{n\ra\infty}\pr[\chi(G(n,d/n))=F(d)]=1\qquad\mbox{for any  $d\in S$}.$$
\end{theorem}

\noindent
Of course, the constants $0.99$ and $1.38$ in the definition of $S_k$ can be replaced by any numbers
less than one and $2\ln2$, respectively.
Theorem~\ref{XCor_chromatic} also answers a question of Alon and Krivelevich~\cite{AlonKriv}
whether the chromatic number of $G(n,d/n)$ is concentrated on a single integer for most $d$ ``in an appropriately defined sense''.\footnote{A proof that the threshold sequence $\dk(n)$ converges would imply
		a one-point concentration result for the chromatic number outside a countable set of average degrees.
		However, the known result~\cite{AchFried} does not.
		Alon and Krivelevich~\cite{AlonKriv} were concerned also with the case that the average degree $d$ is a growing function of $n$.
		In this paper we deal with $d$ fixed as $n\ra\infty$, the original setting considered by \Erdos\ and \Renyi.}

Independently of the mathematics literature, the random graph coloring problem has been studied in statistical physics,
where it is known as the ``diluted mean-field Potts antiferromagnet at zero temperature''.
In fact, physicists have developed a generic, ingenious but highly non-rigorous formalism called the ``cavity method''
for locating phase transitions in random graphs and other discrete structures~\cite{MM,MPZ}.
The so-called ``replica symmetric'' variant of the cavity method predicts upper and lower bounds on $\dk$~\cite{pnas,vanMourik}, namely
	\begin{equation}\label{eqRS}
	\dRS-o_k(1)\leq \liminf_{n\ra\infty}\dk(n)\leq\limsup_{n\ra\infty}\dk(n)\leq\dfirst.
	\end{equation}
\Thm~\ref{Thm_main} establishes the lower bound rigorously.

Additionally, the cavity method yields predictions on the combinatorial nature of the problem, particularly
on the geometry of the set of $k$-colorings of the random graph.
The proof of \Thm~\ref{Thm_main} is based on a ``physics-enhanced'' second moment argument that exploits this geometrical intuition.
In fact, the physics intuition is one of two key ingredients that enable us to improve  over the approach of Achlioptas and Naor~\cite{AchNaor}.
The second one is a novel approach, based on a local variations argument, to the analytical challenge of 
optimizing a certain (non-convex) function over the Birkhoff polytope. 
Neither of these ideas seem to depend on particular features of the graph coloring problem,
and thus we expect that they will prove vital to tackle a variety of further related problems.

An outline of our physics-enhanced second moment argument follows in \Sec~\ref{sec:outline}.
In addition, in \Sec~\ref{XSec_cond} we will see that the density $\dRS$ in~(\ref{Xeqmain}) matches
the {\em condensation} or {\em Kauzmann phase transition} predicted by physicists.
This implies that the bound obtained in \Thm~\ref{Thm_main} is the best possible one that can be
obtained via a second moment-type argument over a certain class of natural random variables
	(see \Sec~\ref{XSec_cond} for details).

\subsection{Related work}\label{XSec_related}
As witnessed by the notorious ``four color problem'' first posed by 
De Morgan in 1852,
solved controversially by Appel and Haken in 1976~\cite{AppelHaken77}, and re-solved by Robertson, Sanders, Seymour and Thomas~\cite{Robertson},
the graph coloring problem has been a central subject in (discrete) mathematics for well over a century.
Thus, it is unsurprising that the chromatic number problem on $\gnp$ has received a big deal of attention
since it was posed by \Erdos\ and \Renyi.
Indeed, the problem has inspired the development of  techniques that are by now widely used in various areas of mathematics, computer science, physics and other disciplines.

For instance, pioneering the use of martingale tail bounds,
Shamir and Spencer~\cite{ShamirSpencer} proved concentration bounds for the chromatic number of $\gnp$.
Their result was enhanced first by \Luczak~\cite{Luczak} and then
by Alon and Krivelevich~\cite{AlonKriv}, who used the \Lovasz\ Local Lemma to prove that the chromatic number
of $\gnp$ is concentrated on two consecutive integers if $p\ll n^{-1/2}$.
In a breakthrough contribution, \Bollobas~\cite{BBColor}
determined the asymptotics of the chromatic number of dense random graphs (i.e., $\gnp$ with $p>n^{-1/3}$).
This result improved prior work by Matula~\cite{Matula}, whose ``merge-and-exposure'' technique \Luczak\ built upon to obtain a similar
result for sparser random graphs~\cite{LuczakColor}.
However, in the case that $p=d/n$ for a fixed real $d>0$,
the setting originally studied by \Erdos\ and \Renyi,
 \Luczak's formula is far less precise than~(\ref{XeqAN})--(\ref{XeqFirst}).
For a comprehensive literature overview see
~\cite{BB,JLR}.

The work of Achlioptas and Naor~\cite{AchNaor}, which gave best prior result on the chromatic number of $G(n,d/n)$,
is based on the {\em second moment method}.
Its use in the context of phase transitions in random discrete structures was pioneered by Achlioptas and Moore~\cite{nae}
and Frieze and Wormald~\cite{FriezeWormald}.
The techniques of~\cite{AchNaor} have been used to prove several further important results.
For instance, Achlioptas and Moore~\cite{AMoColor} identified three (and for some $d$ just two) consecutive integers on which the
chromatic number of the random $d$-regular is concentrated.
This was reduced to two integers for all fixed 
 of $d$ (and one for about half of all $d$) by adding in the small subgraph conditioning technique~\cite{KPGW}.
Recently, the methods developed in this work have been harnessed to improve this result further still~\cite{Regular}. 
Moreover, Dyer, Frieze and Greenhill~\cite{DFG} extended the second moment argument from~\cite{AchNaor} to the
problem of $k$-coloring $h$-uniform random hypergraphs. 
We expect that our approach can be used to obtain improved results in the hypergraph case.
Similarly, it should be possible to improve results of
Dani, Moore and Olsen~\cite{Dani} on a  ``decorated'' coloring problem.

In several problems, sophisticated applications of the second moment method gave bounds
very close to the predictions made by the physicists' cavity method~\cite{MM}. 
Examples where the physics predictions have (largely) been verified rigorously in this way
include the hypergraph $2$-coloring problem~\cite{KostaNAE,Lenka} and the random $k$-SAT problem~\cite{KostaSAT}.
But thus far a general limitation of the rigorous proof techniques has been that they only apply to {\em binary} problems where there are only two values available for each variable.
By contrast, in random graph coloring each variable (vertex) has $k$ values (colors) to choose from, where $k$ can be arbitrarily large.
As we will see in \Sec~\ref{sec:outline}, the large number of available values complicates the problem dramatically.
In effect, random graph coloring remained the last among the intensely-studied benchmark problems in which there remained
a very substantial gap between the physics predictions and the rigorous results,
a situation rectified by the present paper.
Thus, we view this paper as an important step towards the long-term goal of providing a mathematical foundation for the cavity method.

In computer science,
 the \emph{algorithmic} problem of finding a $k$-coloring of $\gnp$ in polynomial time is a long-standing
challenge, mentioned prominently in several influential survey articles (e.g.,~\cite{FMcD,Kriv2}).
Simple greedy algorithms find a $k$-coloring for $d\leq k\ln k\sim\frac12\dk$ \whp~\cite{AchMolloy,GMcD,KSud},
about half the $k$-colorability threshold.
However, no efficient algorithm is known to beat the, in the words of Shamir and Spencer~\cite{ShamirSpencer}, ``most vexing'' factor of two.
In fact, it has been suggested 
changes in the geometry of the set of $k$-colorings
that occur at $d\sim\frac12\dk$
cause the demise of local-search based algorithms~\cite{Barriers,Molloy}.
Interestingly, some of the very phenomena that seem to make the algorithmic problem of coloring $\gnp$ difficult
will turn out to be extremely helpful in the construction of our random variable and thus in the proof of \Thm~\ref{Thm_main}.

\subsection{Notation and preliminaries.}
In addition to $\gnp$, we consider the $\gnm$ model,
	which is a random graph with vertex set $V=\cbc{1,\ldots,n}$ and exactly $m$ edges, chosen uniformly at random amongst all such graphs.
Working with $\gnm$ facilitates the second moment argument because the total number of edges is a deterministic quantity.
Nonetheless, \Lem~\ref{Lemma_Ehud} below shows that any results for $\gnm$ with $m=\lceil dn/2\rceil$ extend to $G(n,d/n)$.
{\bf\em Thus, throughout the paper we always set $m=\lceil dn/2\rceil$.}

Since our goal is to establish a statement that holds with probability tending to $1$ as $n\ra\infty$,
we are always going to assume tacitly that the number $n$ of vertices is sufficiently large for the various estimates to hold.
Similarly, at the expense of the error term $o_k(1)$ in \Thm~\ref{Thm_main} we will tacitly assume that $k\geq k_0$ for a large enough constant $k_0$.

We use the standard $O$-notation to refer to the limit $n\ra\infty$.
Thus, $f(n)=O(g(n))$ means that there exist $C>0$, $n_0>0$ such that for all $n>n_0$ we have $|f(n)|\leq C\cdot|g(n)|$.
In addition, we use the standard symbols $o(\cdot),\Omega(\cdot),\Theta(\cdot)$.
In particular, $o(1)$ stands for a term that tends to $0$ as $n\ra\infty$.
Furthermore, we write $f(n)\sim g(n)$ if $\lim_{n\ra\infty}f(n)/g(n)=1$.

Additionally, we use asymptotic notation in the limit of large $k$.
To make this explicit, we insert $k$ as an index.
Thus, $f(k)=O_k(g(k))$ means that there exist $C>0$, $k_0>0$ such that for all $k>k_0$ we have $|f(k)|\leq C\cdot|g(k)|$.
Further, we write $f(k)=\tilde O_k(g(k))$ to indicate that there exist $C>0$, $k_0>0$ such that
for all $k>k_0$ we have $|f(k)|\leq (\ln k)^C\cdot|g(k)|$.

If $G$ is a graph $v$ is a vertex of $G$, then we denote by $N_G(v)$ the neighborhood of $v$ in $G$, i.e., the set of all vertices $w$ that are connected
to $v$ by an edge of $G$.
Where the graph $G$ is apparent from the context we just write $N(v)$.
If $s\geq1$ is an integer, we write $\brk s$ for the set $\cbc{1,2,\ldots,s}$.
Moreover, throughout the paper we use the conventions that $0\ln 0=0$ and (consistently) that
$0\ln\frac 00=0$.

\section{Outline}\label{sec:outline}

\noindent
In this section we first discuss the second moment method in general and the argument pursued in~\cite{AchNaor} specifically and
investigate why it breaks down beyond the density $\dAN$ from~(\ref{XeqAN}).
Then, we see how the physics intuition can be harnessed to overcome this barrier.
Finally, we comment on the condensation phase transition.

\subsection{The second moment method.}\label{sec:TheSecondMomentMethod}
Suppose that
$Z=Z(\gnm)\geq0$ is a random variable such that
$Z(G)>0$ implies that $G$ is $k$-colorable.
Moreover, suppose that there is a number $C=C(d,k)>0$ that may depend on the average degree $d$ and
the number of colors $k$ but not on $n$ such that
	\begin{equation}\label{Xeqsmm}
	0<\Erw\brk{Z^2}\leq C\cdot\Erw\brk{Z}^2.
	\end{equation}
Then the \emph{Paley-Zygmund inequality}
	\begin{equation}\label{XeqPZ}
	\pr\brk{Z>0}\geq\frac{\Erw\brk{Z}^2}{\Erw\brk{Z^2}}
	\end{equation}
implies that
	$$\liminf_{n\ra\infty}\pr\brk{\gnm\mbox{ is $k$-colorable}}\geq\liminf_{n\ra\infty}\pr\brk{Z>0}\geq(4C)^{-1}>0.$$
This inequality yields a lower bound on the $k$-colorability threshold.

\begin{lemma}[\cite{AchFried}]\label{Lemma_Ehud}
If $d>0$ is such that 
	$\liminf_{n\ra\infty}\pr\brk{\gnm\mbox{ is $k$-colorable}}>0$, then
 $\liminf_{n\ra\infty}\dk(n)\geq d.$
\end{lemma}

\noindent
Thus, in order to obtain a lower bound on $\dk$, we
need to define an appropriate random variable $Z$ and verify~(\ref{Xeqsmm}).
Both of these steps turn out to be non-trivial.

\subsection{Balanced colorings and the Birkhoff polytope.}
The most obvious choice of random variable seems to be the total number $\Zkc$ of $k$-colorings of $\gnm$.
But to simplify the calculations, we
confine ourselves to a particular  type of colorings.
Namely, a map $\sigma:\brk n\ra\brk k$ is \bemph{balanced} if $||\sigma^{-1}(i)|-\frac nk|\leq\sqrt n$ for $i=1,\ldots,k$.
Let $\cB=\cB_{n,k}$ denote the set of all balanced maps.
Moreover, let $\Zkb$ be the number of balanced $k$-colorings of $\gnm$.
This is the random variable that Achlioptas and Naor~\cite{AchNaor} work with.
As it happens, (\ref{Xeqsmm}) does not hold for either $\Zkc$ or $\Zkb$ in the entire range $0<d<\dc$.
We need to understand why.

To get started, we compute the first moment.
By Stirling's formula the number of balanced maps is $\abs\cB=\Theta(k^{n})$.
Furthermore, for $\sigma$ to be a $k$-coloring, the random graph $\gnm$ must not contain any of the
	$$\cF(\sigma)=\sum_{i=1}^k\bink{|\sigma^{-1}(i)|}{2}$$
``forbidden'' edges that join two vertices with the same color under $\sigma$.
If $\sigma$ is balanced, we easily check that $\cF(\sigma)=(1-1/k)\bink n2+O(n)$.
Thus, letting $N=\bink n2$ and using Stirling's formula, we find that the probability that $\sigma$ is a $k$-coloring of $\gnm$ comes to
	$$\bink{N-\cF(\sigma)}m\big/\bink Nm=\Theta((1-1/k)^m).$$
Hence, by the linearity of expectation,
	\begin{eqnarray}\label{XeqFirstMoment}
	\Erw\brk{\Zkb}&=&\Theta(k^n(1-1/k)^{dn/2}).
	\end{eqnarray}

Working out the second moment is not quite so easy.
Since $\Erw[\Zkb^2]$ is the expected number of \emph{pairs} of balanced $k$-colorings, we need to compute the probability that
$\sigma,\tau\in\cB$ \emph{simultaneously} happen to be $k$-colorings of $\gnm$.
Of course, this probability depends on how ``similar'' $\sigma,\tau$ are.
To quantify this, we define the
$k\times k$ \bemph{overlap matrix} $\rho(\sigma,\tau)$ whose entries
\begin{equation}\label{eq:overlap_matrix}
	\rho_{ij}(\sigma,\tau)=\frac kn\cdot |\sigma^{-1}(i)\cap\tau^{-1}(j)|\qquad(i,j=1,\ldots,k)
\end{equation}
represent the proportion of vertices with color $i$ under $\sigma$ and color $j$ under $\tau$.

While in {\em binary} problems the relevant overlap parameter is just a $1$-dimensional
	(e.g.,  in random $k$-SAT, the Hamming distance of two truth assignments),
here the high-dimensional overlap matrix is required.
The need for this high-dimensional overlap parameter is what makes the $k$-colorability problem so difficult.

The upshot is that $\rho(\sigma,\tau)$ contains all the information  necessary to determine the probability that
both $\sigma,\tau$ are $k$-colorings.
In fact, let
	$\Zrb$ be the
number of pairs of balanced $k$-colorings with overlap $\rho$, and
let $\cR$ denote the set of all possible overlap matrices of maps $\sigma,\tau\in\cB$.
For a $k\times k$ matrix $\rho$ we denote the Frobenius norm by $$\norm\rho_2=\bigg(\sum_{i,j=1}^k\rho_{ij}^2\bigg)^{1/2}.$$

\begin{fact}[\cite{AchNaor}]\label{Lemma_frho}
Uniformly for $\rho\in\cR$ we have
	\begin{eqnarray}\label{Xeqffunction}
	\Erw\brk{\Zrb}&=&O(n^{(1-k^2)/2})\cdot\exp\brk{n\cdot f(\rho)},\qquad\mbox{ where}\\
	f(\rho)=f_{d,k}(\rho)&=&
	\ln k-\frac1k\bigg[\sum_{i,j=1}^k\rho_{ij}\ln\rho_{ij}\bigg]
	 	+\frac d2\ln\bigg[1-\frac2k+\frac1{k^2}\norm\rho_2^2
			\bigg].\nonumber
	\end{eqnarray}
\end{fact}
\begin{proof}
Since the function $f$ turns out to be the key object in this paper, we include the simple proof
to explain where it comes from combinatorially.
By Stirling's formula, the total number of $\sigma,\tau\in\cB$ with overlap $\rho$ equals
	\begin{equation}\label{eqLemma_frho1}
	\bink{n}{\rho_{11}\frac nk,\ldots,\rho_{kk}\frac nk}=O(n^{(1-k^2)/2})\cdot\exp\brk{-\sum_{i,j=1}^k n\cdot\frac{\rho_{ij}} k\ln\frac{\rho_{ij}}k}.
	\end{equation}
Now, suppose that $\sigma,\tau$ have overlap $\rho$.
By inclusion/exclusion, the number of ``forbidden'' edges  joining two vertices with the same color under either $\sigma$ or $\tau$ equals
	\begin{eqnarray*}\label{eqLemma_frho2}
	\cF(\sigma,\tau)&=&\sum_{i=1}^k\bink{\sum_{j}\rho_{ij}\frac nk}2+\sum_{j=1}^k\bink{\sum_{i}\rho_{ij}\frac nk}2-\sum_{i,j=1}^k\bink{\rho_{ij}\frac nk}2
		\geq2k\bink{n/k}2-\sum_{i,j=1}^k\bink{\rho_{ij}\frac nk}2.
	\end{eqnarray*}
Let $N=\bink n2$. Then Stirling's formula yields
	\begin{eqnarray}
	\pr\brk{\mbox{$\sigma,\tau$ are $k$-colorings of $\gnm$}}&=&\frac{\bink{N-\cF(\sigma,\tau)}{m}}{\bink{N}{m}}
		=O\bc{1}\cdot\exp\brk{m\bc{1-\frac 2k+\sum_{i,j=1}^k\bcfr{\rho_{ij}}k^2}}.
		\label{eqLemma_frho3}
	\end{eqnarray}
The assertion follows from~(\ref{eqLemma_frho1}), (\ref{eqLemma_frho3}) and the linearity of expectation.
\end{proof}

The bound~(\ref{Xeqffunction}) is essentially tight as similar calculations show that
	\begin{equation}\label{eqsecondMomentBasis1}
	\Erw\brk{\Zrb}=\exp(n\cdot f(\rho)+o(n)).
	\end{equation}
Moreover, by the linearity of expectation we can express the second moment as
	\begin{equation}\label{eqsecondMomentBasis}
	\Erw[\Zkb^2]=\sum_{\rho\in\cR}\Erw[\Zrb].
	\end{equation}
As the total number of summands is $\abs\cR\leq n^{k^2}$, we obtain from~(\ref{eqsecondMomentBasis1}) and~(\ref{eqsecondMomentBasis}) that
	\begin{equation}\label{eqsecondMomentBasis2}
	\frac1n\ln\Erw[\Zkb^2]\sim\max_{\rho\in\cR}\frac1n\ln\Erw[\Zrb]\sim\max_{\rho\in\cR} f(\rho).
	\end{equation}
Further, because we work with balanced colorings, the row and column sums of any $\rho\in\cR$ are $1+O(n^{-\frac12})$.
Thus, let $\Birk$ be the set of all doubly-stochastic $k\times k$ matrices, the \bemph{Birkhoff polytope}.
Together with the continuity of $f$ and the observation that $\cR\cap\cD$ becomes a dense subset of $\cD$ as $n\ra\infty$, (\ref{eqsecondMomentBasis2}) implies that
	\begin{equation}\label{XeqSecondLogscale}
	\frac1n\ln\Erw[\Zkb^2]\sim\max_{\rho\in\Birk} f(\rho).
	\end{equation}

In summary, following~\cite{AchNaor}, we have transformed the calculation of the second moment into the problem of optimizing  $f$ over the Birkhoff polytope $\Birk$.
Let $\bar\rho$ be the matrix  with all entries equal to $\frac1k$, the barycenter of $\Birk$. 
A  glimpse at 
	(\ref{XeqFirstMoment}) reveals that
$f(\bar\rho)\sim\frac2n\ln\Erw\brk{\Zkb}$ corresponds to the square of the first moment.
Therefore, a \emph{necessary} condition for the success of the second moment method 
is that the maximum~(\ref{XeqSecondLogscale}) is attained at~$\bar\rho$.
Indeed, if $f(\rho)>f(\bar\rho)$ for some $\rho\in\Birk$, then $\Erw[\Zkb^2]$ exceeds $\Erw[\Zkb]^2$
by an \emph{exponential} factor $\exp(\Omega(n))$. 
It is not difficult to show that this necessary condition is also sufficient.
Combinatorially, the condition that $\bar\rho$ is the maximizer of $f$ indicates that pairs
$\sigma,\tau$ that, judging by their overlap, look completely uncorrelated make up the lion's share of  $\Erw[\Zkb^2]$. 

\subsection{The singly-stochastic bound.}
Yet solving the optimization problem (\ref{XeqSecondLogscale}) proves seriously difficult.
Ach\-li\-optas and Naor resort to a relaxation: with $\cS\supset\cD$ the set of all $k\times k$ \emph{singly} stochastic matrices, they study
	\begin{equation}\label{eqSinglyStoch}
	\max_{\rho\in \cS}f(\rho).
	\end{equation}
Because $\cS$ is just a product of simplices, (\ref{eqSinglyStoch}) turns out to be much more amenable than~(\ref{XeqSecondLogscale}).
Achlioptas and Naor solve~(\ref{eqSinglyStoch}) completely.
More precisely, they optimize $f$ over the sets $\cbc{\rho\in\cS:\norm\rho_2=s}$ for each $s$, i.e., over the intersection of $\cS$ with a sphere.
Their argument relies on the product structure of $\cS$ and a sophisticated global analysis (going to the {\em sixth} derivative).
The result is that the maximum of~(\ref{eqSinglyStoch}) and therefore also of~(\ref{XeqSecondLogscale})
is attained at the doubly-stochastic $\bar\rho$ for  $d\leq \dAN$. 

However, for $d> \dAN$, the maximum~(\ref{eqSinglyStoch}) is attained elsewhere.
For instance, the matrix $\rho_{\mathrm{half}}$ whose first $k/2$ rows coincide with those of
the identity matrix $\id$ (with ones on the diagonal and zeros elsewhere) and whose last $k/2$ rows have all entries equal to $1/k$
yields a larger function value than $\bar\rho$ for $d>\dAN+o_k(1)$.
Of course, this matrix fails to be doubly-stoachastic.

Hence, one might hope that $\bar\rho$ remains the maximizer of~(\ref{XeqSecondLogscale})
for $d$ up to $\dc$. That is, however, not the case.
Indeed, consider the doubly-stochastic 
	\begin{equation}\label{Xeqrhostable}
	\rhos=(1-1/k)\id+k^{-2}\vecone,
	\end{equation}
where 
$\vecone$ denotes the matrix with all entries equal to one.
A simple calculation reveals that $f(\rhos)>f(\bar\rho)$, and thus that the second moment argument for $\Zkb$ fails,
 for $d$ well below $\dc$.

\subsection{A physics-enhanced random variable.}
Therefore, to prove \Thm~\ref{XXThm_main} we need to work with a different random variable.
The key observation behind its definition is that the second moment~(\ref{XeqSecondLogscale}) is driven up by
certain ``wild'' $k$-colorings $\sigma$.
Their number behaves like a lottery:
while the random graph typically has no wild coloring, a tiny fraction of graphs
have an abundance, boosting the second moment.
To avoid this heavily-tailed random variable, we define a notion of ``\good'' colorings.
This induces a decomposition
	$\Zkb=\Zkg+Z_{k,\mathrm{\bad}}$ such that
	$\Erw\brk{\Zkg}\sim\Erw\brk{\Zkb}$.
The second moment bound~(\ref{Xeqsmm}) turns out to hold for $\Zkg$ if $d\leq\dc-o_k(1)$.

The notion of ``\good'' is inspired by statistical physics predictions on the geometry of the set of $k$-colorings.
More precisely, according to the physicists' cavity method~\cite{pnas,LenkaFlorent}, for $(1+o_k(1))k\ln k<d<\dc$ the set of all $k$-colorings, viewed as a subset
of $\brk k^n$, decomposes into ``tiny clusters'' that are ``well-separated'' from each other.
Formally, we define the \bemph{cluster} of a balanced $k$-coloring $\sigma$ of $\gnm$ as the set
	\begin{equation}\label{XeqCluster}
	\cC(\sigma)=\cbc{\tau\in\cB:
		\tau\mbox{ is a $k$-coloring and $\rho_{ii}(\sigma,\tau)>0.51$ for all $i\in\brk k$}}.
	\end{equation}
In words, $\cC(\sigma)$ contains all balanced $k$-colorings $\tau$ where more than $51\%$ of the vertices
in each color class of $\sigma$ retain their color.
According to the cavity method, for $d<\dc$ each cluster contains only an exponentially small
fraction of all $k$-colorings of $\gnm$ \whp\
But for our purposes it suffices to formalize ``tiny'' by just requiring that
$|\cC(\sigma)|\leq\Erw\brk\Zkc$.

Futher, to formalize the notion that the clusters are ``well-separated'',
we call a balanced $k$-coloring $\sigma$ \bemph{separable} if
	\begin{equation}\label{eqseparable}
	\parbox{12cm}{for any other balanced $k$-coloring $\tau$
	and any $i,j\in\brk k$ such that $\rho_{ij}(\sigma,\tau)>0.51$ we indeed have $\rho_{ij}(\sigma,\tau)\geq1-\kappa$,
	where $\kappa=\ln^{20}k/k$.}
	\end{equation}
In other words, the overlap matrix $\rho(\sigma,\tau)$ does not have entries in the interval $(0.51,1-\kappa)$.
Hence, if two color classes have an overlap of more than $51\%$, then they must, in fact, be nearly identical.
This definition ensures that the clusters of two separable colorings $\sigma,\tau$ are either disjoint or identical.
We thus arrive at the following definition.

\begin{definition}\label{Def_good}
Let $G$ be a graph with $n$ vertices and $m$ edges.
A $k$-coloring $\sigma$ of $G$ is \bemph{\good} if 
\begin{description}
\item[T1] $\sigma$ is balanced,
\item[T2] $\sigma$ is separable, and
\item[T3]  $\abs{\cC\bc\sigma}\leq\Erw\brk{\Zkc(\gnm)}$.
\end{description}
\end{definition}

In \Sec~\ref{Sec_first} we show that a typical $k$-coloring of $\gnm$ is indeed tame, which implies that the expected number of \good\ $k$-colorings satisfies
the following.

\begin{proposition}\label{Prop_first}
There exists a sequence $\eps_k\ra0$ such that
for $d=\dRS-\eps_k$ we have
	$$\textstyle\Erw\brk{\Zkg}\sim\Erw\brk\Zkb=
			\Theta(\exp(\frac n2\cdot f(\bar\rho)))\quad\mbox{and }\quad
				f(\bar\rho)=\frac{2\ln 2}k+o_k(k^{-1})>0.$$
\end{proposition}
\noindent
Thus, going from blanaced to tame colorings has no discernible effect on the first moment,
	which remains exponentially large in $n$ up to at least $d= \dRS-\eps_k$.

Working with tame colorings has a  substantial impact on the second moment.
As before, computing the second moment boils down to a continuous optimization problem.
But  in comparison to~(\ref{XeqSecondLogscale}), this problem is over a {\em significantly} reduced domain $\Dg\subset\Birk$.
Indeed, let us call a $k\times k$-matrix $\rho$ \bemph{separable} if $\rho_{ij}\not\in\bc{0.51,1-\kappa}$ for all $i,j\in\brk k$.
Further, call $\rho$ {\bf\em $k$-stable} if for any $i$ there is $j$ such that $\rho_{ij}>0.51$.
Let $\Dg$ be the set of all $\rho\in\Birk$ that are separable but not $k$-stable.
In particular, the matrix $\rhos$ from~(\ref{Xeqrhostable}) does {\em not} belong to $\Dg$.
Geometrically, one can think of $\Dg$ as being obtained by cutting out (huge) cylinders from the Birkhoff polytope.
In \Sec~\ref{Sec_second} we will see that the second moment calculation for $\Zkg$ boils down to showing that
	\begin{equation}\label{eqOptDgood}
	\max_{\rho\in\Dg}f(\rho)
	\end{equation}
is attained at $\bar\rho$.
Indeed, that~(\ref{eqOptDgood}) mirrors the second moment calculation seems reasonable:
for any two \good\ colorings $\sigma,\tau$ the overlap matrix $\rho(\sigma,\tau)$ is separable by {\bf T2}.
Moreover, if $\rho(\sigma,\tau)$ is $k$-stable, then $\tau\in\cC(\sigma)$ by the very definition of $\cC(\sigma)$, and {\bf T3}
provides an {\em a priori} bound on the number of such $\tau$.

Thus, in a sense the proof strategy that we pursue is the opposite of the one from~\cite{AchNaor}.
While Achlioptas and Naor {\em relax} the optimization problem (by working with a rather significantly  larger domain: singly rather than doubly-stochastic matrices),
here we {\em restrict} the domain by imposing further physics-inspired constraints. 
This approach, carried out in \Sec~\ref{Sec_second}, yields

\begin{proposition}\label{Prop_second}
Assume that $k$ is sufficiently large and that $d=(2k-1)\ln k-c$ for some number $c=O_k(1)$.
If $\Erw[\Zkg]=\Omega(\Erw[\Zkb])$, then
	$0<\Erw[\Zkg^2]\leq C(k)\cdot\Erw\brk\Zkg^2.$
\end{proposition}

The proof of \Prop~\ref{Prop_second} essentially comes down to showing that the maximum~(\ref{eqOptDgood}) is attained at $\bar\rho$.
Even though we work with the reduced domain $\Dg$, this is anything but straightforward.
Indeed, to solve this analytical problem, we develop a novel local variations argument based on properties of the entropy function (among other things).
We expect that this argument will prove useful to tackle many related optimisation problems that come up in second moment arguments.

Finally, \Thm~\ref{XXThm_main} is an immediate consequence of \Prop s~\ref{Prop_first} and~\ref{Prop_second} combined with \Lem~\ref{Lemma_Ehud}.

\subsection{The condensation phase transition}\label{XSec_cond}
Finally, what would it take to close the (small) remaining gap between the new lower bound~(\ref{Xeqmain}) on $\dk$ and the upper bound~(\ref{XeqFirst2})?
According to the physicists' cavity method, this gap is due to a further phase transition, the so-called
	{\em condensation} or {\em Kauzmann transition}, that occurs at $\dRS+o_k(1)$, i.e.,
the lower bound established in \Thm~\ref{Thm_main}.
In fact, the existence and precise location  of this phase transition (including the term hidden in the $o_k(1)$) can be established rigorously~\cite{Condensation}.

According to the cavity method~\cite{pnas}, the geometry of the set of $k$-colorings changes significantly at $\dcond$.
More precisely, for $d<\dcond-o_k(1)$ the set of $k$-colorings decomposes into clusters that each contain only an exponentially small fraction of all $k$-colorings
of $G(n,d/n)$ \whp\
By contrast, for $d>\dcond+o_k(1)$, the size of the largest cluster is conjectured to contain a {\em constant} fraction of all $k$-colorings.
As a result, two random $k$-colorings are heavily correlated, as there is a non-vanishing probability that they belong to the same cluster.
This explains intuitively why the condensation threshold poses an obstacle to the second moment method,
as we saw that a necessary condition for the success of the second moment method is that random pairs of $k$-colorings decorrelate.

More formally, we prove in~\cite{Condensation} that for $d>\dRS+o_k(1)$ there does not exist a random variable $Z=Z(\gnm)$ with the following properties.
First, $Z(G)>0$ only if $G$ is $k$-colorable.
Second,
	$$\Erw[Z(\gnm)]^{1/n}\sim k(1-1/k)^{d/2}\quad\mbox{and}\quad\Erw[Z(\gnm)^2]\leq O(\Erw[Z(\gnm)]^2).$$
By contrast, \Prop s~\ref{Prop_first} and~\ref{Prop_second} show that $\Zkg$ has these two properties if $d<\dRS-o_k(1)$.
Hence, in this sense the approach (and random variable) put forward in the present paper is best possible.

A refined version of the cavity method, the so-called {\em 1-step replica symmetry breaking (``1RSB'') ansatz}~\cite{pnas,KPW,MPWZ,LenkaFlorent}, yields
a precise prediction as to the value of $\dk=\lim_{n\ra\infty}\dk(n)$ (of course, the existence of the limit is taken for granted in the physics work).
However, this prediction is not explicit; for instance, it involves the solution to a seriously complicated fixed point problem on the set of probability distributions on the $k+1$-simplex.
Yet it is possible to obtain an expansion in the limit of large $k$, according to which
	$\dk=2k\ln k-\ln k-1+o_k(1)$.
Proving the 1RSB prediction for $\dk$ remains an open problem.
In a very few binary problems,  asymptotic versions of the 1RSB prediction have been proved rigorously (e.g.,~\cite{KostaNAE}).
However, it seems anything but straightforward to extend these arguments to the random graph coloring problem.
That said, we expect that any 
attempt at determining $\dk$ precisely 
would have to build upon the insights gained in this paper and very possibly its techniques.

\section{The first moment}\label{Sec_first}

\noindent
{\em Throughout this section we keep the assumptions of \Prop~\ref{Prop_first} and the notation introduced in \Sec~\ref{sec:outline}.}

\smallskip\noindent
The following lemma is the key step towards proving \Prop~\ref{Prop_first}.

\begin{lemma}\label{Lemma_first}
There exists a sequence $\eps_k\ra0$ such that for $d=\dRS-\eps_k$ we have
	\begin{align*}
	\pr\brk{\sigma\mbox{ is \good}|\sigma\mbox{ is a $k$-coloring of $\gnm$}}&\sim1\ 
	\mbox{ for any }\sigma\in\cB\qquad\mbox{ and}\\
	f(\bar\rho)=2\ln k+d\ln(1-1/k)&=\frac{2\ln 2}k+o_k(k^{-1})>0.
	\end{align*}
\end{lemma}

\noindent
In fact, once we have \Lem~\ref{Lemma_first}, \Prop~\ref{Prop_first} readily follows from the linearity of expectation, Bayes' formula and the formula~(\ref{XeqFirstMoment})
for $\Erw[\Zkb]$.

To establish \Lem~\ref{Lemma_first}, we denote by $G(n,m,\sigma)$ the random graph $\gnm$ conditional on the event that $\sigma\in\cB$ is a $k$-coloring.
Thus, $G(n,m,\sigma)$ consists of $m$ edges drawn uniformly at random without replacement out of those edges that are bichromatic under $\sigma$.
This probability distribution is also known as the ``planted model''.

To establish the bound {\bf T3} on the cluster size,
we show that \whp\ $G(n,m,\sigma)$ contains a vast ``core''
comprising of vertices that have several neighbors  of each color other than their own that also belong to the core.
Formally, if $G=(V,E)$ is a graph on the vertex set $V=\cbc{1,\ldots,n}$ and $\sigma\in\cB$,
	we define the {\bf\em core} of $(G,\sigma)$ as the largest subset $V'\subset V$
	such that
		\begin{equation}\label{eqCoreProp}
		\abs{\cbc{w\in N(v)\cap V':\sigma(w)=i}}\geq100\quad\mbox{ for all $v\in V'$ and all $i\neq\sigma(v)$}.
		\end{equation}
The core is well-defined: if $V',V''$ satisfy~(\ref{eqCoreProp}), then so does $V'\cup V''$.
(Of course, the constant $100$ is a bit arbitrary.) 

As we will see, due to expansion properties no vertex in the core of $G(n,m,\sigma)$ can be recolored without leaving the cluster $\cC(\sigma)$ \whp\
The basic reason is that recoloring any vertex $v$ in the core sets off an avalanche of recolorings: to give $v$
	another color, we will have to recolor at least 100 vertices that also belong to the core, and so on.

In addition, if a vertex $v$ outside the core is such that for each color other than its own, $v$ has a neighbor in the core of that color, then it should be impossible
to recolor $v$ without leaving $\cC(\sigma)$ as well.
For to assign $v$ some color $i\neq\sigma(v)$ we will have to recolor at least one vertex in the core.
Guided by this observation, we call a vertex $v$ \bemph{$\sigma$-complete}, if for each color $i\neq\sigma(v)$, $v$ has a neighbor $w$ in the core with $\sigma(v)=i$.

If $\sigma$-complete vertices do not contribute to $|\cC(\sigma)|$,
then the cluster size stems from recoloring vertices $v$ that fail to have a neighbor in the core of some color $i\neq\sigma(v)$.
As we shall see, most of these vertices miss out on exactly one color $i\neq\sigma(v)$ and hence have precisely two colors to choose from.
Formally, we call a vertex $v$ \bemph{$a$-free} in $(G,\sigma)$ if, with $V'$ denoting the core, we have
		$\abs{\cbc{i\in\brk k:N(u)\cap V'\cap\sigma^{-1}(i)=\emptyset}}\geq a+1.$

The following lemma summarizes the expansion properties of $G(n,m,\sigma)$ that the proof of \Lem~\ref{Lemma_first} builds upon.

\begin{lemma}\label{Lemma_P}
Let $\sigma\in\cB$ and assume that $2k\ln k-\ln k-2\leq d\leq2k\ln k$.
Let $V_i=\sigma^{-1}(i)$ for $i=1,\ldots,k$.
Then \whp\ the random graph $G(n,m,\sigma)$ has the following four properties.
\begin{description}
  \item [P1] Let $i\in\brk k$. 
  	For any subset $S\subset V_i$ of size $0.509\cdot\frac nk\leq|S|\leq(1-k^{-0.499})\frac nk$, the number of vertices $v\in V\setminus V_i$ that do not have a neighbor in $S$
	is less than $\frac nk-|S|-n^{2/3}$.
  \item[P2] Let $i\in\brk k$. No more than $\frac{\kappa n}{3k}$ vertices $v\not\in V_i$ have less than $15$ neighbors in $V_i$, where $\kappa=\ln^{20}k/k$.
  \item[P3] There is no set $S\subset V$ of size $|S|\leq k^{-4/3}n$ that spans more than $5|S|$ edges.
 \item[P4]  At most $\frac nk(1+\tilde O_k(1/k))$ vertices are $1$-free, and at most $\tilde O_k(k^{-2})n$ vertices are $2$-free.
\end{description}
\end{lemma}

\noindent
The proof of \Lem~\ref{Lemma_P} is based on  arguments that are, by now, fairly standard;
in particular, the ``core'' has, tweaked in various ways, become a standard tool~\cite{Barriers,AK,BollobasBCKW01,Molloy}.
For the sake of completeness, we give a full proof of \Lem~\ref{Lemma_P} in Appendix \ref{apx:PropProof}.
Here we proceed to show how \Lem~\ref{Lemma_P} implies \Lem~\ref{Lemma_first}.

\begin{lemma}\label{Lemma_sep}
Assume that $2k\ln k-\ln k-2\leq d\leq2k\ln k$ and let $\sigma\in\cB$.
Then $\sigma$ is separable in $G(n,m,\sigma)$ \whp
\end{lemma}
\begin{proof}
By \Lem~\ref{Lemma_P} we may assume that the random graph $G(n,m,\sigma)$ has the properties {\bf P1}--{\bf P3}.
Suppose that $\tau\in\cB$ is another $k$-coloring of this random graph and that $i,j\in\brk k$ are such that
	$\rho_{ij}(\sigma,\tau)\geq0.51$.
Our aim is to show that $\rho_{ij}(\sigma,\tau)>1-\kappa$.
Without loss of generality we may assume that $i=j=1$.

Let $R=\sigma^{-1}(1)\setminus\tau^{-1}(1)$, $S=\tau^{-1}(1)\cap \sigma^{-1}(1)$ and $T=\tau^{-1}(1)\setminus \sigma^{-1}(1)$.
Because $\tau$ is a $k$-coloring, none of the vertices in $T$ has a neighbor in $S$.
Furthermore, because $\tau$ is balanced we have $|S\cup T|\geq\frac nk-\sqrt n$, and thus $|T|\geq\frac nk-|S|-\sqrt n$.
Since $|S|=\frac nk\rho_{11}(\sigma,\tau)>0.509\frac nk$, {\bf P1} implies that
	\begin{equation}\label{eqLemma_sep1}
	|S|\geq(1-k^{-0.49})\frac nk.
	\end{equation}

Now, let $U$ be the set of all $v\in T$ that have at least $15$ neighbors in $\sigma^{-1}(1)$.
Then all of these neighbors lie in $R$, because $\tau$ is a $k$-coloring.
Further, as $\sigma,\tau$ are asymptotically balanced we obtain from~(\ref{eqLemma_sep1})
	\begin{eqnarray*}
	|R\cup U|\leq |\sigma^{-1}(1)|-|S|+|T|\leq2\bc{\frac n{k^{1.49}}+\sqrt n}\leq n/k^{4/3}.
	\end{eqnarray*}
Hence, {\bf P3} applies to $R\cup U$.
By the definition of $U$ and~\textbf{P3}, the number $e(R\cup U)$ of edges spanned by $R\cup U$ satisfies
	\begin{equation}\label{eqNowhereDense2}
	15|U|\leq e(R\cup U)\leq5|R\cup U|,\quad\mbox{whence }|U|\leq|R|/2.
	\end{equation}
Let $W=T\setminus U$.
Because $W$ consists of vertices with fewer than $15$ neighbors in $\sigma^{-1}(1)$, {\bf P2} yields
	\begin{equation}\label{eqNowhereDense2a}
	\abs W\leq \frac{\kappa n}{3k}.
	\end{equation}
Since $\sigma,\tau$ are balanced, we have
	\begin{equation}\label{eqNowhereDense2b}
	\abs{S}+|R|=|\sigma^{-1}(1)|\sim\frac{n}{k}\sim|\tau^{-1}(1)|=\abs{S}+|U|+|W|.
	\end{equation}
Hence, by~(\ref{eqNowhereDense2}) and (\ref{eqNowhereDense2a})
	\begin{equation}\label{eqNowhereDense2c}
	|R|=|U|+|W|+o(n)\leq\frac{|R|}2+|W|+o(n)\leq \frac{|R|}2+\frac{\kappa n}{3k}+o(n),\quad\mbox{whence }|R|\leq\frac{2\kappa n}{3k}+o(n).
	\end{equation}
Finally, (\ref{eqNowhereDense2b}) and (\ref{eqNowhereDense2c}) imply that
	$\rho_{11}(\sigma,\tau)=\frac kn\cdot |S|= 1+o(1)-\frac kn\cdot |R|>1-\kappa,$
as desired.
\end{proof}

As a next step, we are going to verify that the $\sigma$-complete vertices take the same color in all the colorings in $\cC(\sigma)$ \whp;
a similar argument  was used in~\cite{Barriers}.	

\begin{lemma}\label{Lemma_complete}
Assume that $2k\ln k-\ln k-2\leq d\leq2k\ln k$ and let $\sigma\in\cB$.
\Whp\ the random graph $G(n,m,\sigma)$ has the following property.
	\begin{quote}
	If $\tau\in\cC(\sigma)$, then 
	for all $\sigma$-complete vertices $v$ we have $\sigma(v)=\tau(v)$ \whp
	\end{quote}
\end{lemma}
\begin{proof}
By \Lem s~\ref{Lemma_P} and~\ref{Lemma_sep} we may assume that {\bf P3} holds and that $\sigma$ is separable in $G(n,m,\sigma)$.
Let $V'$ be the core of this random graph.
Moreover, set
	\begin{eqnarray}\nonumber
	\Delta_i^+=\cbc{v\in V':\tau(v)=i\neq\sigma(v)},&&
	\Delta_i^-=\cbc{v\in V':\tau(v)\neq i=\sigma(v)}\qquad\mbox{for $i\in\brk k$, so that}\\
	\label{eqfrozenContra1}
	\sum_{i=1}^k\abs{\Delta_i^+}=\abs{\cbc{v\in V':\sigma(v)\neq\tau(v)}}&=&\sum_{i=1}^k\abs{\Delta_i^-}.
	\end{eqnarray}
The assumptions that $\sigma$ is separable and that both $\sigma,\tau$ are asymptotically balanced imply that
	\begin{eqnarray}
	\max_{i\in\brk k}|\Delta_i^+|\leq(\kappa+o(1))\frac nk,&&
		\max_{i\in\brk k}|\Delta_i^-|\leq(\kappa+o(1))\frac nk.
	\label{eqfrozenContra2}
	\end{eqnarray}
We are going to show that
	\begin{equation}\label{eqfrozenContra1a}
	\cbc{v\in V':\sigma(v)\neq\tau(v)}=\emptyset.
	\end{equation}
By construction, this implies that $\sigma(v)=\tau(v)$ for all $\sigma$-complete vertices.

To establish (\ref{eqfrozenContra1a}), let $S_i=\Delta_i^+\cup\Delta_i^-$ for $i=1,\ldots,k$.
Because $\Delta_i^+$ is contained in the core, each $v\in\Delta_i^+$ has at least $100$ neighbors in $\sigma^{-1}(i)$.
Since $\tau$ is a $k$-coloring, all of these neighbors  lie in the set $\Delta_i^-$.
Hence, the number $e(S_i)$ of edges spanned by $S_i$ is at least $100|\Delta_i^+|$.
On the other hand, (\ref{eqfrozenContra2}) implies that $|S_i|\leq k^{-4/3}n$ for all $i$.
Therefore, \textbf{P3} entails that $e(S_i)\leq5|S_i|$ for all $i$.
Thus, we obtain
	$100|\Delta_i^+|\leq  e(S_i)\leq5|S_i|\leq5(\abs{\Delta_i^+}+\abs{\Delta_i^-}).$
Consequently, $|\Delta_i^-|\geq2|\Delta_i^+|$ for all $i$.
Thus, (\ref{eqfrozenContra1}) shows that $\Delta_i^+=\Delta_i^-=\emptyset$ for all $i$, whence~(\ref{eqfrozenContra1a}) follows.
\end{proof}

\begin{proof}[Proof of \Lem~\ref{Lemma_first}.]
Let $\sigma\in\cB$.
We need to show that $G(n,m,\sigma)$ enjoys the properties {\bf T2--T3} from Definition~\ref{Def_good} \whp\
The fact that {\bf T2} holds \whp\ follows directly from \Lem~\ref{Lemma_sep}.

With respect to {\bf T3}, by \Lem~\ref{Lemma_complete} we may assume that that for all $\sigma$-complete $v$ and all $\tau\in\cC(\sigma)$ we have
$\tau(v)=\sigma(v)$.
Let $F_j$ be the set of $j$-free vertices for $j=1,2$.
By \Lem~\ref{Lemma_P} we may assume that
	\begin{equation}\label{eqLemma_first_1}
	|F_1|\leq\frac nk(1+\tilde O_k(1/k)),\quad F_2\leq\tilde O_k(k^{-2})n.
	\end{equation}
By construction, for any vertex $v\in F_1\setminus F_2$ there is a set $C_v\subset\brk k$ of at most two colors such that $\tau(v)\in C_v$ for all $\tau\in\cC(\sigma)$.
Hence,
	\begin{equation}\label{eqLemma_first_2}
	\abs{\cC(\sigma)}\leq 2^{F_1\setminus F_2}\cdot k^{F_2}.
	\end{equation}
Combining~(\ref{eqLemma_first_1}) and~(\ref{eqLemma_first_2}), we see that \whp\ in $G(n,m,\sigma)$,
	\begin{equation}\label{eqLemma_first_3}
	\frac1n\ln\cC(\sigma)\leq \frac{\ln 2}{k}+\tilde O_k(k^{-2}).
	\end{equation}

We need to compare the r.h.s.\ of~(\ref{eqLemma_first_3}) with $\frac1n\ln\Erw\brk\Zkb$.
By~(\ref{XeqFirstMoment}) and Taylor expansion,
	\begin{eqnarray*}
	\frac1n\ln\Erw\brk\Zkb&=&\ln k+\frac d2\ln(1-1/k)
		=\ln k-\frac d2\bc{\frac 1k+\frac1{2k^2}+O_k(k^{-3})}.
	\end{eqnarray*}
Writing $d=\dRS-\eps_k=2k\ln k-\ln k-2\ln 2-\eps_k$, we obtain
	\begin{equation}\label{eqLemma_first_4}
	\frac1n\ln\Erw\brk\Zkb=\ln k+\frac d2\ln(1-1/k)
		=\ln k-\frac d2\bc{\frac 1k+\frac1{2k^2}+O_k(k^{-3})}
			=\frac{\eps_k+\ln 2}k+O_k\hspace{-1mm}\bcfr{\ln k}{k^2}\hspace{-1mm}.
	\end{equation}
Letting, say, $\eps_k=\Theta_k(k^{-1/2})$, we obtain from~(\ref{eqLemma_first_3}) and~(\ref{eqLemma_first_4}) that
	$|\cC(\sigma)|\leq\Erw\brk\Zkb$ \whp\
Hence, {\bf T3} holds in $G(n,m,\sigma)$ \whp

Finally, upon direct inspection we find $f(\bar\rho)=2\ln k+d\ln(1-1/k)$.
Thus, (\ref{eqLemma_first_4}) shows that for $d=\dRS-\eps_k=2k\ln k-\ln k-2\ln 2-\eps_k$ we have
	$k\cdot f(\bar\rho)=2\ln 2+o_k(1)>0$, as claimed.
\end{proof}

\section{The Second Moment}\label{Sec_second}

\noindent{\em In this section we keep the assumptions of 
\Prop~\ref{Prop_second} and the
	notation introduced in \Sec~\ref{sec:outline}.}

\subsection{Overview}
The goal is to prove \Prop~\ref{Prop_second}.
As we already hinted at in \Sec~\ref{sec:outline}, this boils down to
maximizing $f(\rho)$ over $\rho\in\Dg$.
Formally, we have

\begin{proposition}\label{Prop_sufficient}
If $f(\rho)<f(\bar\rho)$ for any $\rho\in\Dg\setminus\cbc{\bar\rho}$, then $\Erw[\Zkg^2]\leq O(\Erw[\Zkg]^2)$.
\end{proposition}

\noindent
The proof of \Prop~\ref{Prop_sufficient}, based on the Laplace method, is a mere technical exercise, which we put off to \Sec~\ref{Sec_sufficient}.

\Prop~\ref{Prop_sufficient} reduces the second moment argument to a problem in analysis.
Indeed, neither the function $f$ nor the domain $\Dg$ over which we need to maximize are dependent on $n$
	(though both involve the parameters $d$ and~$k$).
In the following, we aim to establish

\begin{proposition}\label{Prop_opt}
If $\rho\in\Dg\setminus\cbc{\bar\rho}$, then $f(\rho)<f(\bar\rho)$.
\end{proposition}

\noindent
Thus, \Prop~\ref{Prop_second} is immediate from \Prop s~\ref{Prop_sufficient} and~\ref{Prop_opt}.

The proof of \Prop~\ref{Prop_opt} is the heart of the second moment argument.
Of course, we need to take a closer look at the function $f$.
As we will see, it consists of two ingredients: an entropy term and a probability term.
More specifically, suppose that $p:\Omega\ra\brk{0,1}$ is a probability distribution on a finite set $\Omega$ (i.e., $\sum_{x\in\Omega}p(x)=1$).
Recalling our convention that $0\ln 0=0$, we denote by
	$$H(p)=-\sum_{x\in\Omega}p(x)\ln p(x)$$
the \bemph{entropy} of $p$.
Since any $\rho\in\cD$ satisfies $\sum_{i,j}\rho_{ij}=k$, we can view $k^{-1}\rho$ as a probability distribution on $\brk k\times\brk k$.
Hence, we can write
	\begin{eqnarray*}
	f(\rho)&=&H(k^{-1}\rho)+E(\rho),\qquad\mbox{with}\qquad
	E(\rho)=\frac d2\cdot\ln\bc{1-\frac2k+\frac{\norm{\rho}_2^2}{k^2}}.
	\end{eqnarray*}
Combinatorially, $E(\rho)$ corresponds to the (logarithm of the) probability that $\sigma,\tau\in\cB$ with overlap $\rho$
simulataneously happen to be $k$-colorings, cf.\ the proof of Fact~\ref{Lemma_frho}.

It is clear that the entropy is {\em maximized} at the barycentre $\bar\rho$ of the Birkhoff polytope, because $k^{-1}\bar\rho$ is
	the uniform distribution on $\brk k\times\brk k$.
Furthermore, among all the matrices $\rho$ with non-negative entries that sum to $k$, $\bar\rho$ is the one that {\em minimizes} the Frobenius norm and hence $E(\rho)$.
This shows that $\bar\rho$ is a stationary point of $f(\rho)$.
But how do we prove that $\bar\rho$ is the global maximizer of $f$?

The domain $\Dg$ admits a natural decomposition into several subsets.
Let us call $\rho\in\Birk$ \bemph{$s$-stable} if the matrix has precisely $s$ entries that are greater than $0.51$.
Let $\Dgs$ denote the set of all $s$-stable $\rho\in\Dg$.
Geometrically, any $\rho\in\Dgs$ is close to a $k-s$-dimensional face of the Birkhoff polytope.
For if $\rho$ has $s$ entries greater than $0.51$, then by separability these entries are in fact at least $1-\kappa$
	(with $\kappa=\ln^{20}k/k$ as in~(\ref{eqseparable})).
Hence, $\rho$ is close to the face where these $s$ entries are equal to $1$.
Indeed, as all other entries of $\rho$ are smaller than $0.51$, $\rho$ is near a point ``deep inside'' that face.
Consequently, for any $1\leq s<k$ the set $\Dgs$ is disconnected:
	it consists of many tiny ``splinters'' near the $k-s$-dimensional faces of $\Birk$.
Each of these splinters can be mapped to the component where $\rho_{11},\ldots,\rho_{ss}>0.51$ by permuting
the rows and columns suitably, which does not affect the function $f$.

In the following, we are going to optimize $f$ separately over $\Dgs$ for each $0\leq s<k$.
We are going to argue that for each $s$, the point
	$\rhoss$ whose first $s$ diagonal entries are $1$ and whose $(i,j)$-entries are equal to $(k-s)^{-1}$ for $i,j>s$
	comes close to maximizing $f$ over $\Dgs$ (up to a negligible errror term in each case).
Geometrically, $\rhoss$ is the centre of the face defined by $\rho_{11}=\cdots=\rho_{ss}=1$.
Furthermore, in the case $s=0$ we have $\rhoss=\bar\rho$, and we will see that the maximum over $\mathcal D_{0,\mathrm{\good}}$ is attained at this very point.

We start by showing that we may confine ourselves to matrices without an entry in the interval $(0.15,1-\kappa)$.
Recall that $\cS$ is the set of all singly-stochastic $k\times k$-matrices.

\begin{proposition}\label{Prop_singly}
For all $\rho\in \cS$ such that $\rho_{ij}\in\brk{0.15,0.51}$~for some $(i,j)\in\brk k\times\brk k$ we have $f(\rho)<0$.
\end{proposition}

\noindent
We will see shortly how \Prop~\ref{Prop_singly} implies that $\bar\rho$ is the maximizer of $f$ over $D_{0,\mathrm{\good}}$.
In addition, there are three different ranges of $1\leq s<k$ that we  deal with separately.

\begin{proposition}\label{Prop_fewStable}
Suppose that $1\leq s\leq k^{0.999}$.
Then for all $\rho\in \Dgs$ we have $f(\rho)<f(\bar\rho)$.
\end{proposition}

\begin{proposition}\label{Prop_intermediate}
Suppose that $k^{0.999}<s<k-k^{0.49}$.
Then for all $\rho\in \Dgs$ we have $f(\rho)<f(\bar\rho)$.
\end{proposition}

\begin{proposition}\label{Prop_manyStable}
Suppose that $k-k^{0.49}\leq s<k$.
Then for all $\rho\in\Dgs$ we have $f(\rho)<f(\bar\rho)$.
\end{proposition}

The proofs of \Prop s~\ref{Prop_singly} and \ref{Prop_fewStable}--\ref{Prop_intermediate} are based on a local variations argument.
Roughly speaking, we are going to argue that if $\rho\in\Dgs$ is ``far'' from $\rhoss$,
	then a higher function value can be attained by moving slightly in the direction of $\rhoss$.
We expect that this argument can be adapted to perform second moment arguments
in other problems in probabilistic combinatorics.
Indeed, in such arguments the function that needs to be optimized is typically similar in nature to our $f$:
	an entropy term maximised at $\bar\rho$ plus a probability term minimized at $\bar\rho$.

More precisely, the following fact is the cornerstone of the local variations argument.
Let $\rho\in\cS$, let $i\in\brk k$ be a row index, and let $\emptyset\neq J\subset\brk k$ be a set of column indices.
Obtain $\hat\rho\in\cS$ from $\rho$ by letting
	\begin{equation}\label{eqhatrho}
	\mbox{$\hat\rho_{ab}=\rho_{ab}$ for all $(a,b)\not\in\cbc i\times J$ and
	$\hat\rho_{ib}=\frac1{\abs J}\sum_{j\in J}\rho_{ij}$ for all $b\in J.$}
	\end{equation}
That is, $\hat\rho$ is obtained by	 redistributing in row $i$ the total mass 
of the columns in $J$ equally over these columns.
Clearly, the entropy satisfies $H(k^{-1}\hat\rho)\geq H(k^{-1}\rho)$.
In fact, this inequality is strict unless $\hat\rho=\rho$.
However, it may well be that for the probability term we have $E(\hat\rho)<E(\rho)$.
The following proposition trades the increase in entropy against the drop in the probability term and shows that
$f(\hat\rho)\geq f(\rho)$ if $J$ is ``not too small'' and  $\max_{j\in J}\rho_{ij}$ is ``not too big''.

\begin{proposition}\label{Cor_Var}
Suppose that $\rho\in \cS$. Let $i\in\brk k$ and $J\subset\brk k$ be such that for some number $3\ln\ln k/\ln k\leq \lambda\leq 1$ we have $|J|\geq k^\lambda$.
Moreover, assume that $\max_{j\in J}\rho_{ij}<\lambda/2-\ln\ln k/\ln k$.
Then the matrix $\hat\rho$ from~(\ref{eqhatrho}) satisfies $f(\hat\rho)\geq f(\rho)$.
In fact, if $\rho\neq\hat\rho$, then $ f(\hat\rho)>f(\rho)$.
\end{proposition}

\noindent
Let us illustrate the use of \Prop~\ref{Cor_Var} by proving

\begin{corollary}\label{Cor_singly}
If $\rho\in\mathcal D_{0,\mathrm{\good}}\setminus\cbc{\bar\rho}$, then $f(\rho)<f(\bar\rho)$.
\end{corollary}
\begin{proof}
Let $\rho\in\mathcal D_{0,\mathrm{\good}}$.
Then $\rho_{ij}\leq0.51$ for all $i,j$ (as $\rho$ is $0$-stable).
In fact, if there are $i,j$ such that $\rho_{ij}>0.15$, then \Prop~\ref{Prop_singly} implies that
	$f(\rho)<0$, while $f(\bar\rho)>0$ by \Prop~\ref{Prop_first}.
Hence, we may assume that $\rho_{ij}\leq0.15$ for all $i,j$.
Let $\rho[l]$ be the matrix whose first $l$ rows are identical to those of $\bar\rho$,
and whose last $k-l$ rows are identical to those of $\rho$.
Thus, $\rho[0]=\rho$ and $\rho[k]=\bar\rho$.
We claim that
	\begin{equation}\label{eqCor_singly}
	f(\rho[i-1])\leq f(\rho[i]) \quad\mbox{for all }i=1,\ldots,k.
	\end{equation}
To obtain~(\ref{eqCor_singly}), we apply \Prop~\ref{Cor_Var} to the $i$th row of $\rho[i-1]$ with $J=\brk k$ and $\lambda=1$.
This is possible because $\max_{j}\rho_{ij}[i-1]=\max_{j}\rho_{ij}\leq0.15$. The resulting matrix $\hat\rho$ is precisely $\rho[i]$.
Thus, (\ref{eqCor_singly}) follows from \Prop~\ref{Cor_Var}.
Indeed, \Prop~\ref{Cor_Var} shows that one of the inequalities~(\ref{eqCor_singly}) is strict
	(as $\rho\neq\bar\rho$).
Hence, $f(\rho)<f(\bar\rho)$. 
\end{proof}

\noindent

\Prop~\ref{Prop_opt} is immediate from \Prop s~\ref{Prop_fewStable}--\ref{Prop_manyStable} and \Cor~\ref{Cor_singly}.
Thus, we are left to prove \Prop s~\ref{Prop_singly}--\ref{Cor_Var}.
In the Section~\ref{sec:ProofOfCorVar} we prove \Prop~\ref{Cor_Var}.
Building upon that estimate, we then proceed to prove \Prop s~\ref{Prop_singly}--\ref{Prop_manyStable}.
But before we start, we introduce a few pieces of notation and some basic facts. 

\subsection{Preliminaries.}
For $x\in\RR$ we denote by $\sign(x)\in\cbc{-1,0,1}$ the sign of $x$.
Moreover, if $\rho$ is matrix, then $\rho_i$ denotes the $i$th row of $\rho$ and $\rho_{ij}$ the $j$th entry of $\rho_i$.
We let $\norm{\rho}_\infty=\max_{i,j}|\rho_{ij}|$.
Further,
	$$h:[0,1]\ra\RRpos,\ z\mapsto-z\ln z-(1-z)\ln(1-z)$$
denotes the entropy function.
We recall the elementary inequality $h(z)\leq z(1-\ln z)$.
In addition, we note that
	\begin{equation}\label{eqSimpleMax}
	\max_{0<z<1}h(z)-z\ln k\leq1/k.
	\end{equation}
Indeed,  we have $h(z)-z\ln k\leq z(1-\ln z-\ln k)$ 
and differentiating twice, we see that
$z\mapsto z(1-\ln z-\ln k)$ takes its global maximum $1/k$ at $z=1/k$.

We need the following well-known fact about the entropy.

\begin{fact}\label{Prop_H}
Let $p\in\brk{0,1}^k$ be such that $\sum_{i=1}^kp_i=1$.
Then $H(p)\geq0$ and the following two statements hold.
\begin{description}
\item[H1] If $p$ is supported on a set of size $s$, then $H(p)\leq\ln s$.
\item[H2] Let $\cI\subset\brk k$ and suppose that
			$q=\sum_{i\in \cI}p_{i}\in(0,1)$.
		Let $p^\cI$ be the vector with entries $$p^\cI_{i}=p_{i}\cdot\vecone_{i\in \cI}\qquad\mbox{ for $i\in\brk k$.}$$
		Then
			$H(p)=h(q)+qH(q^{-1}p^\cI)+(1-q)H((1-q)^{-1}(p-p^\cI)).$
\end{description}
\end{fact}

\noindent
As an immediate consequence of Fact~\ref{Prop_H}, we have

\begin{corollary}\label{Cor_H}
Let $p\in\brk{0,1}^k$ be such that $\sum_{i=1}^kp_i=1$.
\begin{enumerate}
\item[(i)] Let $\cI\subset\brk k$ and set $q=\sum_{i\in \cI}p_{i}$. Then
		$H(p)\leq h(q)+q\ln \abs\cI+(1-q)\ln(k-\abs\cI).$
\item[(ii)] Let $\cI\subset\cbc{2,\ldots,k}$ be a set of size $0<\abs\cI<k-1$.
	Set $q=\sum_{i\in \cI}p_{i}$. If $p_1<1$, then
		$$H(p)\leq
			h(p_1)+(1-p_1)h(q/(1-p_1))
			+q\ln(\abs\cI)+(1-q-p_1)\ln(k-\abs\cI-1).$$
\end{enumerate}
\end{corollary}
\begin{proof}
The first claim follows simply by first using {\bf H2} and then applying {\bf H1} to $q^{-1}p^\cI$ and $(1-q)^{-1}(p-p^\cI)$.
To obtain the second assertion, use {\bf H2} with $\cI=\cbc{1}$ and then apply (i) to the probability distribution
	$q^{-1}p^{\cI}$.
\end{proof}

Let $\rho\in\cS$ be a singly-stochastic matrix.
We can view each row $\rho_i$ as a probability distribution on $\brk k$.
With this interpretation, we see that
	\begin{equation}\label{eqRowEntropy}
	H(k^{-1}\rho)=\ln k+\frac1k\sum_{i=1}^kH(\rho_i).
	\end{equation}

To facilitate the following calculations, we note that
	\begin{equation}\label{eqHdiff}
	\frac\partial{\partial p}-p\ln p=-1-\ln p.
	\end{equation}
Moreover, differentiating $E(\rho)$ by $y=\norm\rho_2^2$ and recalling that $d=2k\ln k+O_k(\ln k)$, we obtain
	\begin{eqnarray}\label{eqEnergyDiff}
	\frac{\partial}{\partial y}\,\frac d2\ln\bc{1-2/k+y/k^2}&=&\frac{d}{2k^2(1-2/k+y/k^2)}=\frac{\ln k}k(1+\tilde O_k(1/k)).
	\end{eqnarray}
Further, using the expansion $\ln(1+z)=z+z^2/2+O(z^3)$, we obtain the approximation
	\begin{equation}\label{eqEapprox}
	E(\rho)=\frac d{2k^2}\brk{-2k+\norm{\rho}_2^2-2\bc{1-\frac{\norm{\rho}_2^2}{2k}}^2}+o_k(1/k).
	\end{equation}

Finally, we calculate the function values $f(\rhoss)$ explicitly;
	recall that $\rhoss$ is the barycentre of the face of $\Birk$ defined by the equations $\rho_{11}=\cdots=\rho_{ss}=1$.
Let $1\leq s\leq k-1$.
The first $s$ rows of $\rhoss$ have entropy $0$, while the last $k-s$ rows have entropy $\ln(k-s)$.
Hence, (\ref{eqRowEntropy}) yields
	\begin{eqnarray}\label{eqLemma_pure1}
	H(k^{-1}\rhoss)&=&\ln k+\frac{k-s}k\ln(k-s)
		=2\ln k+(1-s/k)\ln(1-s/k)-\frac sk\ln k.
	\end{eqnarray}	
Moreover, $\norm{\rhoss}_2^2=s+1$.
Thus, using (\ref{eqEapprox}) and plugging in $d=2k\ln k-\ln k-c$ for some bounded $c$, we get
	\begin{eqnarray}\nonumber
	E(\rhoss)&=&\frac d{2k^2}\brk{-2k+s+1-2\bc{1-\frac{s+1}{2k}}^2}+o_k(1/k)\\
		&=&-2\ln k+\frac{c}k+\frac{s\ln k}{k}\bc{1+\frac{3}{2k}-\frac{s}{2k^2}}-\frac{cs}{2k^2}+o_k(1/k).\label{eqLemma_pure2}
	\end{eqnarray}
Since $f(\rho)=H(k^{-1}\rho)+E(\rho)$, (\ref{eqLemma_pure1}) and~(\ref{eqLemma_pure2}) yield
	\begin{eqnarray}\label{Lemma_pure}
	f(\rhoss)&=&\frac{c}k+(1-s/k)\ln(1-s/k)
		+\frac{s\ln k}{2k^2}\bc{3-\frac{s}{k}}-\frac{cs}{2k^2}+o_k(1/k).
	\end{eqnarray}

\subsection{Proof of \Prop~\ref{Cor_Var}.}\label{sec:ProofOfCorVar}

We pursue the following strategy.
Suppose that $a,b\in J$ are such that $\rho_{ia}=\min_{j\in J}\rho_{ij}$ and $\rho_{ib}=\max_{j\in J}\rho_{ij}$.
If $\rho_{ia}=\rho_{ib}$, then $\rho=\hat\rho$ and there is nothing to prove.
Otherwise, we are going to argue that increasing $\rho_{ia}$ slightly at the expense of $\rho_{ib}$ yields a matrix $\rho'$ with $f(\rho')>f(\rho)$.
We start by calculating  the partial derivatives of $f$.

\begin{lemma}\label{Lemma_variational}
Let $\rho\in\cS$.
Let $i,j,l\in\brk k$ and set $\delta=\rho_{il}-\rho_{ij}$.
Suppose that $\rho_{ij},\rho_{il}>0$.
Then
	\begin{equation}\label{eqLemma_variational}
	\sign\cbc{\frac{\partial f}{\partial\rho_{ij}}-\frac{\partial f}{\partial\rho_{il}}\bigg|_\rho}
		=\sign\cbc{1+\frac{\delta}{\rho_{ij}}-\exp\bc{\frac{d\cdot\delta}{k-2+\frac1{k}\norm{\rho}_2^2}}}.
	\end{equation}
\end{lemma}
\begin{proof}
Using~(\ref{eqHdiff}), (\ref{eqEnergyDiff}) and the chain rule, we obtain
	\begin{eqnarray*}
	\frac{\partial f}{\partial\rho_{ij}}-\frac{\partial f}{\partial\rho_{il}}&=&\frac1k\brk{\ln\bcfr{\rho_{il}}{\rho_{ij}}-\frac{d}{k}\cdot
			\frac{\rho_{il}-\rho_{ij}}{1-\frac2k+\frac1{k^2}\norm{\rho}_2^2}}.
	\end{eqnarray*}
Substituting $\delta=\rho_{il}-\rho_{ij}$,
we find
	$$\ln\bcfr{\rho_{il}}{\rho_{ij}}-\frac{d}{k}\cdot\frac{\rho_{il}-\rho_{ij}}{1-\frac2k+\frac1{k^2}\norm{\rho}_2^2}
		=\ln\bc{1+\delta/\rho_{ij}}-\frac{d\cdot\delta}{k(1-\frac2k+\frac1{k^2}\norm{\rho}_2^2)}.$$
Taking exponentials completes the proof.
\end{proof}

\noindent
As a next step, we take a closer look at the right hand side of~(\ref{eqLemma_variational}).

\begin{lemma}\label{Lemma_variational2}
Let $\rho\in\cS$,
let $i,j\in\brk k$ and assume that $\rho_{ij}>0$.
\begin{enumerate}
\item If
	\begin{equation}\label{eqnecint}
	\frac 1{\rho_{ij}}>\frac{d}{k-2+\frac1{k}\norm{\rho}_2^2},
	\end{equation}
	then there exists a unique $\delta^*>0$ such that
		$$1+\frac{\delta^*}{\rho_{ij}}=
		\exp\brk{\frac{d\cdot\delta^*}{k-2+\frac1{k}\norm{\rho}_2^2}}.$$
	Furthermore, for all $0<\delta<\delta^*$ we have
		$1+\frac{\delta}{\rho_{ij}}-
		\exp\brk{\frac{d}{k-2+\frac1{k}\norm{\rho}_2^2}\cdot\delta}>0.$
\item If (\ref{eqnecint}) does not hold,
	then for all $\delta>0$ we have
		$1+\frac{\delta}{\rho_{ij}}<
		\exp\brk{\frac{d}{k-2+\frac1{k}\norm{\rho}_2^2}\cdot\delta}.$
\end{enumerate}
\end{lemma}
\begin{proof}
There is at most one $\delta^*>0$ where the straight line $\delta\mapsto1+\frac{\delta}{\rho_{ij}}$ intersects the strictly convex function
	$$\delta\mapsto\exp\brk{\frac{d}{k-2+\frac1{k}\norm{\rho}_2^2}\cdot\delta}.$$
In fact, there is exactly one such $\delta^*$ iff the differential of the linear function is greater than that of the exponential function at $\delta=0$,
	which occurs iff~(\ref{eqnecint}) holds.
\end{proof}

\noindent\emph{Proof of \Prop~\ref{Cor_Var}.} 
If $\rho_{ij}=0$ for all $j\in J$, then $\hat\rho=\rho$ and there is nothing to show.
Thus, assume that $\sum_{j\in J}\rho_{ij}>0$.
Suppose that $\tilde\rho\in\cS$ maximizes $f(\tilde\rho)$ subject to the conditions
\begin{enumerate}
\item[i.] $\tilde\rho_{ab}=\rho_{ab}$ for all $(a,b)\not\in\cbc i\times J$ and
\item[ii.] $\max_{j\in J}\tilde\rho_{ij}\leq\max_{j\in J}\rho_{ij}$.
\end{enumerate}
Such a maximizer $\tilde\rho$ exists because i.--ii.\ define a compact domain.
Because $\tilde\rho\in\cS$ we have
	\begin{equation}\label{eqCor_Var0}
	\sum_{j\in J}\tilde\rho_{ij}=\sum_{j\in J}\rho_{ij}.
	\end{equation}

We claim that $\tilde\rho_{ij}>0$ for all $j\in J$.
Indeed, assume that $\tilde\rho_{ij}=0$ for $j\in J$ but $\tilde\rho_{il}>0$ for some other $l\in J$.
We recall that $f(\rho)=H(k^{-1}\rho)+E(\rho)$.
As~(\ref{eqHdiff}) and~(\ref{eqEnergyDiff}) show,
$\partial H(k^{-1}\rho)/\partial\rho_{ij}$ tends to infinity as $\rho_{ij}$ approaches $0$,
	while $|\partial E(\rho)/\partial\rho_{ij}|$ remains bounded.
Hence, there is $\xi>0$ such that the matrix $\rho'$ obtained from $\tilde\rho$ by replacing $\tilde\rho_{ij}$ by $\xi$
and $\tilde\rho_{il}$ by $\tilde\rho_{il}-\xi$ satisfied $f(\rho')>f(\tilde\rho)$,
in contradiction to the maximality of $f(\tilde\rho)$.

Thus, let $a$ be such that $\tilde\rho_{ia}=\min_{j\in J}\tilde\rho_{ij}>0$. 
Because $\tilde\rho$ is stochastic, we have $\norm{\tilde\rho}_2^2\in\brk{1,k}$ and $|J|\tilde\rho_{ia}\leq\sum_{j\in J}\tilde\rho_{ij}\leq1$.
Therefore, 
our assumptions $\lambda\geq 3\ln\ln k/\ln k$ and $d\leq 2k\ln k$ imply that
	\begin{equation}\label{eqCor_Var1}
	\frac 1{\tilde\rho_{ia}}\geq |J|\geq k^{\lambda}\geq3\ln k>\frac{d}{k-2+\norm{\tilde\rho}_2^2/k}.
	\end{equation}
Thus, (\ref{eqnecint}) is satisfied.
Further, setting $\hat\delta=\lambda/2-\ln \ln k/\ln k$, 
we find
	\begin{align}\nonumber
	\exp\bcfr{d\hat\delta}{k(1-2/k+k^{-2}\norm{\tilde\rho}_2^2)}&\leq\exp\bc{2\hat\delta\ln k
		}&\mbox{[as $d\leq2k\ln k$ and $\norm{\tilde\rho}_2^2\geq1$]}\\
			&\leq k^{\lambda}\ln^{-2}k\leq|J|\ln^{-2}k\nonumber\\
			&<1+\hat\delta/\tilde\rho_{ia}&\mbox{[as $\lambda\geq3\ln\ln k/\ln k$ and $1/\tilde\rho_{ia}\geq|J|$]}.
		\label{eqCor_Var2}
	\end{align}

Now, let $b\in J$ be such that $\tilde\rho_{ib}=\max_{j\in J}\tilde\rho_{ij}$ and
assume that  $\delta=\tilde\rho_{ib}-\tilde\rho_{ia}>0$.
Moreover, recall that we are assuming that $\tilde\rho_{ib}\leq\max_{j\in J}\rho_{ij}\leq\hat\delta$.
Since $\delta\leq\tilde\rho_{ib}\leq\hat\delta$,
	(\ref{eqCor_Var1}) and~(\ref{eqCor_Var2}) yield in combination with \Lem s~\ref{Lemma_variational} and~\ref{Lemma_variational2} that
	$$\frac{\partial f}{\partial\rho_{ia}}-\frac{\partial f}{\partial\rho_{ib}}\bigg|_{\tilde\rho}>0.$$
Hence, there is $\xi>0$ such that the matrix $\rho'$ obtained from $\tilde\rho$ by increasing $\tilde\rho_{ia}$ by $\xi$ and
decreasing $\tilde\rho_{ib}$ by $\xi$ satisfies $f(\rho')>f(\tilde\rho)$.
But this contradicts the maximality of $f(\tilde\rho)$ subject to i.--ii.
Thus, we conclude that
	$\min_{j\in J}\tilde\rho_{ij}=\tilde\rho_{ia}=\tilde\rho_{ib}=\max_{j\in J}\tilde\rho_{ib}$.
Therefore, (\ref{eqCor_Var0}) implies that $\tilde\rho=\hat\rho$ is the unique maximizer of $f$ subject to i.--ii.
\qed

\subsection{Proof of \Prop~\ref{Prop_singly}}\label{Sec_singly}
To proof is based on two key lemmas.
The first one rules out that $f(\rho)$ takes its maximum over $\rho\in\cS$ at a matrix with an entry close to $1/2$.

\begin{lemma}\label{Lemma_P1}
If $\rho \in \cS$ has an entry $\rho_{ij}\in\brk{0.49,0.51}$, then there is $\rho' \in \cS$ such that 	$f(\rho')\geq f(\rho)+\frac{\ln k}{5k}$.
\end{lemma}
\begin{proof}
Without loss of generality we may assume that $(i,j)=(1,1)$ and that
$\rho\in\cS$ maximizes $f$ subject to the condition that $\rho_{11}\in\brk{0.49,0.51}$.
There are two cases.
\begin{description}
\item[Case 1: $\rho_{1j}<0.49$ for all $j\geq2$]
	Applying \Prop~\ref{Cor_Var} to the set $J=\cbc{2,\ldots,k}$ (with $\lambda=\frac{\ln(k-1)}{\ln k}$), we see that
	$\rho_{1j}=\frac{1-\rho_{11}}{k-1}$ for all $j\geq2$, due to the maximality of $f(\rho)$.
	Hence, \Cor~\ref{Cor_H} yields
		\begin{eqnarray}\label{eqEntropyMiddle1}
		H(\rho_1)&\leq&h(\rho_{11})+(1-\rho_{11})\ln(k-1)\leq\ln2+0.51\ln k.
		\end{eqnarray}
	Moreover, because $\rho_{11}\leq0.51$ we have
		\begin{eqnarray}\label{eqEntropyMiddle2}
		\norm{\rho_1}_2^2
			&\leq&0.51^2+(k-1)\bcfr{1-\rho_{11}}{k-1}^2\leq0.261.
		\end{eqnarray}
	Let $\rho'$ be the matrix obtained from $\rho$ by replacing the first row by $(1,0,\ldots,0)$.
	Since $H(1,0,\ldots,0)=0$, (\ref{eqRowEntropy}) and (\ref{eqEntropyMiddle1}) yield
		\begin{eqnarray}\nonumber
		f(\rho)-f(\rho')&=&H(k^{-1}\rho)-H(k^{-1}\rho')+E(\rho)-E(\rho')\\
			&=&\frac{H(\rho_1)-H(1,0,\ldots,0)}k+E(\rho)-E(\rho')\leq\frac{\ln2+0.51\ln k}k+E(\rho)-E(\rho').
			\label{eqEntropyMiddle3}
		\end{eqnarray}
	Furthermore, (\ref{eqEntropyMiddle2}) entails
		$\norm{\rho}_2^2-\norm{\rho'}_2^2\leq\norm{\rho_1}_2^2-1\leq-0.739.$
	Hence, (\ref{eqEnergyDiff}) yields
		\begin{equation}\label{eqEntropyMiddle4}
		E(\rho)-E(\rho')\leq-(0.739+\tilde O_k(1/k))\ln k/k\leq-0.73\ln k/k.
		\end{equation}
	Combining~(\ref{eqEntropyMiddle3}) and~(\ref{eqEntropyMiddle4}), we obtain $f(\rho)-f(\rho')\leq\frac1k\brk{\ln2-0.22\ln k}\leq-\frac{\ln k}{5k}$.
\item[Case 2: there is $j\geq2$ such that $\rho_{1j}>0.49$]
	We may assume that $j=2$.
	Because $\sum_{j}\rho_{1j}=1$, we see that $\max_{j\geq 3}\rho_{1j}\leq0.02$.
	Hence, we can apply \Prop~\ref{Cor_Var} to $J=\cbc{3,\ldots,k}$ (with, say, $\lambda=1/2$).
	Due to the maximality of $f(\rho)$, we obtain
	$\rho_{1j}=(1-\rho_{11}-\rho_{12})/(k-2)$ for all $j\geq3$.
	Hence, \Cor~\ref{Cor_H} yields
		\begin{eqnarray}\label{eqEntropyMiddle5}
		H(\rho_1)&\leq&h(\rho_{11})+h(\rho_{12})
			+0.02\ln(k-2)\leq2\ln2+0.02\ln k.
		\end{eqnarray}
	Further, because $\rho_{11}^2+\rho_{12}^2\leq0.51^2+0.49^2$ as $\rho_{11},\rho_{12}\in\brk{0.49,0.51}$ and
		$\rho_{11}+\rho_{12}\leq1$, we see that
		\begin{eqnarray}\label{eqEntropyMiddle6}
		\norm{\rho_1}_2^2&\leq&0.51^2+0.49^2+(k-2)\bcfr{1-\rho_{11}-\rho_{12}}{k-2}^2\leq0.501.
		\end{eqnarray}
	As in the first case, obtain $\rho'$ from $\rho$ by replacing the first row by $(1,0,\ldots,0)$.
	From~(\ref{eqEntropyMiddle6}) we obtain $\norm{\rho}_2^2-\norm{\rho'}_2^2\leq0.501-1=-0.499$.
	Hence, (\ref{eqEnergyDiff}) yields
		\begin{eqnarray}\label{eqEntropyMiddle7}
		E(\rho)-E(\rho')&\leq&-0.499(1+\tilde O_k(1/k))\ln k/k\leq-0.49\ln k/k.
		\end{eqnarray}
	Combining (\ref{eqEntropyMiddle5}) and~(\ref{eqEntropyMiddle7}), we find
		\begin{eqnarray*}
		f(\rho)-f(\rho')&=&H(k^{-1}\rho)-H(k^{-1}\rho')+E(\rho)-E(\rho')\\
			&\leq&\frac1k\brk{2\ln2+0.02\ln k-0.49\ln k}\leq-\frac{\ln k}{5k}.
		\end{eqnarray*}
\end{description}
Hence, in either case we obtain the desired bound.
\end{proof}

The second key ingredient is

\begin{lemma}\label{Lemma_P2}
We have $\max_{\rho\in\cS} f(\rho)\leq\frac{\ln k}{8k}+O_k(1/k).$
\end{lemma}

\noindent
The proof of \Lem~\ref{Lemma_P2} requires two intermediate steps.
We start with the following exercise in calculus.

\begin{lemma}\label{Lemma_calculus}
Let $\xi:b\in(0,k/2)\mapsto k^{2b/k}(b^{-1}-k^{-1})$.
Let $\mu=\frac k2(1-\sqrt{1-2/\ln k})$.
Then $\xi$ is decreasing on the interval $(0,\mu)$ and increasing on $(\mu,k/2)$.
Furthermore, we have 
	\begin{equation}\label{eqLemma_calculus}
	-1/2\leq \xi'(b)\leq -3/2\qquad\mbox{for $b\in(0.99,1.01)$}.
	\end{equation}
\end{lemma}
\begin{proof}
The derivatives of $\xi$ are
	\begin{eqnarray*}
	\xi'(b)&=&k^{2b/k}\brk{\frac{2\ln k}k\bc{\frac1b-\frac1k}-\frac1{b^2}},\quad
	\xi''(b)=2k^{2b/k}\brk{\frac{2\ln^2k}{k^2}\bc{\frac1b-\frac1k}-\frac{2\ln k}{kb^2}+\frac1{b^3}}.
	\end{eqnarray*}
The first derivative vanishes at the two points $b=\frac k2(1\pm\sqrt{1-2/\ln k})$ only.
Moreover, an elementary calculation shows that $\mu=\frac k2(1-\sqrt{1-2/\ln k})$ is a local minimum,
while $\frac k2(1+\sqrt{1-2/\ln k})>k/2$ is a local maximum. 
Hence, $\xi$ is decreasing on the interval $(0,\mu)$ and increasing on $(\mu,k/2)$.
The last assertion follows by direct inspection of the above expression for $\xi'$.
\end{proof}

\begin{lemma}\label{Lemma_signly}
Let $\rho\in\cS$.
Suppose that $i\in\brk k$ is such that $\rho_{ij}\not\in\brk{0.49,0.51}$ for all $j\in\brk k$.
\begin{enumerate}
\item Suppose that $\rho_{ij}\leq0.49$ for all $j\in\brk k$.
	Let $\rho'$ be the stochastic matrix with entries
		$$\rho'_{hj}=\rho_{hj}\mbox{ and }\rho_{ij}'=1/k\quad\mbox{ for all }j\in\brk k,h\in\brk k\setminus\cbc i.$$
	Then $f(\rho)\leq f(\rho')$.
\item Suppose that $\rho_{ij}\geq0.51$ for some $j\in\brk k$.
	Then there is a number $\alpha=1/k+\tilde O_k(1/k^2)$ such that for
	the stochastic  matrix $\rho''$ with entries
		$$\rho''_{hj}=\rho_{hj}\mbox{ and }\rho_{ii}''=1-\alpha,\
			\rho''_{ih}=\frac{1-\alpha}{k-1}\quad\mbox{ for all }j\in\brk k,h\in\brk k\setminus\cbc i$$
 	we have
		$f(\rho)\leq f(\rho'').$
\end{enumerate}
\end{lemma}
\begin{proof}
To obtain the first assertion, we simply apply \Prop~\ref{Cor_Var} to row $i$ and $J=\brk k$ (with $\lambda=1$).
With respect to the second claim, we may assume without loss that $i=j=1$ and $\rho_{11}\geq0.51$. 
Let $\hat\rho\in\cS$ be the matrix that maximizes $f$ subject to the conditions
\begin{enumerate}
\item[i.] $\hat\rho_{11}\geq0.51$.
\item[ii.] $\hat\rho_a=\rho_a$ for all $a\in\cbc{2,\ldots,k}$.
	(In words, the last $k-1$ rows of $\hat\rho$ and $\rho$ coincide.)
\end{enumerate}
Since $\hat\rho_{1j}\leq1-\hat\rho_{11}\leq0.49$ for all $j\geq2$,
\Prop~\ref{Cor_Var} applies to $J=\cbc{2,\ldots,k}$  (with $\lambda=\frac{\ln(k-1)}{\ln k}$) and yields
	\begin{equation}\label{eqNecMax}
	\hat\rho_{12}=\cdots=\hat\rho_{1k}=\frac{1-\hat\rho_{11}}{k-1}.
	\end{equation}
Let $\delta=\hat\rho_{11}-\hat\rho_{12}$, let $0\leq\beta\leq0.49k$ be such that $\hat\rho_{11}=1-\beta/k$
and let $Q=1-1/k+\|\hat\rho\|_2^2/k^2.$

Because $\hat\rho$ is the maximizer of $f$ subject to i.\ and ii., \Lem~\ref{Lemma_variational} implies that
	\begin{equation}\label{eqNecMax2}
	\mbox{either $\beta\in\cbc{0,0.49k}$,  or }1+\frac{\delta}{\hat\rho_{12}}=\exp\bc{\frac{\delta d}{kQ}}.
	\end{equation}
We are going to argue that~(\ref{eqNecMax2}) entails that $\beta=1+\tilde O_k(1/k)$.

First, we observe that $\beta>0$.
For~(\ref{eqHdiff}) shows that the derivative $\partial H(\rho_1)/\partial\rho_{11}$ of the entropy of row $\rho_1$ tends to $-\infty$ as $\rho_{11}$ approaches $1$,
	while (\ref{eqEnergyDiff}) implies that the derivative $\partial E(\rho)/\partial\rho_{11}$ remains bounded in absolute value.
Hence, the maximality of $f(\rho)$ implies that $\beta>0$.

Further, since $\|\hat\rho\|_2^2\in\brk{1,k}$, we have $Q\geq(1-1/k)^2$.
Moreover, (\ref{eqNecMax}) implies that $\delta=\hat\rho_{11}-O_k(1/k)$.
Therefore, recalling that $d=2k\ln k+O_k(\ln k)$, we obtain
	\begin{eqnarray*}
	\exp\bcfr{\delta d}{kQ}&=&k^{2\hat\rho_{11}}\bc{1+\tilde O_k(1/k)}=k^{2(1-\beta/k)}(1+O_k(\ln k/k)),\\
	1+\frac{\delta}{\hat\rho_{12}}&=&\frac{\hat\rho_{11}}{\hat\rho_{12}}=\frac{(k-1)\hat\rho_{11}}{1-\hat\rho_{11}}=k^2(1/\beta-1/k) (1+O_k(1/k))
		\qquad\mbox{[as $\rho_{11}=1-\beta/k$]}.
	\end{eqnarray*}
Thus, with
	$\xi(b)=k^{2b/k}(b^{-1}-k^{-1})$ the function from \Lem~\ref{Lemma_calculus}, we see that for a certain $\eta=O_k(\ln k/k)$,
	\begin{equation}\label{eqUglyXi}
	(1-\eta)\cdot\xi(\beta)\leq\bc{1+\frac{\delta}{\hat\rho_{12}}}\exp\bc{-\frac{\delta d}{kQ}}\leq(1+\eta)\cdot\xi(\beta).
	\end{equation}
Let $\mu=\frac k2(1-\sqrt{1-2/\ln k})=(1+o_k(1))\frac{k}{2\ln k}$.
By \Lem~\ref{Lemma_calculus}, $\xi$ is decreasing on $(0,\mu)$.
Moreover, $\xi'(b)$ is negative and bounded
away from $0$ for $b$ close to $1$.
Hence, setting $\gamma=\ln^2 k/k$, we find
	\begin{align*}
		\xi(\beta)&\leq \xi(1+\gamma)<(1+\eta)^{-1}&\mbox{if  $\beta\in\brk{1+\gamma,\mu}$.}
	\end{align*}
In addition, $\xi$ is increasing on $(\mu,k/2)$.
Thus,
	\begin{align*}
		\xi(\beta)&\leq\xi(0.49k)\leq
			k^{0.98}
			\bc{\frac1{0.49k}-\frac1k}
				<(1+\eta)^{-1}&\mbox{if  $\beta\in\brk{\mu,0.49k}$.}
	\end{align*}
Plugging these two bounds into~(\ref{eqUglyXi}), we get
	\begin{align}\label{eqsinglyI}
	1+\frac{\delta}{\hat\rho_{12}}&<\exp\bc{\frac{\delta d}{kQ}}&\mbox{if $\beta\in[1+\gamma,0.49k]$.}
	\end{align}
Similarly, because $\mu$ is the unique local minimum of $\xi$, 
we have
	$$\xi(\beta)\geq\xi(1-\gamma)>(1-\eta)^{-1}\qquad\mbox{if }\beta\in(0,1-\gamma).$$
Hence, (\ref{eqUglyXi}) yields
	\begin{align}\label{eqsinglyII}
	1+\frac{\delta}{\hat\rho_{12}}&>\exp\bc{\frac{\delta d}{kQ}}&\mbox{if $\beta\in(0,1-\gamma)$.}
	\end{align}

Since we already know that $\beta>0$, (\ref{eqNecMax2}), (\ref{eqsinglyI}) and~(\ref{eqsinglyII}) imply $\beta\in[1-\gamma,1+\gamma]$.
Thus, $\beta=1+\tilde O_k(1/k)$ and  consequently $\hat\rho_{11}=1-\beta/k=1-1/k+\tilde O_k(k^{-2})$, as desired.
\end{proof}

\begin{proof}[Proof of \Lem~\ref{Lemma_P2}]
\Lem~\ref{Lemma_P1} implies that $\max_{\rho\in\cS}f(\rho)$ is attained at a matrix $\rho$ without entries in $\brk{0.49,0.51}$.
Therefore, \Lem~\ref{Lemma_signly} shows that the maximizer $\rho$ has the following form
	for some integer $0\leq s\leq k$ and certain $\alpha_i=1/k+\tilde O_k(1/k^2)$: 
	\begin{eqnarray}	
        \rho_{ij}=\left\{\begin{array}{cl}\label{eq:rho_maxi}
		1-\alpha_i&\mbox{ if }i=j\in\brk s,\\
		\frac{\alpha_i}{k-1}&\mbox{ if }i\in\brk s,j\neq i,\\
		1/k&\mbox{ otherwise.}
		\end{array}\right.
   \end{eqnarray}
Thus, for $i\in\brk s$ we have
	\begin{eqnarray}\label{eqLemmaP2_1}
	H(\rho_i)&=&h(1-\alpha_i)+\alpha_i\ln(k-1)\leq h(\alpha_i)+\alpha_i\ln k,\\
	\norm{\rho_i}_2^2&=&(1-\alpha_i)^2+\alpha_i^2/(k-1).\label{eqLemmaP2_2}
	\end{eqnarray}
Let $\rho'$ be the matrix obtained from $\rho$ by replacing the first $s$ rows by $(1,0,\ldots,0)$.
This matrix satisfies
	\begin{eqnarray}\label{eqLemmaP2_3}
	H(\rho_i')&=&0,\qquad\norm{\rho_i'}_2^2=1\qquad\mbox{for }i\in\brk s.
	\end{eqnarray}
Set $\alpha=\frac1s\sum_{i=1}^s\alpha_i=\frac1k+\tilde O_k(k^{-2})$.
Then (\ref{eqRowEntropy}), (\ref{eqLemmaP2_1})--(\ref{eqLemmaP2_3}) and the concavity of $h$ imply that
	\begin{eqnarray}\label{eqStochUpperCoarse1}
	H(k^{-1}\rho)-H(k^{-1}\rho')&=&\frac1k\sum_{i=1}^sH(\rho_i)\leq\frac sk\brk{h(\alpha)+\alpha\ln k}\leq\frac{\alpha s}k\brk{1-\ln\alpha+\ln k}\leq\frac{2\alpha s}k\brk{1+\ln k},
		\qquad\\
	\norm{\rho}_2^2-\norm{\rho'}_2^2&\leq&\sum_{i=1}^s\brk{(1-\alpha_i)^2+\frac{\alpha_i^2}{k-1}-1}
		=\sum_{i=1}^s\alpha_i \brk{-2+\alpha_i(1+1/(k-1))}\nonumber\\
		&=&\alpha s\brk{-2+O_k(1/k)}.\label{eqStochUpperCoarse3}
	\end{eqnarray}
Plugging~(\ref{eqStochUpperCoarse3}) into~(\ref{eqEnergyDiff}), we obtain
	\begin{eqnarray}\label{eqStochUpperCoarse2}
	E(\rho)-E(\rho')&\leq&\alpha s\brk{-2+O_k(1/k)}\cdot(1+\tilde O_k(1/k))\frac{\ln k}{k}
		\leq-\frac{2\alpha s}{k}\brk{\ln k+\tilde O_k(1/k)}.
	\end{eqnarray}
Combining~(\ref{eqStochUpperCoarse1}) and~(\ref{eqStochUpperCoarse2}) and recalling that $\alpha=1/k+\tilde O_k(1/k^2)$, we see that
	\begin{equation}\label{eqStochUpperCoarse99}
	f(\rho)-f(\rho')\leq\frac{2\alpha s}{k}\brk{1+\tilde O_k(1/k)}\leq 3/k.
	\end{equation}

To complete the proof, we calculate $f(\rho')$.
Recall that $d = 2k \ln k - \ln k - c$ with $c$ bounded.
Moreover, (\ref{eqLemmaP2_3}) shows that $\norm{\rho'_i}_2^2=1$ for $i=1,\ldots,s$.
In addition, since $\rho_{ij}'=1/k$ for all $i>s$, $j\in\brk k$, we get $\norm{\rho_i'}_2^2=1/k$ for $i>s$.
Hence, $\norm{\rho'}_2^2=1+(1-1/k)s$.
Thus, using~(\ref{eqEapprox}) and performing an elementary calculation, we get
	\begin{eqnarray*}
	E(\rho')&=&\frac d{2k^2}\brk{-2k+\norm{\rho'}_2^2-2\bc{1-\frac{\norm{\rho'}_2^2}{2k}}^2}+o_k(1/k)\\
		&=&-2\ln k+\frac{c}k+\frac{s\ln k}{k}\bc{1+\frac{1}{2k}-\frac{s}{2k^2}}-\frac{cs}{2k^2}+o_k(1/k).
	\end{eqnarray*}
Further, 
	$H(\rho_i')=0$ for $i\leq s$, while $H(\rho_i')=\ln k$ for $i>s$.
Hence, (\ref{eqRowEntropy}) yields $H(k^{-1}\rho')=\ln k+(1-s/k)\ln k=2\ln k-\frac sk\ln k$.
Thus,
	\begin{eqnarray}\nonumber
	f(\rho')&=&H(\rho')+E(\rho')=\frac{c}k+\frac{s\ln k}k\bc{\frac{1}{2k}-\frac{s}{2k^2}}-\frac{cs}{2k^2}+o_k(1/k)\\
		&=&\frac{c}k+\frac sk(1-s/k)\cdot\frac{\ln k}{2k}-\frac{cs}{2k^2}+o_k(1/k)=\frac sk(1-s/k)\cdot\frac{\ln k}{2k}+O_k(1/k).
			\label{eqexcess}
	\end{eqnarray}
Finally, combining~(\ref{eqStochUpperCoarse99}) and~(\ref{eqexcess}), we see that
	$f(\rho)\leq\frac sk(1-s/k)\cdot\frac{\ln k}{2k}+O_k(1/k)\leq\frac{\ln k}{8k}+O_k(1/k)$, as claimed.
\end{proof}

\begin{proof}[Proof of Proposition \ref{Prop_singly}]
Suppose that $\rho\in\cS$ has an entry $\rho_{ij}\in\brk{0.49,0.51}$.  We claim that $f(\rho)<0$.
Indeed, by \Lem s~\ref{Lemma_P1} and~\ref{Lemma_P2}
	$$f(\rho)\leq\max_{\rho'\in\cS}f(\rho')-\frac{\ln k}{5k}\leq\frac{\ln k}{8k}+O_k(1/k)-\frac{\ln k}{5k}<0.$$

Now, suppose that $\rho\in \cS$ has a row $i$ such that $\max_{j\in\brk k}\rho_{ij}\in\brk{0.15,0.49}$.
Without loss of generality, we may assume $i=1$ and $\rho_{11}=\max_{j\in\brk k}\rho_{ij}$.
In fact, we may assume that $\rho$ is the maximizer of $f$ subject to the condition $\rho_{11}=\max_j\rho_{1j}\in\brk{0.15,0.49}$.
Again, we show that $f(\rho)<0$.

What can we say about this maximizer $\rho$?
We apply \Prop~\ref{Cor_Var} to $i=1$ and $J=\cbc{2,\ldots,k}$:
if we let $\lambda=\ln(k-1)/\ln k$, then $|J|=k-1\geq k^{\lambda}$.
Moreover, $\rho_{1j}\leq0.49<\lambda/2-10/\ln k$ for all $j\in J$.
Hence, \Prop~\ref{Cor_Var} implies that
	\begin{equation}\label{eqallRhosTheSame}
	\rho_{12}=\cdots=\rho_{1k}.
	\end{equation}
Thus, \Cor~\ref{Cor_H} shows that the entropy of $\rho_1$ 
	is
	\begin{eqnarray*}
	H(\rho_1)&\leq&h(\rho_{11})+(1-\rho_{11})\ln(k-1).
	\end{eqnarray*}
By comparison,
let $\hat\rho$ be the matrix obtained from $\rho$ by replacing the first row by $\frac1k\vecone$.
Then
	$H(\hat\rho_1)=\ln k$.
Therefore, (\ref{eqRowEntropy}) yields
	\begin{eqnarray}\label{eqLemmaHalfBetter1}
	H(k^{-1}\rho)-H(k^{-1}\hat\rho)
		&=&-\frac1k\brk{\ln k-h(\rho_{11})-\bc{1-\rho_{11}}\ln\bc{k-1}}\leq-\rho_{11}\frac{\ln k}k+O_k(1/k).
	\end{eqnarray}
Moreover, (\ref{eqallRhosTheSame}) yields
$\norm{\rho_{1}}_2^2=\rho_{11}^2+(1-\rho_{11})^2/(k-1)$ and $\norm{\hat\rho_{1}}_2^2=1/k$, whence
	\begin{eqnarray*}
	\norm{\rho}_2^2-\norm{\hat\rho}_2^2&\leq&\rho_{11}^2+\frac{(1-\rho_{11})^2}{k-1}-1/k\leq\rho_{11}^2.
	\end{eqnarray*}
Hence, (\ref{eqEnergyDiff}) implies
	$E(\rho)-E(\hat\rho)\leq\rho_{11}^2\frac{\ln k}{k}+\tilde O_k(1/k^2).$
Combining this estimate with~(\ref{eqLemmaHalfBetter1}), we get
	\begin{equation}\label{eqLemmaHalfBetter1a}
	f(\rho)-f(\hat\rho)=H(k^{-1}\rho)-H(k^{-1}\hat\rho)+E(\rho)-E(\hat\rho)\leq-\rho_{11}(1-\rho_{11})\frac{\ln k}{k}+O_k(1/k).
	\end{equation}
Since $f(\hat\rho)\leq\frac{\ln k}{8k}+O_k(1/k)$ by \Lem~\ref{Lemma_P2}, we obtain from~(\ref{eqLemmaHalfBetter1a})
	$$f(\rho)\leq\brk{\frac18-\rho_{11}(1-\rho_{11})}\frac{\ln k}k+O_k(1/k).$$
The assertion follows because $\rho_{11}(1-\rho_{11})>1/8$ for $\rho_{11}\in\brk{0.15,0.49}$.
\end{proof}

\subsection{Proof of \Prop~\ref{Prop_fewStable}}\label{Sec_fewStable}

Let $1\leq s\leq k^{0.999}$ and let $\rho\in\Dgs$ be the maximiser of $f$.
Without loss of generality we may assume that $\rho_{ii}\geq0.51$ 
for $i=1,\ldots,s$ and $f(\rho_{ij})<0.51$ for all $(i,j)\not\in\cbc{(1,1),\ldots,(s,s)}$.
Because $\rho$ is separable, this implies that in fact $\rho_{ii}\geq 1-\kappa$ for $i=1,\ldots,s$, with $\kappa=\ln^{20}k/k$ as in~(\ref{eqseparable}). 
Furthermore,  if there is a pair $(i,j)\not\in\cbc{(1,1),\ldots,(s,s)}$ such that $\rho_{ij}\geq0.15$, then \Prop~\ref{Prop_singly} implies that $f(\rho)<0$.
In this case we are done, because $f(\bar\rho)>0$ by \Prop~\ref{Prop_first}.
Thus, assume from now on that $\rho_{ij}<0.15$ for all $(i,j)\not\in\cbc{(1,1),\ldots,(s,s)}$.

Let $\hat\rho$ be the singly-stochastic matrix with entries
	$$\hat\rho_{ij}=\left\{
		\begin{array}{cl}
		\rho_{ij}&\mbox{ if }i\in\brk k,j\leq s,\\
		\frac1{k-s}\sum_{l>s}\rho_{il}&\mbox{ if }i\in\brk k,j>s.
		\end{array}
		\right.$$
Since $k-s=(1-o_k(1))k$ and $\max_{j>s}\rho_{ij}<0.15$,
we can apply \Prop~\ref{Cor_Var} to $J=\brk k\setminus\brk s$ for any $i\in\brk k$ (with, say, $\lambda=1/2$).
Hence,
	\begin{equation}\label{eqProp_fewStable_1}
	f(\rho)\leq f(\hat\rho).
	\end{equation}
We are going to compare $f(\hat\rho)$ with $f(\rhoss)$, the barycentre of the face of $\Birk$ where the first $s$ diagonal entries are equal to one.
To this end, we need to estimate $f(\hat\rho)=H(k^{-1}\hat\rho)+E(\hat\rho)$.

As $\hat\rho$ is stochastic and $\hat\rho_{ii}=\rho_{ii}\geq1-\kappa$ for $i\leq s$, we find that
	\begin{equation}\label{eqClaim_eqfewStable4_0}
	q_i=\sum_{j\neq i}\hat\rho_{ij}=1-\rho_{ii}\leq\kappa\qquad\mbox{for $i\leq s$.}
	\end{equation}
Further, let $q_i=\sum_{j=1}^s\hat\rho_{ij}$ for $i>s$.
Because $\rho$ is doubly-stochastic and $\rho_{ii}\geq1-\kappa$ for $i\leq s$, we see that
	\begin{eqnarray}\label{eqfewStable1}
	\sum_{i>s}q_i=\sum_{i>s}\sum_{j=1}^s\hat\rho_{ij}=\sum_{i>s}\sum_{j=1}^s\rho_{ij}=\sum_{i=1}^s\sum_{j>s}\rho_{ij}\leq\kappa s.
	\end{eqnarray}
Based on~(\ref{eqClaim_eqfewStable4_0})--(\ref{eqfewStable1}), we obtain the following estimate of the entropy.

\begin{claim}\label{Claim_eqfewStable4}
We have $H(k^{-1}\hat\rho)\leq H(k^{-1}\rhoss)+o_k(1/k)$.
\end{claim}
\begin{proof}
By \Cor~\ref{Cor_H} and~(\ref{eqClaim_eqfewStable4_0}),
	\begin{eqnarray}\label{eqfewStable2}
	H(\hat\rho_i)&\leq&
		h(q_i)+q_i\ln k\leq h(\kappa)+\kappa\ln k\qquad\mbox{ for $i\leq s$}.
	\end{eqnarray}
Once more by \Cor~\ref{Cor_H},
	\begin{eqnarray}\label{eqfewStable2a}
	H(\hat\rho_i)&\leq&h(q_i)+q_i\ln s+(1-q_i)\ln(k-s)\leq h(q_i)+q_i\ln s+\ln(k-s)\qquad\mbox{for }i>s.
	\end{eqnarray}
Since $h$ is concave, (\ref{eqfewStable1}) and (\ref{eqfewStable2a}) yield
	\begin{equation}\label{eqfewStable3}
	\frac1k\sum_{i>s}H(\hat\rho_i)\leq\frac{k-s}k\ln(k-s)+\frac1k\sum_{i>s}(h(q_i)+q_i\ln s)\leq\frac{k-s}k\ln(k-s)+h\bcfr{\kappa s}k+\frac{\kappa s}k\ln s.
	\end{equation}
Plugging the bounds~(\ref{eqfewStable2}) and~(\ref{eqfewStable3}) into~(\ref{eqRowEntropy}), we arrive at
	\begin{align*}
	H(k^{-1}\hat\rho)&=\ln k+\frac1k\sum_{i=1}^kH(\hat\rho_i)\\
		&\leq\ln k+\frac sk\bc{h(\kappa)+\kappa\ln k}+
		\frac{k-s}k\ln(k-s)+h(\kappa s/k)+\frac{\kappa s}k\ln s\nonumber\\
		&\leq\ln k+\frac{k-s}k\ln(k-s)+o_k(1/k)&\mbox{[as $\kappa=\tilde O_k(1/k)$ and $s\leq k^{0.999}$]}\\
		&=H(k^{-1}\rhoss)+o_k(1/k)&[\mbox{by~(\ref{eqLemma_pure1})}],
	\end{align*}
thereby proving the claim.
\end{proof}

\begin{claim}\label{Claim_eqfewStable_E}
We have $E(\hat\rho)\leq E(\rhoss)+o_k(1/k).$
\end{claim}
\begin{proof}
As a first step, we show that there is a constant $\gamma>0$ such that
	\begin{equation}\label{eqSimpleFrobBound}
	\norm{\rho}_2^2\leq s+1+(\kappa s)^2\leq s+1+k^{-\gamma}.
	\end{equation}
Indeed, as $\hat\rho$ is a stochastic matrix, we have
	\begin{equation}\label{eqSimpleFrobBound1}
	\norm{\hat\rho_i}_2^2\leq1\quad\mbox{ for $i=1,\ldots,s$}.
	\end{equation}
Furthermore, since $\sum_{j>s}\rho_{ij}\leq1$ for each $i\in\brk k\setminus\brk s$, we have
	\begin{equation}\label{eqSimpleFrobBound2}
	\sum_{i>s}\sum_{j>s}\hat\rho_{ij}^2=(k-s)\sum_{i>s}\bcfr{\sum_{j>s}\rho_{ij}}{k-s}^2\leq1. 
	\end{equation}
Moreover, (\ref{eqfewStable1}) shows that $\sum_{i>s}q_i=\sum_{i>s}\sum_{j\leq s}\hat\rho_{ij}\leq\kappa s$.
Hence,
	\begin{equation}\label{eqSimpleFrobBound3}
	\sum_{i>s}\sum_{j\leq s}\hat\rho_{ij}^2\leq\big(\sum_{i>s}\sum_{j\leq s}\hat\rho_{ij}\big)^2\leq(\kappa s)^2. 
	\end{equation}
As $s\leq k^{0.999}$ and because $\kappa=\ln^{20}k/k$, there is a constant $\gamma>0$ such that
 $\kappa s\leq k^{-0.001}\ln^{20}k\leq k^{-\gamma/2}$ (provided that $k$ is sufficiently large).
Thus, combining~(\ref{eqSimpleFrobBound1})--(\ref{eqSimpleFrobBound3}), we obtain~(\ref{eqSimpleFrobBound}).

By comparison, we have $\|\rhoss\|_2^2=s+1$.
Hence, the bound (\ref{eqEnergyDiff}) on the derivative of $E$ and~(\ref{eqSimpleFrobBound}) yield
	$E(\hat\rho)\leq E(\rhoss)+o_k(1/k)$, as claimed.
\end{proof}

Combining Claims~\ref{Claim_eqfewStable4} and~\ref{Claim_eqfewStable_E}, we see that $f(\hat\rho)\leq f(\rhoss)+o_k(1/k)$.
Hence, (\ref{eqProp_fewStable_1}) yields
	\begin{align*}
	f(\rho)&\leq f(\hat\rho)\leq f(\rhoss)+o(1/k)\\
		&\leq \frac{c}k+(1-s/k)\ln(1-s/k)
			+\frac{s\ln k}{2k^2}\bc{3-\frac{s}{k}}-\frac{cs}{2k^2}+o_k(1/k)&[\mbox{due to~(\ref{Lemma_pure})}]\\
		&\leq \frac{c}k+(1-s/k)\ln(1-s/k)+o_k(1/k)&[\mbox{because $s\leq k^{0.999}$}]\\
		&\leq \frac{c}k-\frac sk(1-s/k)+o_k(1/k)&[\mbox{as $\ln(1-x)\leq-x$}]\\
		&= f(\bar\rho)-\frac sk(1-s/k)+o_k(1/k)&[\mbox{by \Prop~\ref{Prop_first}}].
	\end{align*}
The last expression is decreasing in $s$ (for $1\leq s\leq k^{0.999}$).
Thus, $f(\rho)< f(\bar\rho)-1/k+o_k(1/k)$.
This implies the assertion because we chose $\rho$ to be the maximizer of $f$ over $\Dgs$.
\qed

\subsection{Proof of \Prop~\ref{Prop_intermediate}}\label{Sec_intermediate}

Suppose that $k^{0.999}<s<k-k^{0.49}$ and let $\rho\in\Dgs$ be the maximizer of $f$ over $\Dgs$.
We may assume without loss that $\rho_{ii}\geq0.51$ for $i=1,\ldots,s$ and $\rho_{ij}<0.51$ for $(i,j)\not\in\cbc{(1,1),\ldots,(s,s)}$.
Due to separability, we thus have $\rho_{ii}\geq 1-\kappa$ for $i=1,\ldots,s$.
Further, we may assume that $\rho_{ij}\leq0.15$ for all $(i,j)\not\in\cbc{(1,1),\ldots,(s,s)}$ as otherwise
\Prop~\ref{Prop_singly} yields $f(\rho)<0<f(\bar\rho)$.

Let $\hat\rho$ be the stochastic matrix with entries
	$$\hat\rho_{ij}=\left\{
		\begin{array}{cl}
		\rho_{ij}&\mbox{ if }i=j\in\brk s,\\
		\frac1{s-1}{\sum_{l\in\brk s\setminus\cbc i}\rho_{il}}&\mbox{ if }i,j\leq s,\,i\neq j,\\
		\frac1{k-s}{\sum_{l>s}\rho_{il}}&\mbox{ if }j>s,\\
		\frac1{s}{\sum_{l\leq s}\rho_{il}}&\mbox{ if }j\leq s<i.
		\end{array}
		\right.$$
Since $\max_{i\neq j}\rho_{ij}\leq0.15$ and $s,k-s>k^{0.49}$, we can apply \Prop~\ref{Cor_Var}
to $J_i=\brk{k}\setminus\brk s$ and to $J_i'=\brk s\setminus\cbc i$ for all $i\in\brk k$ (with, say, $\lambda=0.4$).
We thus obtain
	\begin{equation}\label{eqProp_intermediate_avg}
	f(\rho)\leq f(\hat\rho).
	\end{equation}

To estimate $f(\hat\rho)$, let
	$$q_i=\sum_{j>s}\rho_{ij}=\sum_{j>s}\hat\rho_{ij}\mbox{ for $i\leq s$ and }q_i=\sum_{j\leq s}\rho_{ij}=\sum_{j\leq s}\hat\rho_{ij}\mbox{ for }i>s.$$
Since $\rho$ is doubly-stochastic and $\rho_{ii}\geq1-\kappa$ for $i\leq s$, we see that
	\begin{eqnarray}\label{eqintermediate4}
	q&=&\sum_{i>s}q_i=\sum_{i\leq s}q_i\leq\sum_{i=1}^s1-\rho_{ii}\leq\kappa s.
	\end{eqnarray}
In addition, let
	\begin{eqnarray}
	t_i&=&\sum_{ j\in\brk s\setminus\cbc i}\hat\rho_{ij}=
			\sum_{j\in\brk s\setminus\cbc i}\rho_{ij}\leq1-\rho_{ii}\leq\kappa\qquad\mbox{for }i\leq s.\label{eqintermediate5}
	\end{eqnarray}

\begin{claim}\label{Claim_eqintermediate77}
We have
	$\displaystyle
	H(\hat\rho)\leq2\ln k+\frac{3q(2+\ln k)}k+\bc{1-s/k}\ln(1-s/k)-\frac{s\ln k}k+\frac{2\ln k}k\sum_{i=1}^st_i+O_k(1/k)$.
\end{claim}
\begin{proof}
Applying \Cor~\ref{Cor_H}, we obtain 
	\begin{eqnarray}\label{eqintermediate0}
	H(\hat\rho_i)&\leq&h(t_i)+t_i\ln s+h(q_i)+q_i\ln(k-s)\qquad\mbox{for }i\leq s.
	\end{eqnarray}
Set
	$$\tilde H=\frac1k\sum_{i\leq s}h(t_i)+t_i\ln s.$$
Summing (\ref{eqintermediate0}) up, recalling from~(\ref{eqintermediate4}) that $q=\sum_{i\leq s}q_i$, and using the convavity of $h$, we get
	\begin{eqnarray}\label{eqintermediate1}
	\frac1k\sum_{i=1}^sH(\hat\rho_i)&\leq&\tilde H+\frac sk h(q/s)+\frac qk\ln(k-s).
	\end{eqnarray}
Furthermore, again by \Cor~\ref{Cor_H}, for $i>s$ we have
	\begin{eqnarray*}
	H(\hat\rho_i)&\leq&h(q_i)+q_i\ln s+(1-q_i)\ln(k-s).
	\end{eqnarray*}
Once more due to the concavity of $h$ and as $q=\sum_{i>s}q_i$, we see that
	\begin{eqnarray}\label{eqintermediate2}
	\frac1k\sum_{i>s}H(\hat\rho_i)&\leq&\frac{k-s}k h(q/(k-s))+\frac qk\ln s+\frac{k-s-q}k\ln(k-s).
	\end{eqnarray}
Combining~(\ref{eqintermediate1}) and~(\ref{eqintermediate2}), we get
	\begin{eqnarray*}
	H(\hat\rho)&\leq&\tilde H+\ln k+\brk{\frac sk h(q/s)+\frac qk\ln(k-s)}+
		\brk{\frac{k-s}kh(q/(k-s))+\frac qk\ln s}
		+\frac{k-s-q}k\ln(k-s).
	\end{eqnarray*}
Using the elementary inequality $h(z)\leq z(1-\ln z)$ to
simplify the above, we get
	\begin{eqnarray}\nonumber
	H(\hat\rho)-\tilde H&\leq&\ln k+\frac qk\brk{2+\ln(s/q)+\ln((k-s)/q)+\ln s+\ln(k-s)}+\frac{k-s-q}k\ln(k-s)\\
		&\leq&\ln k+\frac qk\brk{2+2\ln(s)+\ln(k-s)-2\ln q}+\frac{k-s}k\ln(k-s)\nonumber\\
		&\leq&\ln k+\frac{3q(2+\ln k)}k+\frac{k-s}k\ln(k-s)+O_k(1/k)\qquad\mbox{[as $-z\ln z\leq1$ for all $z>0$]}\nonumber\\
		&=&2\ln k+\frac{3q(2+\ln k)}k+\bc{1-s/k}\ln(1-s/k)-\frac{s\ln k}k+O_k(1/k).
		\label{eqintermediate3}
	\end{eqnarray}
Since $s\leq k$, we obtain
	\begin{equation}		\label{eqintermediate7}
	\tilde H-\frac{2\ln k}k\sum_{i=1}^st_i
		=\frac1k\sum_{i\leq s}h(t_i)+t_i(\ln s-2\ln k)
		\leq\frac1k\sum_{i=1}^s h(t_i)-t_i\ln k\leq\frac1k\qquad\mbox{[due to~(\ref{eqSimpleMax})].}
	\end{equation}
Finally, the assertions follows by combining~(\ref{eqintermediate3}) and~(\ref{eqintermediate7}).
\end{proof}

\begin{claim}\label{Claim_eqintermediate6}
We have
	$ E(\hat\rho)
		=-2\ln k+\frac{s\ln k}{k}\bc{1+\frac{3}{2k}-\frac{s}{2k^2}}-\frac{2\ln k}k\sum_{i=1}^st_i+\tilde O_k(1/k).$
\end{claim}
\begin{proof}
As a first step, we show that
	\begin{equation}\label{eqFrobNotSoSimple}
	\norm{\rho}_2^2\leq s+1-2\sum_{i=1}^st_i+o_k(1/\ln k).
	\end{equation}
Indeed, together with the definition of $\hat\rho$, equation (\ref{eqintermediate5}) shows that for $i\in\brk s$,
	\begin{eqnarray}\label{eqFrobNotSoSimple1}
	\hat\rho_{ii}^2&\leq&(1-t_i)^2=1-2t_i+t_i^2\leq1-2t_i+\kappa^2\qquad\mbox{ and}\\
	\sum_{j\in\brk s\setminus\cbc i}\hat\rho_{ij}^2&=&(s-1)\cdot\bcfr{t_i}{s-1}^2\leq\frac{\kappa^2}{s-1}\leq\kappa^2.
	\end{eqnarray}
Moreover, since $\hat\rho$ is stochastic and $\hat\rho_{ii}\geq1-\kappa$ if $i\leq s$, we have
	\begin{equation}\label{eqFrobNotSoSimple3}
	\sum_{j\in\brk k\setminus\brk s}\hat\rho_{ij}^2\leq\kappa^2\qquad\mbox{for }i\in\brk s.
	\end{equation}
Combining~(\ref{eqFrobNotSoSimple1})--(\ref{eqFrobNotSoSimple3}) and recalling that $\kappa=\tilde O_k(k^{-1})$, we obtain
	\begin{eqnarray}\label{eqFrobNotSoSimple4}
	\sum_{i=1}^s\norm{\hat\rho_i}_2^2&\leq&s+3\kappa^2s-2\sum_{i=1}^st_i=s+o_k(1/\ln k)-2\sum_{i=1}^st_i.
	\end{eqnarray}
Further, since $\rho_{jj}\geq1-\kappa$ for $j\leq s$ and because $\rho$ is doubly-stochastic,
we have $\rho_{ij}\leq\kappa$ for all $j\leq s<i$. 
By the construction of $\hat\rho$, this implies that $\hat\rho_{ij}\leq\kappa$ for all $j\leq s<i$.
Furthermore, $q=\sum_{i>s}\sum_{j\in\brk s}\hat\rho_{ij}\leq\kappa s$ by~(\ref{eqintermediate4}). 
As a sum of squares is maximized if the summands are as unequal as possible, we obtain
	\begin{eqnarray}\label{eqFrobNotSoSimple5}
	\sum_{i>s}\sum_{j\in\brk s}\hat\rho_{ij}^2&\leq&\kappa^2s=o_k(1/\ln k).
	\end{eqnarray}
In addition, once more by the construction of $\hat\rho$,
	\begin{eqnarray}\label{eqFrobNotSoSimple6}
	\sum_{i>s}\sum_{j>s}\hat\rho_{ij}^2&=&
		\sum_{i>s}(k-s)\bcfr{\sum_{j>s}\rho_{ij}}{k-s}^2\leq
		(k-s)^2\cdot\bcfr{1}{k-s}^2=1.
	\end{eqnarray}
Combining~(\ref{eqFrobNotSoSimple4})--(\ref{eqFrobNotSoSimple6}), we obtain~(\ref{eqFrobNotSoSimple}).

By comparison, we have $\|\rhoss\|_2^2=s+1$.
Hence, (\ref{eqEnergyDiff}) implies together with (\ref{eqFrobNotSoSimple}) that
	$$E(\hat\rho)\leq E(\rhoss)-\frac{2\ln k}k\sum_{i=1}^st_i+\tilde O_k(1/k).$$
Plugging in the expression~(\ref{eqLemma_pure2}) for $E(\rhoss)$ yields the assertion.
\end{proof}

Finally, combining
Claims~\ref{Claim_eqintermediate77} and~\ref{Claim_eqintermediate6}, we see that
	\begin{eqnarray}\nonumber
	f(\rho)&\leq&f(\hat\rho)\leq\bc{1-s/k}\ln(1-s/k)+\frac{3q(2+\ln k)}k+
		\frac{s\ln k}{k}\bc{\frac{3}{2k}-\frac{s}{2k^2}}
			+\tilde O(1/k)\\
		&=&\bc{1-s/k}\ln(1-s/k)+\tilde O(1/k)\leq-\frac sk(1-s/k)+\tilde O_k(1/k).\label{eqeqintermediate666}
	\end{eqnarray}
Our assumption $k^{0.999}<s<k-k^{0.49}$ ensures that $-\frac sk(1-s/k)+\tilde O_k(1/k)<0$.
Thus, (\ref{eqeqintermediate666}) and \Prop~\ref{Prop_first} show that $f(\rho)<0<f(\bar\rho)$.
This completes the proof as $\rho$ was chosen to be the maximizer of $f$ over $\Dgs$.\qed

\subsection{Proof of \Prop~\ref{Prop_manyStable}}\label{Sec_manyStable}

Suppose that $k-\sqrt k\leq s\leq k-1$ and that $\rho\in\Dgs$ maximizes of $f$ over $\Dgs$.
As before, we assume without loss that $\rho_{ii}\geq0.51$ for $i=1,\ldots,s$ and $\rho_{ij}<0.51$ for $(i,j)\not\in\cbc{(1,1),\ldots,(s,s)}$.
Thus, $\rho_{ii}\geq 1-\kappa$ for $i=1,\ldots,s$ as $\rho$ is separable.
Further, if $\rho_{ij}>0.15$ for some $(i,j)\in\cbc{(1,1),\ldots,(s,s)}$, then $f(\rho)<0<f(\bar\rho)$ by \Prop~\ref{Prop_singly}.
Hence, we assume $\rho_{ij}\leq0.15$ for all $(i,j)\not\in\cbc{(1,1),\ldots,(s,s)}$.

Let $q_i=\sum_{j\neq i}\rho_{ij}$ for $i\in\brk s$.
Because $\rho$ is doubly-stochastic and $\rho_{ii}\geq1-\kappa$ for $i\leq s$, we see that
	\begin{equation}\label{eqmanyStableq}
	q=\sum_{i=1}^sq_i=
		\sum_{i=1}^s\sum_{j\neq i}\rho_{ij}=\sum_{i=1}^s1-\rho_{ii}\leq\kappa s.
	\end{equation}
In addition, let
	$$t_i=\sum_{j>s}\rho_{ij},\quad t=\sum_{i=1}^st_i.$$
Since $\rho$ is doubly-stochastic, we have
	\begin{equation}\label{eqmanyStableqt}
	t=\sum_{i=1}^s\sum_{j>s}\rho_{ij}=\sum_{i>s}\sum_{j=1}^s\rho_{ij}.
	\end{equation}
We are going to compare $f(\rho)$ with $f(\id)$, where $\id$ is the identity matrix (with ones on the diagonal and zeros elsewhere).

\begin{claim}\label{Claim_eqentropycrit3}
With
	$\cH=\frac1k\sum_{i=1}^s h(\rho_{ii})$ we have
	$H(k^{-1}\rho)\leq\ln k+\cH+\frac{q}k\ln k+0.51(k-s)\frac{\ln k}k$.
\end{claim}
\begin{proof}
\Cor~\ref{Cor_H} implies together with the concavity of $h$ that
	\begin{eqnarray}
	\frac1k\sum_{i=1}^sH(\rho_i)
	&\leq&\frac1k\sum_{i=1}^sh(\rho_{ii})+q_i h(t_i/q_i)+t_i\ln(k-s)+(q_i-t_i)\ln s\nonumber\\
	&\leq&\cH+\frac qk h(t/q)+\frac tk\ln(k-s)+\frac{q-t}k\ln(s)\nonumber\\
	&\leq&\cH+\frac tk(1-\ln t+\ln q)+\frac tk\ln(k-s)+\frac{q-t}k\ln(s)\qquad
		 [\mbox{as $h(z)\leq z(1-\ln z)$}].
	\label{eqentropycrit1a}
	\end{eqnarray}
Because $-z\ln z\leq1$ for all $z>0$, we have $-\frac tk\ln t\leq 1/k$.
Moreover, as $\rho$ is doubly-stochastic (\ref{eqmanyStableqt}) implies that $t\leq k-s$.
Additionally, (\ref{eqmanyStableq}) shows that $q\leq \kappa s\leq\kappa k=\tilde O_k(1)$, because $\kappa=\ln^{20}k/k$.
Thus,
	$$\frac tk(1-\ln t+\ln q)\leq\frac{k-s}k\cdot O_k(\ln\ln k).$$
Plugging this last estimate into~(\ref{eqentropycrit1a}), we obtain	
	\begin{eqnarray}	\label{eqentropycrit1}
	\frac1k\sum_{i=1}^sH(\rho_i)&\leq&\cH+\frac tk\ln(k-s)+\frac{q-t}k\ln(s)+\frac{k-s}k\cdot O_k(\ln\ln k).
	\end{eqnarray}
Furthermore, using  \Cor~\ref{Cor_H}, (\ref{eqmanyStableqt}) and the concavity of $h$, we see that
	\begin{eqnarray}
	\frac1k\sum_{i>s}H(\rho_i)&\leq&
			\frac1k\sum_{i>s}h\bc{\sum_{j=1}^s\rho_{ij}}+\sum_{j=1}^s\rho_{ij}\ln(s)+\bc{1-\sum_{j=1}^s\rho_{ij}}\ln(k-s)\nonumber\\
		&\leq&\frac{k-s}k h\bcfr{t}{k-s}+\frac tk\ln s+\frac{k-s-t}k\ln(k-s)\nonumber\\
		&\leq&\frac{k-s}k\ln2+\frac tk\ln s+\frac{k-s-t}k\ln(k-s)\qquad[\mbox{as $h(z)\leq\ln2$ for all $z$}].\label{eqentropycrit2}
	\end{eqnarray}
Plugging~(\ref{eqentropycrit1}) and~(\ref{eqentropycrit2}) into~(\ref{eqRowEntropy}), we find
	\begin{eqnarray}
	H(k^{-1}\rho)
		&\leq&\ln k+\cH+\frac{q}k\ln k+\frac{k-s}k\ln(k-s)+\frac{k-s}k\cdot O_k(\ln\ln k)\nonumber\\
		&\leq&\ln k+\cH+\frac{q}k\ln k+\frac{k-s}{2k}\ln k+\frac{k-s}k\cdot O_k(\ln\ln k)\qquad[\mbox{as $k-s\leq\sqrt k$}]\nonumber\\
		&\leq&\ln k+\cH+\frac{q}k\ln k+0.51(k-s)\frac{\ln k}k,		\label{eqentropycrit3}
	\end{eqnarray}
as claimed.
\end{proof}

\begin{claim}\label{Claim_eqentropycrit4}
We have
	$E(\rho)\leq E(\id)+(1+\tilde O_k(1/k))\frac{\ln k}{k}\bc{-0.85(k-s)+\sum_{i=1}^s(\rho_{ii}^2-1)}.$
\end{claim}
\begin{proof}
The Frobenius norm of $\rho$ can be estimated as follows.
Since $\rho_{ii}\geq1-\kappa$ for all $i\leq s$ and $\rho$ is stochastic,
we have $\rho_{ij}\leq\kappa$ for all $i\leq s$, $j\neq i$.
Hence, the bound~(\ref{eqmanyStableq}) implies together with the fact that a sum of squares is maximized by having the summands
as unequal as possible that
	\begin{eqnarray}\label{eqUltFrob1}
	\sum_{i=1}^s\norm{\rho_i}_2^2\leq \left\lceil\frac{q}\kappa\right\rceil\cdot\kappa^2+
		\sum_{i=1}^s\rho_{ii}^2
		\leq s\kappa^2+\sum_{i=1}^s\rho_{ii}^2
		\leq\tilde O_k(1/k)+\sum_{i=1}^s\rho_{ii}^2\qquad[\mbox{as $\kappa\leq\ln^{20}k/k$}].
	\end{eqnarray}
A similar argument applies to the remaining rows.
More precisely, if $i>s$ then $\rho_{ij}\leq 0.15$ for all $j$ by our initial assumption on $\rho$.
Therefore,
	\begin{equation}\label{eqUltFrob2}
	\sum_{i>s}\norm{\rho_i}_2^2\leq \frac{k-s}{0.15}\cdot(0.15)^2=0.15(k-s).
	\end{equation}
Combining~(\ref{eqUltFrob1}) and~(\ref{eqUltFrob2}), we arrive at
	\begin{equation}\label{eqUltFrob3}
	\norm{\rho}_2^2\leq 
			\sum_{i=1}^s\rho_{ii}^2+0.15(k-s)+\tilde O_k(1/k).
	\end{equation}
By comparison, $\norm\id_2^2=k$.
Thus, (\ref{eqUltFrob3}) yields
	$\norm\rho_2^2-\norm\id_2^2\leq
		-0.85(k-s)+\sum_{i=1}^s(\rho_{ii}^2-1)+\tilde O_k(1/k)$.
Combining this estimate with (\ref{eqEnergyDiff}) completes the proof.
\end{proof}

Observing that $H(k^{-1}\id)=\ln k$ and
using
Claims~\ref{Claim_eqentropycrit3} and~\ref{Claim_eqentropycrit4}, we obtain 
	\begin{eqnarray}\nonumber
	f(\rho)-f(\id)&=&H(k^{-1}\rho)-\ln k+E(\rho)-E(\id)\\
		&\leq&\cH+\frac qk\ln k-
		\frac{k-s}{3k}\ln k+(1+\tilde O_k(1/k))\frac{\ln k}k\sum_{i=1}^s(\rho_{ii}^2-1).
				\label{eqmanyStablef}
	\end{eqnarray}
To complete the proof, let $r_i=1-\rho_{ii}$ for $i=1,\ldots,s$.
Then~(\ref{eqmanyStableq}) shows that $q=\sum_{i=1}^sr_i$.
Moreover, $\cH=\frac1k\sum_{i=1}^sh(r_i)$, as $h(1-z)=h(z)$ for all $z$.
Since $r_i\leq\kappa=\tilde O_k(1/k)$, we have
	\begin{eqnarray*}
	\cH+\frac qk\ln k+\frac{\ln k}k\sum_{i=1}^s(\rho_{ii}^2-1)&=&
		\frac1k\sum_{i=1}^s\brk{h(r_i)+r_i\ln k+((1-r_i)^2-1)\ln k}\\
		&=&\frac1k\sum_{i=1}^s\brk{h(r_i)+r_i\ln k+(r_i^2-2r_i)\ln k}\\
		&\leq&\tilde O_k(1/k^2)+\frac1k\sum_{i=1}^sh(r_i)-r_i\ln k
			\leq O_k(1/k)\qquad\mbox{[by~(\ref{eqSimpleMax})]}.
	\end{eqnarray*}
Plugging this bound into~(\ref{eqmanyStablef}) and recalling that $s\leq k-1$, we get
	\begin{equation}\label{eqeqmanyStable_final}
	f(\rho)\leq-\frac{k-s}{3k}\ln k+O_k(1/k)+f(\id)\leq f(\id)-\frac{k-s}{3k}\ln k+O_k(1/k)<f(\id).
	\end{equation}
Finally, we calculate $f(\id)=\ln k+\frac d2\ln(1-1/k)=\frac12f(\bar\rho)$.
Since $f(\bar\rho)>0$ (by \Prop~\ref{Prop_first}), we conclude that $f(\id)<f(\bar\rho)$.
Thus, the assertion follows from~(\ref{eqeqmanyStable_final}).

\section{The Laplace method}\label{Sec_sufficient}

\noindent{\em In this section we keep the assumptions of 
\Prop~\ref{Prop_second} and the
	notation introduced in \Sec~\ref{sec:outline}.}

\smallskip
\noindent
In this section we prove \Prop~\ref{Prop_sufficient}.
Recalling that $\cR=\cR_{n,k}$ is the (discrete) set of overlap matrices, let
	$$Z_{\rho',\mathrm{\good}}=\abs{\cbc{(\sigma,\tau)\in\cB\times\cB:\text {$\sigma,\tau$ are \good\ $k$-colorings of $\gnm$ and $\rho(\sigma,\tau)=\rho'$}}}
		\qquad\text{for }\rho'\in\cR.$$
Then we can cast the second moment as
	\begin{equation}\label{eqdecomp}
	\Erw\brk{\Zkg^2}=\sum_{\rho\in\cR}\Erw\brk{\Zrg}.
	\end{equation}
Because any \good\ $k$-coloring is balanced, Fact~\ref{Lemma_frho} yields
	\begin{equation}\label{eqfrho}
	\Erw\brk{\Zrg}\leq\Erw\brk{\Zrb}\leq O(n^{(1-k^2)/2})\cdot\exp(n\cdot f(\rho))\qquad\mbox{uniformly for }\rho\in\cR.
	\end{equation}
By Taylor-expanding $f$ around $\bar\rho$, we can estimate the contribution to the sum~(\ref{eqdecomp})
resulting from $\rho$ near~$\bar\rho$.

\begin{lemma}\label{Lemma_central}
There exist $C=C(k)>0$ and $\eta=\eta(k)>0$ such that with $\cR_0=\cbc{\rho\in\cR:\norm{\rho-\bar\rho}_2<\eta}$ we have
	$$\sum_{\rho\in\cR_0}\Erw\brk{\Zrg}\leq C\cdot\Erw[\Zkg]^2.$$
\end{lemma}
\begin{proof}
By construction, we have $\sum_{i,j=1}^k\rho_{ij}=k$ for all $\rho\in\cR$.
Therefore, we can parameterize $\cR$ as follows.
Let
	\begin{eqnarray*}
	\cL:\brk{0,1}^{k^2-1}\ra\brk{0,1}^{k^2},&&\hat\rho=(\hat\rho_{ij})_{(i,j)\in\brk k^2\setminus\cbc{(k,k)}}\mapsto\cL(\hat\rho)=(\cL_{ij}(\hat\rho))_{i,j\in\brk k},\mbox{ where}\\
	\cL_{ij}(\hat\rho)=\hat\rho_{ij}\mbox{ for $(i,j)\neq(k,k)$}&\mbox{and}&\cL_{kk}(\hat\rho)=k-\sum_{(i,j)\neq(k,k)}\hat\rho_{ij}.
	\end{eqnarray*}
Moreover, let $\hat\cR=\cL^{-1}(\cR)\mbox{ and }\tilde\rho=\cL^{-1}(\bar\rho).$

We compute the Hessian of $f\circ\cL=H\circ \cL+E\circ\cL$ at $\tilde\rho$.
A direct calculation yields for $(a,b)\neq(i,j)$
	\begin{equation}\label{eqcentralNew1}
	\frac\partial{\partial\hat\rho_{ij}}H\circ\cL(\hat\rho)\big|_{\hat\rho=\tilde\rho}=0,\quad
	\frac{\partial^2}{\partial\hat\rho_{ij}^2}H\circ\cL(\hat\rho)\big|_{\hat\rho=\tilde\rho}=-2,\quad
	\frac{\partial^2}{\partial\hat\rho_{ij}\partial\hat\rho_{ab}}H\circ\cL(\hat\rho)\big|_{\hat\rho=\tilde\rho}=-1.
	\end{equation}
Furthermore,
	$$
	\frac\partial{\partial\hat\rho_{ij}}\norm{\cL(\hat\rho)}_2^2\big|_{\hat\rho=\tilde\rho}=0,\quad
	\frac{\partial^2}{\partial\hat\rho_{ij}^2}\norm{\cL(\hat\rho)}_2^2\big|_{\hat\rho=\tilde\rho}=4,\quad
	\frac{\partial^2}{\partial\hat\rho_{ij}\partial\hat\rho_{ab}}\norm{\cL(\hat\rho)}_2^2\big|_{\hat\rho=\tilde\rho}=2.
	$$
Thus, by the chain rule
	\begin{equation}\label{eqcentralNew2}
	\frac\partial{\partial\hat\rho_{ij}}E\circ\cL(\hat\rho)\big|_{\hat\rho=\tilde\rho}=0,\quad
	\frac{\partial^2}{\partial\hat\rho_{ij}^2}E\circ\cL(\hat\rho)\big|_{\hat\rho=\tilde\rho}=\frac{2d}{k^2(1-1/k)^2},\quad
	\frac{\partial^2}{\partial\hat\rho_{ij}\partial\hat\rho_{ab}}E\circ\cL(\hat\rho)=
		\frac{d}{k^2(1-1/k)^2}.
	\end{equation}
Combining~(\ref{eqcentralNew1}) and~(\ref{eqcentralNew2}), we see that the first derivative
of $f\circ\cL$ at the point $\tilde\rho$ vanishes, and that the Hessian is
	\begin{equation}\label{eqcentralNew3}
	D^2f\circ\cL(\hat\rho)|_{\hat\rho=\tilde\rho}=-\bc{1-\frac{d}{k^2(1-1/k)^2}}\cdot\bc{\id+\vecone},
	\end{equation}
where $\vecone$ denotes the matrix with all entries equal to one and $\id$ is the identity matrix.

As $\id$ is positive definite, $\vecone$ is positive semidefinite and $d/(k^2(1-1/k)^2)=O_k(\ln k/k)<\frac12$,
	(\ref{eqcentralNew3}) shows that the Hessian is negative definite at $\tilde\rho$.
In fact, by continuity there exist numbers $\tilde\eta,\tilde\xi>0$ independent of $n$ such that the largest eigenvalue of $D^2f\circ\cL$ is smaller than $-\tilde\xi$ at all points
$\hat\rho$ such that $\|\hat\rho-\tilde\rho\|_2<\tilde\eta$.
Further, because $\cL$ is linear there is an $n$-independent $\eta>0$ such that for all $\rho\in\cR_0=\cbc{\rho\in\cR:\norm{\rho-\bar\rho}_2<\eta}$
we have $\|\cL^{-1}(\rho)-\tilde\rho\|_2<\tilde\eta$.
Hence, by Taylor's formula there is a number $\xi>0$ that does not depend on $n$ such that
	\begin{equation}\label{eqcentral3}
	f\circ\cL(\hat\rho)\leq f(\bar\rho)-\xi\sum_{(i,j)\neq(k,k)}(\hat\rho_{ij}-1/k)^2\qquad\mbox{for all }\hat\rho\in \hat\cR_0=\cL^{-1}(\cR_0).
	\end{equation}
Combining~(\ref{eqfrho}) and~(\ref{eqcentral3}), we obtain
	\begin{eqnarray}\nonumber
	\sum_{\rho\in\cR_0}\Erw\brk{\Zrg}
		&\leq&\exp\bc{f(\bar\rho)n}\cdot
			O(n^{(1-k^2)/2})\sum_{\hat\rho\in\hat\cR_0}\exp\brk{-n\cdot\xi\sum_{(i,j)\neq(k,k)}(\hat\rho_{ij}-1/k)^2}\nonumber\\
		&\leq&\exp\bc{f(\bar\rho)n}\cdot O(1)\int_{\RR^{k^2-1}}\exp\brk{-\xi\sum_{(i,j)\neq(k,k)}(\hat z_{ij}-1/k)^2}d\hat z\nonumber\\
		&\leq&\exp\bc{f(\bar\rho)n}\cdot O(1)\brk{\int_{-\infty}^\infty\exp\brk{-\xi z^2}d z}^{k^2-1}
		=O(1)\cdot\exp\bc{f(\bar\rho)n}.
			\label{eqcentral4}
	\end{eqnarray}
Finally, a direct calculation shows that $f(\bar\rho)=2(\ln k+\frac d2\ln(1-1/k))$, whence $\exp\bc{f(\bar\rho)n}=O(k^n(1-1/k)^m)^2$
	(as $m=\lceil dn/2\rceil$).
Thus, the assertion follows from \Prop~\ref{Prop_first} and~(\ref{eqcentral4}).
\end{proof}

To estimate the contribution of $\rho\not\in\cR_0$, we decompose $\cR\setminus\cR_0$ into three subsets:
	\begin{eqnarray*}
	\cR_1&=&\cbc{\rho\in\cR\setminus\cR_0:\rho\mbox{ fails to be separable}},\\
	\cR_2&=&\cbc{\rho\in\cR\setminus(\cR_0\cup\cR_1):\mbox{for each $i$ there is $j$ such that $\rho_{ij}>0.51$}},\\
	\cR_3&=&\cR \setminus \bc{\cR_0 \cup \cR_1 \cup \cR_2}.
	\end{eqnarray*}
Condition {\bf T2} from Definition~\ref{Def_good} directly implies that
	\begin{equation}\label{eqProp_sufficient1}
	\Erw\brk{\Zrg}=0\qquad\mbox{for all }\rho\in\cR_1.
	\end{equation}
With respect to $\cR_2$, we have

\begin{lemma}\label{Lemma_R2}
There is a number $C=C(k)>0$ such that
	$\sum_{\rho\in\cR_2}\Erw\brk{\Zrg}\leq C\cdot\Erw[\Zkg]^2.$
\end{lemma}
\begin{proof}
Let $\cR_2'$ be the set of all $k$-stable $\rho'\in\cR$ (i.e., $\rho_{ii}'>0.51$ for all $i\in\brk k$).
Because we restrict ourselves to balanced $k$-colorings, the row and column sums of each matrix $\rho\in\cR$ are $1+O(n^{-1/2})$.
Hence, for any matrix $\rho\in\cR$ there is at most one entry greater than $0.51$ in each row or column.
Thus, suppose that $\sigma,\tau$ are \good\ $k$-colorings of $\gnm$ such that $\rho(\sigma,\tau)\in\cR_2$.
Then each row and each column of $\rho(\sigma,\tau)$ have {\em exactly} one entry that is greater than $0.51$.
Therefore, there exists a permutation $\pi:\brk k\ra\brk k$ such that $\sigma,\pi\circ\tau$ are two colorings such that $\rho(\sigma,\pi\circ\tau)\in\cR_2'$.
Consequently,
	\begin{equation}\label{eqProp_sufficient2}
	\sum_{\rho\in\cR_2}\Erw\brk{\Zrg}\leq k!\sum_{\rho\in\cR_2'}\Erw\brk{\Zrg}.
	\end{equation}
Further, if $\sigma,\tau$ are $k$-colorings such that $\rho(\sigma,\tau)\in\cR_2'$, then $\tau\in\cC(\sigma)$ by the very definition of the cluster $\cC(\sigma)$.
Therefore, by the linearity of expectation and Bayes' formula, we have
	\begin{equation}\label{eqProp_sufficient3}
	\sum_{\rho\in\cR_2'}\Erw\brk{\Zrg}=\sum_{\sigma\in\cB}
		\Erw\brk{\cC\bc\sigma|\sigma\mbox{ is a \good\ $k$-coloring}}\cdot\pr\brk{\sigma\mbox{ is a \good\ $k$-coloring}}
	\end{equation}
Now, if $\sigma$ is a \good\ $k$-coloring, then by {\bf T3} we know that $\cC\bc\sigma\leq\Erw[\Zkb]$ with certainty.
Thus, (\ref{eqProp_sufficient2}) yields
	\begin{eqnarray}\nonumber
	\sum_{\rho\in\cR_2'}\Erw\brk{\Zrg}&\leq&\Erw\brk{\Zkb}\sum_{\sigma\in\cB}
		\pr\brk{\sigma\mbox{ is a \good\ $k$-coloring}}\leq\Erw\brk{\Zkb}\cdot\Erw\brk{\Zkg}\\
		&\leq&(1+o(1))\Erw\brk{\Zkg}^2\qquad\mbox{[by \Prop~\ref{Prop_first}].}
			\label{eqProp_sufficient4}
	\end{eqnarray}
Combining~(\ref{eqProp_sufficient2}) and~(\ref{eqProp_sufficient4}), we get
	$\sum_{\rho\in\cR_2}\Erw\brk{\Zrg}\leq O(\Erw\brk{\Zkg}^2),$
as claimed.
\end{proof}

To bound the contribution of $\rho\in\cR_3$,
 we need the following observation.

\begin{lemma}\label{Lemma_fluctuations}
There is a number $C=C(k)>0$ such that
for any $\rho\in\cR$ there is $\rho'\in\Birk$ with $\norm{\rho-\rho'}_2< C/\sqrt n$. 
\end{lemma}
\begin{proof}
Let $\rho\in\cR$.
By construction, we have $\sum_{i,j}\rho_{ij}=k$.
Hence, while there is $i\in[k]$ such that the row sum is $\sum_{j}\rho_{ij}=1+\alpha>1$,
there must be another row $l$ such that $\sum_{j}\rho_{lj}=1-\alpha'<1$.
Thus, by replacing row $i$ by $(1-\alpha'')\rho_i$ and row $l$ by $\rho_l+\alpha''\rho_i$ for some suitable $\alpha''\leq 2k/\sqrt n$,
we can ensure that at least one of the row sums is one.
After at most $k-1$ steps, we thus obtain a stochastic matrix $\rho''$ such that $\|\rho-\rho''\|_2=2k^3/\sqrt n$.
Repeating the same operation for the columns yields the desired doubly-stochastic $\rho'$.
\end{proof}

\begin{lemma}\label{Lemma_R3}
If $f(\rho)<f(\bar\rho)$ for any $\rho\in\Dg\setminus\cbc{\bar\rho}$, then
	$\sum_{\rho\in\cR_3}\Erw\brk{\Zrg}\leq \Erw[\Zkg]^2.$
\end{lemma}
\begin{proof}
Let $\eta>0$ be the number from \Lem~\ref{Lemma_central} and let $\cD'$ be the set of all $\rho\in\Dg$ such that $\norm{\rho-\bar\rho}_2\geq\eta/2$.
The set $\cD'$ is compact.
Hence, our assumption that $f(\rho)<f(\bar\rho)$ for any $\rho\in\Dg\setminus\cbc{\bar\rho}$ implies that
there exists a number $\gamma>0$ (independent of $n$) such that
	\begin{equation}\label{eqProp_sufficient5}
	\max_{\rho\in\cD'}f(\rho)<f(\bar\rho)-\gamma.
	\end{equation}
In fact, because the function $f$ is uniformly continuous on $[0,1]^{k^2}$, there is $0<\delta<\eta/3$ such that
	\begin{equation}\label{eqProp_sufficient6}
	\max_{\rho\in\cD''}f(\rho)<f(\bar\rho)-\gamma/2,\qquad\mbox{where }\quad\cD''=\{\rho\in\brk{0,1}^{k^2}:\mbox{there is $\rho'\in\cD'$ with }\norm{\rho-\rho'}_2<\delta\}.
	\end{equation}

We claim that $\cR_3\subset\cD''$.
Indeed, any $\rho\in\cR_3$ satisfies $\norm{\rho-\bar\rho}_2\geq\eta$ (as otherwise $\rho\in\cR_0$), is separable (as otherwise $\rho\in\cR_1$),
	and is not stable (as otherwise $\rho\in\cR_2$).
Moreover, by \Lem~\ref{Lemma_fluctuations} there is a doubly-stochastic $\rho'$ such that $\norm{\rho-\rho'}_2<C/\sqrt n$.
However, this matrix $\rho'$ may or may not be separable and/or stable.
To rectify this, we form a convex combination between $\rho'$ and a suitable doubly-stochastic matrix.
More precisely, suppose that the matrix $\rho$ has precisely $l<k-1$ entries that are greater than $0.51$.
Each row and each column contain at most one such entry (as $\rho\in\cB$).
Thus, we may assume without loss of generality that $\rho_{11},\ldots,\rho_{ll}>0.51$.
Now, let $\rho''$ be the doubly-stochastic matrix with $\rho''_{11}=\cdots=\rho_{ll}''=1$ and $\rho_{ij}''=(k-l)^{-1}$ for $i,j>l$.
If $\beta>0$ is a small enough number, then $\rho'''=(1-\beta)\rho'+\beta\rho''\in\cD'$ and $\norm{\rho-\rho'''}_2<\delta$.
Thus, $\rho\in\cD''$.

As $\cR_3\subset\cD''$, (\ref{eqProp_sufficient6}) yields
	\begin{equation}\label{eqProp_sufficient7}
	\max_{\rho\in\cR_3}f(\rho)<f(\bar\rho)-\gamma/2.
	\end{equation}
Thus, (\ref{eqfrho}) implies
	\begin{eqnarray}\nonumber
	\sum_{\rho\in\cR_3}\Erw[\Zrg]&\leq&
		|\cR_3|\exp(n(f(\bar\rho)-\gamma/2))
		\leq \abs\cR\exp(n(f(\bar\rho)-\gamma/2))\\[-3mm]
		&\leq&n^{k^2}\exp(n(f(\bar\rho)-\gamma/2))\leq\exp(n(f(\bar\rho)-\gamma/3)).
			\label{eqProp_sufficient8}
	\end{eqnarray}
Upon direct inspection, we find
	$f(\bar\rho)=2(\ln k+\frac d2\ln(1-1/k)).$
Recalling that $m=\lceil dn/2\rceil$, we thus obtain from \Prop~\ref{Prop_first}
	\begin{equation}		\label{eqProp_sufficient10}	
	\exp(n(f(\bar\rho)-\gamma/3))\leq \Erw\brk{\Zkg}^2\cdot\exp(-\gamma n/4).
	\end{equation}
Combining~(\ref{eqProp_sufficient8}) and~(\ref{eqProp_sufficient10}), we obtain
	$$
	\sum_{\rho\in\cR_3}\Erw[\Zrg]=\Erw\brk{\Zkg}^2\cdot n^{k^2}\exp(-\gamma n/4)\leq\Erw\brk{\Zkg}^2,
	$$
thereby completing the proof.
\end{proof}

\noindent
Finally, \Prop~\ref{Prop_sufficient} follows from~(\ref{eqProp_sufficient1}) and \Lem s \ref{Lemma_central}, \ref{Lemma_R2} and~\ref{Lemma_R3}.

\appendix

\section{Proof of \Lem~\ref{Lemma_P}}\label{apx:PropProof} 

\noindent
{\em Throughout this section, we assume that $2k\ln k-\ln k-2\leq d\leq2k\ln k$.
	In addition, we fix some $\sigma\in\cB$ and we let $V_i=\sigma^{-1}(i)$ for $i=1,\ldots,n$.}

To simplify the calculations we consider the following variant of the planted model.
Given $\sigma$, $n$ and $q\in(0,1)$, we let $\cG(n,q,\sigma)$ be the random graph
in which any two vertices $v,w$ with $\sigma(v)\neq\sigma(w)$ are adjacent with probability $p$ independently.
The following observation relates this model to the planted model $G(n,m,\sigma)$ from \Lem~\ref{Lemma_P}.

\begin{fact}\label{Lemma_equivalentModel}
Given $\sigma\in\cB$, let $p$ be such that the expected number of edges in $\cG(n,p,\sigma)$ is equal to $m=\lceil dn/2\rceil$.
There is a number $C=C(k)>0$ such that
	$$\pr\brk{G(n,m,\sigma)\in\cA}\leq C\sqrt n\cdot\pr\brk{\cG(n,p,\sigma)\in\cA}\qquad\mbox{ for any event $\cA$}.$$
\end{fact}
\begin{proof}
By the choice of $p$, the number $e(\cG(n,p,\sigma))$ of edges of the random graph $G(n,p,\sigma)$ has a binomial distribution with mean
	\begin{equation}\label{eqLemma_equivalentModel}
	p\brk{\bink n2-\sum_{i=1}^k\bink{|V_i|}{2}}=m.
	\end{equation}
Hence, Stirling's formula shows that for some number $C=C(k)>0$ we have
	$\pr\brk{e(\cG(n,p,\sigma))=m}\geq (C\sqrt n)^{-1}$.
Further, given that $e(\cG(n,p,\sigma))=m$, the distribution of the random graph $\cG(n,p,\sigma))$ is identical to that of $G(n,m,\sigma)$.
Thus, for any event $\cA$
	\begin{eqnarray*}
	\pr\brk{G(n,m,\sigma)\in\cA}
		&\leq&\frac{\pr\brk{\cG(n,p,\sigma)\in\cA}}{\pr\brk{e(\cG(n,p,\sigma))=m}}
		\leq C\sqrt n\cdot \pr\brk{\cG(n,p,\sigma)\in\cA},
	\end{eqnarray*}
as claimed.
\end{proof}

\noindent
{\em From here on out, we fix $\sigma\in\cB$ and choose $p\in(0,1)$ such that the expected number of edges in $\cG(n,p,\sigma)$ is equal to $m$;
because $\sigma$ is balanced, (\ref{eqLemma_equivalentModel}) implies that
	\begin{equation}\label{eqplantedp}
	p\sim\frac{k}{k-1}\cdot\frac dn.
	\end{equation}}%
In the following, we are going to show that the properties {\bf P1--P4} are satisfied in $\cG(n,p,\sigma)$ with
probability $1-O(1/n)$.
Then Fact~\ref{Lemma_equivalentModel} readily implies that they hold in $G(n,m,\sigma)$ \whp

The following instalment of the Chernoff bound will prove useful.

\begin{lemma}[\cite{JLR}]\label{Lemma_Chernoff}
Let $\varphi(x)=(1+x)\ln(1+x)-x$.
Let $X$ be a binomial random variable with mean $\mu>0$.
Then for any $t>0$,
	\begin{eqnarray*}
	\pr\brk{X>\Erw\brk X+t}\leq\exp(-\mu\cdot\varphi(t/\mu)),&&
	\pr\brk{X<\Erw\brk X-t}\leq\exp(-\mu\cdot\varphi(-t/\mu)).
	\end{eqnarray*}
In particular, for any $t>1$ we have
	$\pr\brk{X>t\mu}\leq\exp\brk{-t\mu\ln(t/\eul)}.$
\end{lemma}

\subsection{Proof of P1.}
We may assume $i=1$ without loss of generality.
Let $0.509\leq\alpha\leq1-k^{-0.499}$ and let $S\subset V_1$ be a set of size $|S|=\alpha n/k$.
Because in $\cG(n,p,\sigma)$ edges occur independently,
for any $v\in V\setminus V_1$ the number of neighbors of $v$ in $S$ has distribution $\Bin(\alpha n/k,p)$.
Hence, as $\sigma$ is balanced the number $X_S$ of $v\in  V\setminus V_1$ with no neighbor in $S$ has a binomial distribution
with mean $n(1-1/k+o(1))(1-p)^{\alpha n/k}$.
Our assumption on $d$ and~(\ref{eqplantedp}) imply that $(1-p)^{\alpha n/k}\leq\exp\brk{-\alpha np/k}\leq 2k^{-2\alpha}$.
Thus,
	\begin{equation}			\label{eqBlowupX3}
	\Erw\brk{X_S}\leq(1+o(1))  n(1-1/k)\cdot 2k^{-2\alpha}.
	\end{equation}
Consequently, by \Lem~\ref{Lemma_Chernoff}
	\begin{eqnarray}
	\pr\brk{X_S\geq(1-\alpha)n/k-n^{2/3}}
		&\leq&\exp\brk{-(1-\alpha+o(1))\frac nk\cdot\ln\bc{\frac{1-\alpha}{2\eul}\cdot k^{2\alpha-1}}}.
			\label{eqBlowupX2}
	\end{eqnarray}
By comparison, because $\sigma$ is balanced, for a given $\alpha$ the number of ways to choose $S$ is
	\begin{eqnarray}\label{eqBlowupX1}
	\bink{(1+o(1))n/k}{(1-\alpha+o(1))n/k}&\leq&\bcfr{\eul }{1-\alpha}^{(1-\alpha+o(1))\frac nk}=\exp\brk{\frac{n}k(1-\alpha+o(1))\bc{1-\ln(1-\alpha)}}.
	\end{eqnarray}
Let us call $S$ \emph{$\alpha$-bad} if $X_S\geq (1-\alpha)\frac nk-n^{2/3}$.
Combining~(\ref{eqBlowupX3}), (\ref{eqBlowupX2}) and~(\ref{eqBlowupX1}) and taking the union bound over $S\subset V_1$ with $|S|=\alpha n/k$, we obtain
	\begin{eqnarray*}
	\pr\brk{\mbox{there is an $\alpha$-bad $S$}}&\leq&\exp\brk{\frac{(1-\alpha)n}k\cdot\bc{1-\ln(1-\alpha)-\ln\bc{\frac{1-\alpha}{2\eul}\cdot k^{2\alpha-1}}}+o(n)}.
	\end{eqnarray*}
To complete the proof of {\bf P1}, we are going to show that the right hand side is $\exp(-\Omega(n))$.

Thus, we need to estimate
	\begin{eqnarray*}
	1-\ln(1-\alpha)-\ln\bc{\frac{1-\alpha}{2\eul}\cdot k^{2\alpha-1}}&=&\ln\bc{\frac{2\eul^2}{(1-\alpha)^2}k^{1-2\alpha}}.
	\end{eqnarray*}
This is negative iff
	\begin{equation}\label{eqLemma_blowup1}
	\exp\brk{\bc{\frac12-\alpha}\ln k}<\frac{1-\alpha}{\sqrt 2\eul}.
	\end{equation}
By convexity, the exponential function on the l.h.s.\ and the linear function on the r.h.s.\ intersect at most twice,
and between these two intersections the linear function is greater.
Further, an explicit calculation verifies that the r.h.s.\ of~(\ref{eqLemma_blowup1}) is larger than the l.h.s.\ at both $\alpha=0.509$ and
$\alpha=1-k^{-0.499}$.
Thus, (\ref{eqLemma_blowup1}) is true in the entire range $0.509<\alpha<1-k^{-0.499}$.
\qed

\subsection{Proof of P2}
In $\cG(n,p,\sigma)$, for each vertex $v\in V\setminus V_i$ the number of neighbors of $v$ in $V_i$ has distribution $\Bin(|V_i|,p)$.
Due to~(\ref{eqplantedp}) and because $\sigma$ is balanced, the mean is $\lambda=|V_i|p\sim\frac nk p>2\ln k$.
Hence, by Stirling's formula the probability that $v$ has fewer than $15$ neighbors in $V_i$ is
	$q\leq 2\lambda^{14}\exp(-\lambda)\leq2k^{-2}\ln^{14}k.$
Further, because the event of having fewer than $15$ neighbors in $V_i$ occurs independently for all $v\in V\setminus V_i$,
	the total number $Y_i$ of such vertices has a binomial distribution $\Bin(|V\setminus V_i|,q)$.
As $\sigma$ is balanced, the mean is
	$|V\setminus V_i|q\leq (1-1/k+o(1))n\cdot q\leq 3k^{-2}\ln^{14}k.$
Since we chose $\kappa=k^{-1}\ln^{20}k$,
a straightforward application of \Lem~\ref{Lemma_Chernoff} (the Chernoff bound) implies that
	$\pr\brk{Y_i>\frac{\kappa n}{3k}}\leq\exp(-\Omega(n)),$
as desired.\qed

\subsection{Proof of P3}\label{Sec_P3}
Let $0<\alpha<k^{-4/3}$ and let $S\subset V$ of size $|S|=\alpha n$.
The number $e(S)$ of edges spanned by $S$ in $\cG(n,p,\sigma)$ is stochastically dominated by a random variable with
distribution $\Bin(\bink{\alpha n}2,p)$.
For any two vertices $v,w\in S$ are connected with probability at most $p$ in $\cG(n,p,\sigma)$ (as the probability is exactly $p$ if $\sigma(v)\neq\sigma(w)$
and $0$ otherwise).
Thus,
	$$\pr\brk{e(S)\geq 5|S|}\leq\pr\brk{\Bin\bc{\bink{\alpha n}2,p}\geq 5\alpha n}\leq\bink{\bink{\alpha n}{2}}{5\alpha n}p^{5\alpha n}.$$
Now, let $X_\alpha$ be the number of sets $S$ of size $|S|=\alpha n$ such that $e(S)\geq 5|S|$.
Let $d'=pn\sim\frac{dk}{k-1}$.
By the union bound,
	\begin{equation}\label{eqP3proof}
	\pr\brk{X_\alpha>0}\leq\bink{n}{\alpha n}\bink{\bink{\alpha n}{2}}{5\alpha n}p^{5\alpha n}
		\leq \bcfr{\eul}{\alpha}^{\alpha n}\bcfr{\eul\alpha d'}{10}^{5\alpha n}
		\leq\brk{\eul\bcfr{\eul d'}{10}^5\alpha^4}^{\alpha n}.
	\end{equation}
Further, let $X=\sum_\alpha X_\alpha$, where the sum ranges over $0<\alpha<k^{-4/3}$ such that $\alpha n$ is an integer.
Then~(\ref{eqP3proof}) implies together with the assumption that $\alpha<k^{-4/3}$ that
	$$\pr\brk{X>0}\leq\sum_\alpha \brk{\eul\bcfr{\eul d'}{10}^5\alpha^4}^{\alpha n}=O(1/n).$$
Thus, the probability that there is a set violating {\bf P3} is $O(1/n)$.
\qed

\subsection{Proof of P4}
We start by estimating the size of the core;
	the proof of the following proposition draws on arguments developed in~\cite{Barriers,AK}.

\begin{proposition}\label{Prop_Core_size}
With probability $1-\exp\bc{-\Omega(n)}$, the core of $\cG(n,p,\sigma)$ contains
$(1-\tilde O_k(k^{-1}))n$ vertices.
\end{proposition}

The proof of Proposition \ref{Prop_Core_size} is constructive:
basically, we iteratively remove vertices of that have too few neighbors of some color other than their own among the remaining vertices.
More precisely, we consider the following process.
For a vertex $v$ and a set $S$ of vertices let $e(v,S)$ denote the number of neighbors of $v$ in $S$ in $\cG(n,p,\sigma)$.
\begin{description}
  \item[CR1] For $i ,j\in\brk k$, $i\neq j$, let $W_{ij}=\{v \in V_i:e(v,V_j) < 300\}$, $W_{ii}=\emptyset$, $W_{i}=\cup_{j=1}^k W_{ij}$, and $W=\cup_{i=1}^k W_i$.
  \item[CR2] For $i \ne j$, let $U_{ij}=\{v \in V_i
  									: e(v,W_j) > 100\}$ and $U = \cup_{i \ne j} U_{ij}$.
  \item[CR3] Set $Z^{(0)} = U$ and repeat the following for $i\geq0$:\\
		$\mbox{\ }\qquad\bullet$	if there is  $v \in V \setminus Z^{(i)}$ such that $e(v,Z^{\bc i})\geq100$, pick one such $v$ and let $Z^{(i+1)}=Z^{(i)} \cup \{v\}$;\\
		$\mbox{\ }\qquad\bullet$		otherwise, let $Z^{(i+1)}=Z^{(i)} \cup \{v\}$.
\end{description}
Let $Z=\cup_{i\geq 0}Z^{\bc i}$ be the final set resulting from {\bf CR3}.
By construction, the set $V\setminus(W\cup Z)$ is contained in the core.
To complete the proof of \Prop~\ref{Prop_Core_size}, we bound the sizes of $W$, $U$ and $Z$ (Lemmas \ref{Lem_size_of_W}, \ref{Lem_size_of_U} and \ref{Lem_size_of_Z}).

\begin{lemma}\label{Lem_size_of_W}
With probability at least $1-\exp\bc{-\Omega(n)}$ we have $|W_{ij}|\leq\tilde O_k(k^{-3})$ for any $i,j$.
\end{lemma}
\begin{proof}
Fix $i,j$, $i \ne j$.
Due to the independence of the edges in $\cG(n,p,\sigma)$, for any $v \in V_i$ the number $e(v,V_j)$ of neighbors in $V_j$ has distribution $\Bin(|V_j|,p)$.
As $\sigma$ is balanced, (\ref{eqplantedp}) shows that the mean is $\mu=|V_j|p\geq2\ln k$.
Using the Chernoff bound (Lemma \ref{Lemma_Chernoff}), 
 we obtain
	$\pr\left[|e(v,V_j)| \le 300\right] \le \exp\left(-2\ln k +O_k\left(\ln \ln k\right)\right) = \tilde O_k(k^{-2}) .$
Hence, by the linearity of expectation and because $\sigma$ is balanced, $\EX[|W_{ij}|] \le \tilde O_k(k^{-2}) \cdot |V_i|= n\cdot \tilde O_k(k^{-3})$.
Further, once more due to the independence of the edges in $\cG(n,p,\sigma)$, $|W_{ij}|$ is a binomial random variable.
Thus, using the Chernoff bound once more (with, say, $t= k^{-4}n$),
we see that $\pr[|W_{ij}|\leq\tilde O_k(k^{-3})n]\geq1-\exp(-\Omega(n)),$ as required.
\end{proof}

\begin{lemma}\label{Lem_size_of_U} With probability at least $1-\exp\bc{-\Omega(n)}$ we have $|U| \le n/k^{30}$.
\end{lemma}
\begin{proof}
We define two sets whose union contains $U_{ij}$:
$$U'_{ij}=\{v \in V_i  : e(v,W_j\setminus W_{ji}) \ge 50\}, \quad  U''_{ij}=\{v \in V_i  : e(v,W_{ji}) \ge 50\}.$$
Thus, it suffices to bound the sizes of $U'_{ij}$, $U''_{ij}$ separately.

Let's start with $U'_{ij}$.
By construction, which vertices belong to $W_j\setminus W_{ji}$ is independent of the edges between color classes $V_i,V_j$.
Hence, for any $v\in V_i$ the number $e(v,W_i\setminus W_{ji})$ has distribution $\Bin(|W_i\setminus W_{ji}|,p)$.
Thus, 
	$$\Erw\brk{e(v,W_i\setminus W_{ji})\,\big|\,|W_j\setminus W_{ji}| \le n\cdot\tilde O_k(k^{-2})}\leq pn \cdot\tilde O_k(k^{-2})\leq\tilde O_k(k^{-1}).$$
Therefore, the Chernoff bound (Lemma \ref{Lemma_Chernoff}) applied with, say, $t=45$ yields
	\begin{equation}\label{eqLemma_size_of_U_1}
	\pr\brk{v\in U_{ij}'\,\big|\,|W_j\setminus W_{ji}| \le n\cdot\tilde O_k(k^{-2})}\leq \tilde O_k(k^{-45}).
	\end{equation}
Once more due to the independence of the edges in $\cG(n,p,\sigma)$, the events $v\in U_{ij}'$ are mutually independent for $v\in V_i$.
by \Lem~\ref{Lem_size_of_W}, this event occurs with probability $1-\exp(-\Omega(n))$.
In effect, given $|W_j\setminus W_{ji}| \le n\cdot\tilde O_k(k^{-2})$, $|U_{ij}'|$ has a binomial distribution.
Thus, (\ref{eqLemma_size_of_U_1}) implies together with the Chernoff bound (applied with, say, $t=k^{-100}n$) that
	\begin{equation}\label{eqLemma_size_of_U_2}
	\pr\brk{|U_{ij}'|>nk^{-40}\,\big|\,|W_j\setminus W_{ji}| \le n\cdot\tilde O_k(k^{-2})}\leq \exp(-\Omega(n)).
	\end{equation}
Further, \Lem~\ref{Lem_size_of_W} implies that
	$\pr[|W_j\setminus W_{ji}| \le n\cdot\tilde O_k(k^{-2})]\geq1-\exp(-\Omega(n))$.
Combining this bound with (\ref{eqLemma_size_of_U_2}), we obtain
	\begin{equation}\label{eqLemma_size_of_U_2}
	\pr\brk{|U_{ij}'|>nk^{-40}}\leq \exp(-\Omega(n)).
	\end{equation}

With respect to $U''_{ij}$, we observe the following.
Given that $w\in W_{ji}$, we know that $w$ has fewer than $300$ neighbors in $V_i$.
But the fact that $w\in W_{ji}$ has no implications as to {\em which} $v\in V_i$ vertex $w$ is adjacent to.
Thus, given that $w\in W_{ji}$ and given $e(w,V_i)$, the actual set of neighbors of $w$ in $V_i$ is a random subset of $V_i$
of size $e(w,V_i)\leq300$.
In fact, these sets are mutually independent for all $w\in W_{ji}$.
Thus, we can bound $|U_{ij}''|$ by means of the following balls and bins experiment:
	let us think of the vertices in $V_i$ as bins.
Then each vertex $w\in W_{ji}$ tosses $300$ balls randomly into the bins $V_i$, independently of all other vertices in $W_{ji}$.
In this experiment, let $\cX$ be the set of $v\in V_i$ that receive at least 50 balls.
Then $|U_{ij}''|$ is dominated by $|\cX|$ stochastically.

Now, consider one $v\in V_i$.
Given $|W_{ji}|$, the number of balls that land in $v$ has distribution $\Bin(300|W_{ji}|,|V_i|^{-1})$.
Therefore, the Chernoff bound yields
	$$\pr\brk{v\in\cX\big| |W_{ji}|\leq n\cdot\tilde O_k(k^{-3})}\leq
		\pr\brk{\Bin(\tilde O_k(k^{-3})n,(1+o(1))k/n)\geq50}\leq k^{-45}.$$
Hence, by the linearity of expectation $\Erw|\cX|\leq nk^{-45}$.
Hence, Azuma's inequality yields
	$$\pr\brk{|U_{ij}''|> n k^{-40}\big| |W_{ji}|\leq n\cdot\tilde O_k(k^{-3})}\leq
	\pr\brk{|\cX|> n k^{-40}\big| |W_{ji}|\leq n\cdot\tilde O_k(k^{-3})}\leq
		\exp(-\Omega(n)).
	$$
Thus, \Lem~\ref{Lem_size_of_W} implies
	\begin{equation}\label{eqLemma_size_of_U_4}
	\pr\brk{|U_{ij}''|> n k^{-40}}\leq	\exp(-\Omega(n)).
	\end{equation}
Finally, the assertion follows from~(\ref{eqLemma_size_of_U_2}) and~(\ref{eqLemma_size_of_U_4}), with room to spare.
\end{proof}

\begin{lemma}\label{Lem_size_of_Z} With probability at least $1-\exp\bc{-\Omega(n)}$ we have $|Z| \le n/k^{29}$.
\end{lemma}
\begin{proof}
Lemma \ref{Lem_size_of_U} entails that with probability at least $1-\exp\bc{-\Omega(n)}$, $|U| \le n/k^{30}$.
Assume that this is indeed the case.
Further, suppose that $|Z\setminus U|\geq i^*=n/k^{30}$.
Let us stop the process {\bf CR3} at this point, and let $Z^*=Z^{(i^*)}$.
By construction, the graph induced on $S=U \cup Z^*$ spans at least $100i^*\geq50|S|$ edges, while $|S|\leq2k^{-30}n$.
Thus, the set $S$ violates condition {\bf P3}.
But since we saw in \Sec~\ref{Sec_P3} that {\bf P3} is satisfied with probability $1-\exp(-\Omega(n))$, the assertion follows.
\end{proof}

Now, \Prop~\ref{Prop_Core_size} is immediate from \Lem s~\ref{Lem_size_of_W}--\ref{Lem_size_of_Z}.
For a set $Y\subset V$ let us denote by $N(Y)$ the set of all vertices $v\in V$ that have a neighbor in $Y$ in $\cG(n,p,\sigma)$.
As a further step towards the proof of {\bf P4}, we establish

\begin{lemma}\label{Lem_expansion_of_Z}
With probability $1-\exp(-\Omega(n))$ the random graph $\cG(n,p,\sigma)$ has the following property.
	\begin{equation}
	\parbox{14cm}{Let $Y\subset V$ be a set of $|Y|\leq nk^{-29}$ vertices.
	Then $|N(Y)|\leq nk^{-20}$.}
	\end{equation}
\end{lemma}
\begin{proof}
Let $\alpha<k^{-29}$ be the largest number such that $\alpha n$ is an integer and let $q=1-(1-p)^{\alpha n}$.
For a set $Y\subset V$ with $|Y|=\alpha n$ the number of vertices $v\in V\setminus Y$ that have a neighbor in $Y$ in $\cG(n,p,\sigma)$
is stochastically dominated by $\Bin(n,q)$.
This is because for any vertex $y\in Y$ the probability that $v,y$ are adjacent is either $p$ (if $\sigma(v)\neq \sigma(y)$) or $0$ (if $\sigma(v)=\sigma(y)$).
Hence, observing that $p\leq \alpha np$ and using the Chernoff bound, we get
	\begin{equation}\label{eqLem_expansion_of_Z1}
	\pr\brk{|N(Y)\setminus Y|\geq nk^{-21}}\leq \pr\brk{\Bin(n,q)\geq nk^{-21}}\leq\exp(-nk^{-21}).
	\end{equation}
Now, let $X$ be the number of sets $Y$ with $|Y|=\alpha n$ such that $|N(Y)\setminus Y|\geq nk^{-21}$.
Together with the union bound, (\ref{eqLem_expansion_of_Z1}) shows
	\begin{eqnarray}\label{eqLem_expansion_of_Z2}
	\pr\brk{X>0}&\leq&\bink n{\alpha n}\exp(-nk^{-21})\leq
		\exp\brk{n\bc{\alpha(1-\ln\alpha)-k^{-21}}}\leq\exp(-\Omega(n));
	\end{eqnarray}
the last inequality follows because $\alpha(1-\ln\alpha)\leq 32k^{-29}\ln k$ for $0<\alpha<k^{-29}$.
Thus, we obtain from~(\ref{eqLem_expansion_of_Z2}) that $X_\alpha=0$ for all such $\alpha$ with probability $1-\exp(-\Omega(n))$.
If so, we see that any set $Y$ of size $|Y|\leq nk^{-29}$ satisfies
	$|N(Y)|\leq |Y|+|N(Y)\setminus Y|\leq n(k^{-29}+k^{-21})\leq nk^{-20},$
as claimed.
\end{proof}

\begin{corollary}\label{Cor_expansion_of_Z}
With probability $1-\exp(-\Omega(n))$ we have $|N(Z)|\leq n k^{-20}$.
\end{corollary}
\begin{proof}
This is immediate from \Lem s~\ref{Lem_size_of_Z} and~\ref{Lem_expansion_of_Z}.
\end{proof}

We define two sets of vertices, which capture the 1-free and 2-free vertices.
In what follows, when always let $i,j\in\brk k$, $i\neq j$.
Let $S_0$ be the set of vertices that have zero neighbors in some color class other than their own.
Moreover,
	$S_1 = \{ v\in V \setminus S_0 : \exists i,j \text{ s.t. } v \in V_i \text{ and } N(v)\cap V_j \subseteq W_j\}.$
By the construction of the core, we have

\begin{fact}\label{Fact_1-free}
If $v$ is $1$-free, then $v\in S_0\cup S_1\cup Z\cup N(Z)$.
\end{fact}

\noindent
We proceed by estimating the sizes of $S_0$, $S_1$.

\begin{lemma}\label{Lemma_S0}
With probability $1-\exp(-\Omega(n))$ we have $|S_0|\leq\frac nk$.
\end{lemma}
\begin{proof}
Consider a vertex $v\in V_i$.
The number $e(v,V_j)$ of neighbors of $V_i$ in $V_j$ has distribution $\Bin(|V_j|,p)$.
Since $\sigma$ is balanced, (\ref{eqplantedp}) yields $\pr\brk{e(v,V_j)=0}\leq(1-p)^{|V_j|}\leq k^{-2}$.
Thus, by the union bound,
	\begin{equation}\label{eqLemma_S0}
	\pr\brk{v\in S_0}\leq\sum_{j}\pr\brk{e(v,V_j)=0}\leq (k-1)k^{-2}.
	\end{equation}
Because the events $\{v\in S_0\}$ are mutually independent for all $v\in V_i$, the Chernoff bound and~(\ref{eqLemma_S0}) yield
	$\pr\brk{|S_0\cap V_i|>n/k^2}\leq\exp(-\Omega(n)).$
Taking the union bound over $i$ completes the proof.
\end{proof}

\begin{lemma}\label{Lemma_S1}
With probability $1-\exp(-\Omega(n))$ we have $|S_1|\leq\tilde O_k(k^{-2})n$.
\end{lemma}
\begin{proof}
Fix $i\neq j$.
The total number $e(V_i,V_j)$ of edges joining $V_i$ and $V_j$ in $\cG(n,p,\sigma)$ has distribution $\Bin(|V_i\times V_j|,p)$.
Because $\sigma$ is balanced, the Chernoff bound yields
	\begin{equation}\label{eqLemma_S1_1}
	\pr\brk{e(V_i,V_j)\geq \frac12k^{-2}n^2p}\geq1-\exp(-\Omega(n)).
	\end{equation}
In addition, we claim that the number $e(V_i,W_j)$ of $V_i$-$W_j$-edges satisfies
	\begin{equation}\label{eqLemma_S1_2}
	\pr\brk{e(V_i,W_j)\leq\tilde O_k(k^{-3})n^2p}\geq1-\exp(-\Omega(n)).
	\end{equation}
Indeed, by \Lem~\ref{Lem_size_of_W} we may assume that $|W_j\setminus W_{ji}|\leq\tilde O_k(k^{-2})n$.
By construction, the set $W_j\setminus W_{ji}$ is independent of the random bipartite subgraph
of $\cG(n,p,\sigma)$ consisting of the $V_i$-$V_j$-edges.
Hence, the number $e(V_i,W_j\setminus W_{ji})$ of edges between $V_i$ and $W_j\setminus W_{ji}$
has distribution $\Bin(|V_i\times(W_j\setminus W_{ji}),p)$.
Given the upper bound on $|W_j\setminus W_{ji}|$, the Chernoff bound thus implies that
	\begin{equation}\label{eqLemma_S1_3}
	\pr\brk{e(V_i,W_j \setminus W_{ji})\leq\tilde O_k(k^{-3})n^2p}\geq1-\exp(-\Omega(n)).
	\end{equation}
Further, by construction the number of $V_i$-$W_{ji}$-edges is bounded by $300|W_{ji}|$.
Since by \Lem~\ref{Lem_size_of_W} we may assume that $|W_{ji}|\leq n\tilde O_k(k^{-3})$,
(\ref{eqLemma_S1_3}) implies~(\ref{eqLemma_S1_2}).

Let us condition on the event $\cA$ that
	$b=e(V_i,V_j\setminus W_j)\geq\frac13k^{-2}n^2p$ and $r=e(V_i,W_j)\leq O_k(k^{-3})\leq n^2p$.
Let us think of the vertices in $V_i$ as bins, and of the $V_i$-$V_j\setminus V_j$ edges as balls that are tossed independently and uniformly into the bins.
More precisely, we think of the $V_i$-$V_j\setminus W_j$ edges as blue balls, and of the $V_i$-$W_j$-edges as red balls.
Let $\cX_{ij}$ be the number of bins $v\in V_i$ that receive at least one ball but that do not receive a blue ball.
Now, given that $v$ receives $l$ balls in total, the probability that all the balls it receives are red is equal to the probability that
a hypergeometric random variable with parameters $l,b,r$ takes the value $l$.
Therefore, summing over all $l\geq1$ and using our conditions on $b,r$, we see that
	$\pr\brk{v\in\cX_{ij}}\leq\tilde O_k(k^{-3}).$
Because $\sigma$ is balanced, we thus obtain
	\begin{equation}\label{eqLemma_S1_4}
	\Erw[|\cX_{ij}|\,|\,\cA]\leq \frac nk\cdot O_k(k^{-3}).
	\end{equation}
In fact, because the balls are tossed into the bins independently of each other, Azuma's inequality implies together with~(\ref{eqLemma_S1_4}) that
	\begin{equation}\label{eqLemma_S1_5}
	\pr[|\cX_{ij}|\leq\tilde O_k(k^{-4})n\,|\,\cA]\geq 1-\exp(-\Omega(n)).
	\end{equation}
Since $\pr[\cA]\geq1-\exp(-\Omega(n))$ by~(\ref{eqLemma_S1_1}) and~(\ref{eqLemma_S1_2}),
(\ref{eqLemma_S1_5}) yields that $\pr[|\cX_{ij}|\leq\tilde O_k(k^{-4})n]\geq 1-\exp(-\Omega(n))$.
Taking the union bound over $i,j$ completes the proof because
	$S_1\subset\cup_{i,j}\cX_{ij}$.
\end{proof}

Fact~\ref{Fact_1-free} implies together with \Lem~\ref{Lem_size_of_Z}, \Cor~\ref{Cor_expansion_of_Z}, \Lem~\ref{Lemma_S0} and \Lem~\ref{Lemma_S1}
the desired bound on the number of $1$-free vertices.
To bound the number of $2$-free variables, we need

\begin{lemma}\label{Lemma_2-free}
Let $i,j,l\in\brk k$ be distinct.
With probability at least $1-\exp(-\Omega(n))$ there are no more than $n\tilde O_k(k^{-5})$ vertices $v\in V_i$
such that $e(v,V_j)\leq100$ and $e(v,V_l)\leq100$.
\end{lemma}
\begin{proof}
For any $v$, $e(v,V_j)$, $e(v,V_l)$ are independent binomial variables.
Because $\sigma$ is balanced, their means are $(1+o(1))\frac nkp$.
Hence, (\ref{eqplantedp}) shows that $\pr\brk{e(v,V_j),e(v,V_l)\leq100}\leq\tilde O_k(k^{-4})$.
Consequently, the expected number of $v\in V_i$ with $e(v,V_j),e(v,V_l)\leq100$ is $nO_k(k^{-5})$.
In fact, this is a binomial random variable due to the independence of the edges in $\cG(n,p,\sigma)$.
Thus, the assertion follows from the Chernoff bound.
\end{proof}

Now, let $S_2$ be the set of all $v\in V_i$ such that there exist distinct $j,l\in[k]\setminus\cbc i$ such that
$e(v,V_j)\leq100$ and $e(v,V_l)\leq100$.
By construction, if $v$ is $2$-free, then $v\in S_2\cup Z\cup N(Z)$ (note that $U\subset Z$).
Thus, the desired bound on the number of $2$-free vertices follows  from
\Lem~\ref{Lem_size_of_Z}, \Cor~\ref{Cor_expansion_of_Z} and \Lem~\ref{Lemma_2-free}.
\qed


\begin{thebibliography}{29}


\bibitem{AchFried}
D.~Achlioptas, E.~Friedgut:
A sharp threshold for $k$-colorability.
Random Struct.\ Algorithms {\bf 14} (1999) 63--70.

\bibitem{Barriers}
D.~Achlioptas, A.~Coja-Oghlan:
Algorithmic barriers from phase transitions.
Proc.~49th FOCS (2008) 793--802.

\bibitem{AchMolloy}
D.\ Achlioptas, M.\ Molloy:
The analysis of a list-coloring algorithm on a random graph.
Proc.\ 38th FOCS (1997) 204--212.


\bibitem{AMoColor}
D.~Achlioptas, C.~Moore:
The chromatic number of random regular graphs.
Proc.~8th RANDOM (2004) 219--228

\bibitem{nae}
D.~Achlioptas, C.~Moore:
Random $k$-SAT: two moments suffice to cross a sharp threshold.
SIAM Journal on Computing {\bf 36} (2006) 740--762.

\bibitem{AchNaor}
D.~Achlioptas, A.~Naor:
The two possible values of the chromatic number of a random graph.
Annals of Mathematics {\bf 162} (2005), 1333--1349.

\bibitem{AK}
N.~Alon, N.~A. Kahale, A spectral technique for coloring random 3-colorable graphs.
SIAM J.\ Comput.~\textbf{26} (1997) 1733--1748.

\bibitem{AlonKriv}
N.~Alon, M.~Krivelevich: The concentration of the chromatic number of random graphs.
Combinatorica {\bf 17} (1997) 303--313


\bibitem{AppelHaken77}
K.~Appel, W.~Haken:
Every planar map is four colorable.
Illinois Journal of Mathematics {\bf 21} (1977) 429--567

\bibitem{Condensation}
V.~Bapst, A.~Coja-Oghlan, S.~Hetterich, F.~Ra\ss mann, D.~Vilenchik:
The condensation transition in random graph coloring.
arXiv:1404.5513 (2014).


\bibitem{BBColor}
B.~Bollob\'as: The chromatic number of random graphs.
Combinatorica {\bf8} (1988) 49--55

\bibitem{BB}
B.~Bollob\'as: Random graphs. 2nd edition.
Cambridge University Press (2001)

\bibitem{BollobasBCKW01}
B.~Bollob\'as, C.~Borgs, J.~Chayes, J.-H.~Kim, D.~Wilson:The scaling window of the 2-SAT transition. Random Struct.\ Algorithms {\bf 18} (2001) 201--256.



\bibitem{Covers}
A.~Coja-Oghlan: Upper-bounding the $k$-colorability threshold by counting covers.
Electronic Journal of Combinatorics {\bf 20} (2013) P32.

\bibitem{Regular}
A.~Coja-Oghlan, S.~Hetterich, C.~Efthymiou: On the chromatic number of random regular graphs.
arXiv:1308.4287 (2013).

\bibitem{KostaNAE}
A.~Coja-Oghlan, K.~Panagiotou:
Catching the $k$-NAESAT threshold.
Proc.\ 44th STOC (2012) 899--908.

\bibitem{KostaSAT}
A.~Coja-Oghlan, K.~Panagiotou:
Going after the $k$-SAT threshold.
Proc.\ 45th STOC (2013), to appear.

\bibitem{Lenka}
A.~Coja-Oghlan, L.~Zdeborov\'a:
The condensation transition in random hypergraph 2-coloring.
Proc.~23rd SODA (2012) 241--250.

\bibitem{Dani}
V.~Dani, C.~Moore, A.~Olson:
Tight bounds on the threshold for permuted $k$-colorability.
Proc.\ 16th  RANDOM (2012) 505--516.


\bibitem{DFG}
M.~Dyer, A.~Frieze, C.~Greenhill:
On the chromatic number of a random hypergraph.
Preprint (2012).

\bibitem{ER}
P.\ Erd\H os, A.\ R\'enyi: On the evolution of random graphs.
Magayar Tud.\ Akad.\ Mat.\ Kutato Int.\ Kozl.\ {\bf 5} (1960) 17--61.

\bibitem{FMcD}
A.~Frieze, C.~McDiarmid: Algorithmic theory of random graphs.
Random Struct.\ Algorithms {\bf 10} (1997) 5--42

\bibitem{FriezeWormald}
A.~Frieze, N.~Wormald:
Random $k$-Sat: a tight threshold for moderately growing $k$.
Combinatorica {\bf 25} (2005) 297--305.

\bibitem{Gilbert}
E.\ Gilbert:
Random graphs.
Annals Math.\ Statist.\ {\bf 30} (1959) 1141--1144.

\bibitem{GMcD}
G.~Grimmett, C.~McDiarmid: On colouring random graphs.
\MPCPS\ {\bf 77} (1975) 313--324

\bibitem{JLR}
S.~Janson, T.~{\L}uczak, A.~Ruci\'nski: Random Graphs, Wiley  2000.

\bibitem{KPGW}
G.~Kemkes, X.~P\'erez-Gim\'enez, N.~Wormald:
On the chromatic number of random $d$-regular graphs.
Advances in Mathematics {\bf 223}  (2010) 300--328.

\bibitem{Kriv2}
M.~Krivelevich: Coloring random graphs -- an algorithmic perspective.
Proc.\ 2nd Colloquium on Mathematics and Computer Science (2002) 175-195.

\bibitem{KSud}
M.~Krivelevich, B.~Sudakov: Coloring random graphs.
\IPL\ {\bf 67} (1998) 71--74


\bibitem{pnas}
F.~Krzakala, A.~Montanari, F.~Ricci-Tersenghi, G.~Semerjian, L.~Zdeborova:
Gibbs states and the set of solutions of random constraint satisfaction problems.
Proc.~National Academy of Sciences {\bf104} (2007) 10318--10323.

\bibitem{KPW} 
F.\ Krzakala, A.\ Pagnani, M.\ Weigt:
Threshold values, stability analysis and high-$q$ asymptotics for the coloring problem on random graphs.
Phys.\ Rev.\ E {\bf70} (2004) 046705.

\bibitem{LuczakColor} 
T.~{\L}uczak: The chromatic number of random graphs.
\COMB\ {\bf11} (1991) 45--54

\bibitem{Luczak}
T.~{\L}uczak: A note on the sharp concentration of the chromatic number of random graphs.
\COMB\ {\bf 11} (1991) 295--297

\bibitem{Matula}
D.\ Matula: Expose-and-merge exploration and the chromatic number of a random graph.
Combinatorica {\bf 7} (1987) 275--284.

\bibitem{MM}
M.~M\'ezard, A.~Montanari:
Information, physics and computation.
Oxford University Press~2009.

\bibitem{MPZ}
M.~M\'ezard, G.~Parisi, R.~Zecchina:
Analytic and algorithmic solution of random satisfiability problems.
Science {\bf 297} (2002) 812--815.

\bibitem{Molloy}
M.~Molloy: The freezing threshold for $k$-colourings of a random graph.
Proc.\ 43rd STOC (2012) 921--930.


\bibitem{MPWZ}
R.~Mulet, A.~Pagnani, M.~Weigt, R.~Zecchina:
Coloring random graphs.
Phys.\ Rev.\ Lett.\ {\bf 89} (2002) 268701

\bibitem{vanMourik} 
J.\ van Mourik, D.\ Saad:
Random Graph Coloring - a Statistical Physics Approach.
Phys.\ Rev.\ E {\bf66} (2002) 056120.


\bibitem{Robertson}
N.~Robertson, D.\ Sanders, P.\ Seymour, R.\ Thomas: The four-colour theorem. J.\ Combin.\ Theory Ser.\ B {\bf 70} (1997) 2--44.


\bibitem{ShamirSpencer}
E.~Shamir, J.~Spencer: Sharp concentration of the chromatic number of random graphs $\gnp$.
\COMB\ {\bf 7} (1987) 121--129


\bibitem{LenkaFlorent} 
L.~Zdeborov\'a, F.~Krzakala: Phase transition in the coloring of random graphs.
Phys.\ Rev.\ E {\bf76} (2007) 031131.

\end{thebibliography}
\end{document}